\documentclass{article}
\usepackage{graphicx} 
\usepackage[dvipsnames]{xcolor}
\usepackage[colorlinks=true,linkcolor=blue,citecolor=blue,urlcolor=blue,bookmarks=true]{hyperref}
\usepackage{schmidhuber}

\newcommand{\sym}{\mathrm{symp}}
\newcommand{\Spec}{\mathrm{Spec}}
\renewcommand{\th}{^{\mathrm{th}}}
\usepackage{bm}
\usepackage{authblk}
\usepackage{booktabs}
\usepackage[normalem]{ulem}

\usepackage[most]{tcolorbox}

\usepackage{float}          

\usepackage{algorithm}      
\usepackage{algpseudocode}  
\algrenewcommand\algorithmicrequire{\textbf{Input:}}
\algrenewcommand\algorithmicensure {\textbf{Output:}}

\DeclareMathOperator{\Bin}{\mathsf{Bin}}

\title{Hamiltonian Decoded Quantum Interferometry \vspace{1mm}}

\author[1,2]{\quad\quad\, Alexander Schmidhuber\thanks{Co-first authors}}
\author[2]{Jonathan Z. Lu\protect\footnotemark[1]}
\author[1]{Noah Shutty}
\author[1]{\newline Stephen Jordan}
\author[2,3]{Alexander Poremba\thanks{Joint supervisors}}
\author[2,4]{Yihui Quek\protect\footnotemark[2]}

\affil[1]{Google Quantum AI, Venice, CA}
\affil[2]{Massachusetts Institute of Technology, Cambridge, MA}
\affil[3]{Boston University, Boston, MA} 
\affil[4]{École Polytechnique Fédérale de Lausanne, Lausanne, Switzerland}

\date{}

\begin{document}
\maketitle
{
\renewcommand\thefootnote{}\footnotetext{%
\small
\textbf{Correspondence:} alexsc@mit.edu, lujz@mit.edu,
shutty@google.com,  \\ \hspace*{3.3cm} stephenjordan@google.com, 
poremba@bu.edu, yihui.quek@epfl.ch.}
}

\begin{abstract}
We introduce Hamiltonian Decoded Quantum Interferometry (HDQI), a quantum algorithm that utilizes coherent Bell measurements and the symplectic representation of the Pauli group to reduce Gibbs sampling and Hamiltonian optimization to classical decoding. For a signed Pauli Hamiltonian $H$ and any degree-$\ell$ polynomial $\calP$, HDQI prepares a purification of the density matrix $$\rho_\calP(H) = \calP^2(H)/\Tr[\calP^2(H)]$$ by solving a combination of two tasks: decoding $\ell$ errors on a classical code defined by $H$, and preparing a \emph{pilot state} that encodes the anti-commutation structure of $H$.
Choosing $\calP(x)$ to approximate $\exp(-\beta x/2)$ yields Gibbs states at inverse temperature $\beta$; other choices of $\calP$ prepare approximate ground states, microcanonical ensembles, and other spectral filters.  

The decoding problem inherits structural properties of $H$; in particular, local Hamiltonians map to LDPC codes. Preparing the pilot state is always efficient for commuting Hamiltonians, but highly non-trivial for non-commuting Hamiltonians. Nevertheless, we prove that this state admits an efficient matrix product state representation for Pauli Hamiltonians whose anti-commutation graph decomposes into connected components of logarithmic size.

HDQI efficiently prepares Gibbs states at arbitrary temperatures for a class of physically motivated commuting Hamiltonians -- including the toric code, color code, and Haah's cubic code -- but we also develop a matching efficient classical algorithm for this task, thereby delineating the boundary of efficient classical simulation. For a non-commuting semiclassical spin glass and commuting stabilizer Hamiltonians with quantum defects, HDQI provably prepares Gibbs states up to a constant inverse-temperature threshold using polynomial quantum resources and quasi-polynomial classical pre-processing. These results position HDQI as a versatile new algorithmic primitive and the first extension of Regev’s reduction to non-abelian groups.

\end{abstract}

\newpage
\tableofcontents 


\newpage

\section{Introduction}

One of the earliest envisioned applications of quantum computers is the prediction of equilibrium properties of quantum-mechanical Hamiltonians \cite{feynman1982simulating,lloyd1996universal}. 
Algorithmically, this goal reduces to two central tasks. 
The first is (approximate) \emph{ground state preparation}, also known as \emph{Hamiltonian optimization} for combinatorially motivated Hamiltonians such as Quantum Max-Cut. 
The second is the task of \emph{sampling} from physically or statistically motivated distributions over a Hamiltonian, most notably Gibbs states at finite temperature. 

On the optimization side, there has been a long line of work on developing suitable algorithmic primitives; these include adiabatic/annealing paradigms~\cite{farhi2001adiabatic,kadowaki1998quantum}, variational and hybrid approaches~\cite{peruzzo2014variational,Cerezo_2021}, quantum phase estimation methods~\cite{Abrams_1999,kitaev1995quantummeasurementsabelianstabilizer}, filtering-based approaches~\cite{Poulin_2009,Lin_2020} and dissipative methods~\cite{cubitt2023dissipativegroundstatepreparation}. The theory of $\mathsf{QMA}$-completeness~\cite{kitaev2002local,bookatz2012qma}
complements these predominantly heuristic algorithms with fundamental worst-case hardness barriers, but leaves open the possibility of quantum speed-ups on typical instances or structured families of quantum optimization 
problems arising in physics, chemistry and materials science. 

On the sampling side, preparing Gibbs states is an essential primitive across physics and chemistry.
Gibbs ensembles provide the foundation for computing macroscopic observables from microscopic models, enabling the estimation of free energies, partition-function proxies, and thermal expectation values~\cite{PathriaBeale2011,Chandler1987,Callen1985,LandauLifshitz1980,Sachdev2011,AltlandSimons2010,Mermin1965,Tuckerman2010,LandauBinder2014}. 
Classical methods for Gibbs sampling (e.g., \emph{Markov Chain Monte Carlo} (MCMC)~\cite{speagle2020conceptualintroductionmarkovchain}) often suffer from long mixing times or sign problems, especially in frustrated systems or quantum many-body contexts.
Existing \textit{quantum} algorithms for thermal state preparation, such as quantum Metropolis sampling~\cite{Temme_2011} and related Gibbs-sampling methods~\cite{GilyenThermal23a,chen2023efficient}, achieve speedups in principle \cite{BCL24,rajakumar2024gibbs}, but exhibit runtimes that depend heavily on certain mixing properties or temperature-dependent parameters, which are often difficult to determine in practice. 

A recent work introduced \emph{Decoded Quantum Interferometry} (DQI)~\cite{JSW24}, a novel algorithmic primitive which naturally interpolates between sampling and approximate optimization; DQI reduces both these problems to a \emph{quantum} state preparation task, which then itself reduces to a \emph{classical} syndrome decoding problem. 
This algorithmic paradigm builds on Regev's quantum reduction~\cite{10.1145/1060590.1060603,cryptoeprint:2009/285}, which traces back to seminal work of Aharonov and Ta-Shma~\cite{aharonov2003adiabaticquantumstategeneration} and which also appeared in the more recent breakthrough of Yamakawa and Zhandry~\cite{10.1145/3658665}.
At a high level, DQI finds approximate optima of classical objective functions $f$ over finite fields, such as MAX-XOR-SAT over $\mathbb{F}_2^n$, by preparing a state corresponding to a superposition of a degree-$\ell$ polynomial $\mathcal{P}$ evaluated at $f$:
\begin{align}\label{eq:DQI-state}
\ket{\mathcal{P}(f)} \propto \sum_{\mathbf{x} \in \mathbb{F}_2^n} \mathcal{P}(f(\mathbf{x})) \ket{\mathbf{x}}.
\end{align}
Measurement of $\ket{\calP(f)}$ in the computational basis yields a distribution wherein a string $\mathbf{x}$ is obtained with probability proportional to $\calP(f(\mathbf{x}))^2$; choosing $\calP$ appropriately enhances the likelihood of finding a string with a large objective value. 
The key conceptual insight of DQI is that preparing the state in Eqn.~(\ref{eq:DQI-state}) reduces to a syndrome decoding problem for a classical code defined by the structure of the objective function $f$.

While DQI offers exciting new avenues for potential speedups in optimization and sampling, these appear limited to classical objective functions.
Indeed, an alternative description of the DQI algorithm (over $\F_2$) is that it samples from the density matrix $\rho_\calP (H) \propto {\calP^2(H)}$ for a classical, diagonal Hamiltonian $H$ encoding the cost function $f$ via $f(\mathbf{x}) = \langle \mathbf{x} | H | \mathbf{x} \rangle$ for $\mathbf{x} \in \mathbb{F}_2^n$. 
Such Hamiltonians consist solely of $Z$-type Pauli interactions.
General non-diagonal Pauli Hamiltonians capture a much richer landscape of tasks relevant to physics and chemistry, as well as Hamiltonian complexity theory~\cite{cubitt2016complexityclassificationlocalhamiltonian}.
For example, general Pauli Hamiltonians include the transverse‐field Ising model~\cite{bravyi2014complexityquantumisingmodel}, Quantum Max-Cut \cite{gharibian2019almost,anshu2020beyond}, Kitaev's toric code Hamiltonian~\cite{kitaev2003anyons} and other \emph{stabilizer}/\emph{commuting} Hamiltonians~\cite{BV05}, as well as quantum spin glasses and ``spin'' versions of the SYK model~\cite{Hanada_2024}. 
The broad interest in such general Hamiltonians naturally motivates the central guiding question of our work.
\begin{tcolorbox}[colback=gray!10,colframe=black!50,title=Guiding question]
Can the decoding-based framework of \emph{Decoded Quantum Interferometry} be extended from diagonal classical Hamiltonians to general non-commuting Pauli Hamiltonians, thereby reducing inherently quantum tasks such as Gibbs sampling and Hamiltonian optimization to classical decoding?
\end{tcolorbox}

\paragraph{Contributions.}
In this work, we answer this question in the affirmative by introducing Hamiltonian Decoded Quantum Interferometry (HDQI). 
HDQI is a strict and powerful generalization of DQI, and uses coherent Bell measurements and the symplectic representation of the Pauli group to reduce Gibbs sampling and Hamiltonian optimization to classical decoding and quantum state preparation.
We illustrate HDQI for Hamiltonians of the form\footnote{We discuss more general Hamiltonians in Section~\ref{sec:beyond_paulis}.} 
\begin{equation}\label{eq:H_hDQI}
    H = \sum_{i=1}^m v_i P_i,
\end{equation}
where $v_i \in \set{\pm 1}$ and the $P_i$ are arbitrary $n$-qubit Pauli operators.
Note that the diagonal Hamiltonians optimizable by DQI (over $\F_2$) are a special case of Eqn.~(\ref{eq:H_hDQI}). 
HDQI prepares a purification of the state
\begin{equation}
\label{eq:HDQI_state_intro}
    \rho_\calP (H) = \frac{\calP^2(H)}{\Tr \left[\calP^2(H)\right]},
\end{equation}
where $\calP$ is an arbitrary degree-$\ell$ polynomial. Choosing $\calP(x)$ to approximate $\exp(-\beta x/2)$ yields Gibbs states at inverse temperature $\beta$; other choices of $\calP$ prepare approximate ground states, microcanonical ensembles, and other spectral filters. If $H$ is diagonal, the state in Eqn.~(\ref{eq:HDQI_state_intro}) mimics the output distribution of DQI.

Following the framework of DQI and Regev's reduction, HDQI prepares the density matrix in Eqn.~(\ref{eq:HDQI_state_intro}) by solving a classical syndrome decoding problem. 
If the Hamiltonian is non-commuting, a further difficulty arises: the preparation of a \emph{pilot} state that captures the non-abelian expansion of $\calP(H)$. 
By extending DQI to the Pauli group, our work represents the first investigation of Regev’s reduction in the non-abelian setting.
The algorithmic procedure of HDQI is summarized in Section~\ref{sec:overview}; 
in short, it consists of the following three steps:
\begin{enumerate}
    \item \textbf{Pilot state preparation.} Prepare an $m$-qubit state which we term a ``pilot" state, which carries amplitudes of the non-commutative expansion of $\calP(H)$ into elementary symmetric polynomials. 
    For commuting Hamiltonians this pilot state is a weighted superposition of Dicke states, as in \cite{JSW24}. 
    Generally, however, preparing the pilot state is highly non-trivial. 
    
    \item \textbf{Controlled operations on Bell pairs.} 
    In a second register, prepare a maximally entangled state across two $n$-qubit systems. 
    Controlled on the pilot state, apply Pauli operations $P_i$ to the first half of the maximally entangled state, creating a superposition over both registers. 
    
    \item \textbf{Uncomputation by decoding.} To prepare $\rho_\calP(H)$ on the second register, it remains to uncompute the pilot state register. 
    We show that this uncomputation step corresponds to coherently decoding a classical error-correcting code associated to $H$, which we call the {\em symplectic code} of $H$.
    If we wish to prepare $\rho_{\calP}$ where $\calP$ is a degree-$\ell$ polynomial, then we must decode errors of weight up to $\ell$.
\end{enumerate}
Step 2 is always guaranteed to be efficient in our setting where $m=\poly(n)$; the complexity of HDQI arises from Steps 1 and 3. 
In other words, HDQI is efficient for a polynomial $\calP$ whenever (a) we can efficiently prepare the pilot state corresponding to $\calP(H)$, and (b) there exists an efficient decoder for the symplectic code of $H$ in the sense stated in Step 3. 

For commuting Hamiltonians, the pilot state can always be prepared efficiently, and so the runtime bottleneck is the decoding step. 
We identify a ``nearly independent" structure---roughly, Hamiltonians whose Pauli terms have only constantly many non-trivial relations among them---which makes this decoding step (and hence the full HDQI procedure) efficient. 
This class encompasses essentially all 2D translation-invariant Pauli stabilizer Hamiltonians, the color code, cluster Hamiltonians, and 3D Hamiltonians like Haah's code. 
For these systems, HDQI efficiently prepares Gibbs states at arbitrary temperature, as well as arbitrary microcanonical ensembles and other spectral filters. 
Gibbs sampling for such Hamiltonians is a topic of active research \cite{LinLinThermal24, PSS25}.
Notably, \cite{PSS25} conjectures efficient quantum preparation via dualities to decoupled Ising models. 
HDQI bypasses such dualities entirely, providing a direct and provably efficient eigenstate and Gibbs sampler.

We complement our quantum algorithm with an improved classical result.
Surprisingly, the same nearly independent structure also permits a classical procedure for sampling eigenstates according to an arbitrary spectral weighting function $\mathcal{P}$. 
More precisely, while HDQI prepares the purification of $\rho_{\mathcal P}$---a coherent superposition over eigenstates---the classical algorithm outputs stabilizer descriptions of eigenstates $\ket{\lambda}$ distributed proportionally to $\mathcal{P}(\lambda)^2$. 
Despite its conceptual simplicity, this appears to be a previously unrecognized classical algorithm for sampling ground or thermal stabilizer states.

Our third main contribution relates to HDQI for non-commuting Hamiltonians, for which we do not expect the pilot state preparation step to be efficient in general. 
We analyze the pilot state using two combinatorial tools: the anticommutation graph of $H$ and a derived quantity we call the antisymmetry character. 
The latter factorizes over connected components of the graph, inducing a tensorial structure that allows the pilot state to be expressed as a matrix product state (MPS) with linear bond dimension. 
In the MPS representation, each site is a qudit of dimension $q$, where $q \leq 2^d$ and $d$ is the size of the largest connected component of the anticommutation graph.
Hence, if all components have size $O(\log n)$, HDQI runs in polynomial quantum time after an $O(n^{d})$ classical pre-processing step required to explicitly compute the MPS matrices.

As an application, we consider ``defected'' Hamiltonians obtained by perturbing a commuting model with a small fraction of non-commuting terms.
More precisely, we study random local $Z$-type Hamiltonians (classical spin glasses) where each $Z$ Pauli is replaced by an $X$ Pauli with probability $p$. 
We prove that if the defect probability $p$ is below a constant threshold, the anticommutation graph has small connected components with high probability. 
Consequently, HDQI efficiently Gibbs-samples such semiclassical spin-glass models up to constant temperature. 
We also show that nearly independent Hamiltonians whose symplectic code has high distance, including the toric code and Haah's cubic code, remain Gibbs-sampleable even if every term is defected with small constant probability. 
As stated, these samplers all require a quasipolynomial classical pre-processing step. 
The pre-processing cost can, however, also be reduced to polynomial if the success probability of the sampler is allowed to be a constant arbitrarily close to $1$ as opposed to $1 - o(1)$.
\begin{figure}[t!]
    \centering
    \includegraphics[width=\linewidth]{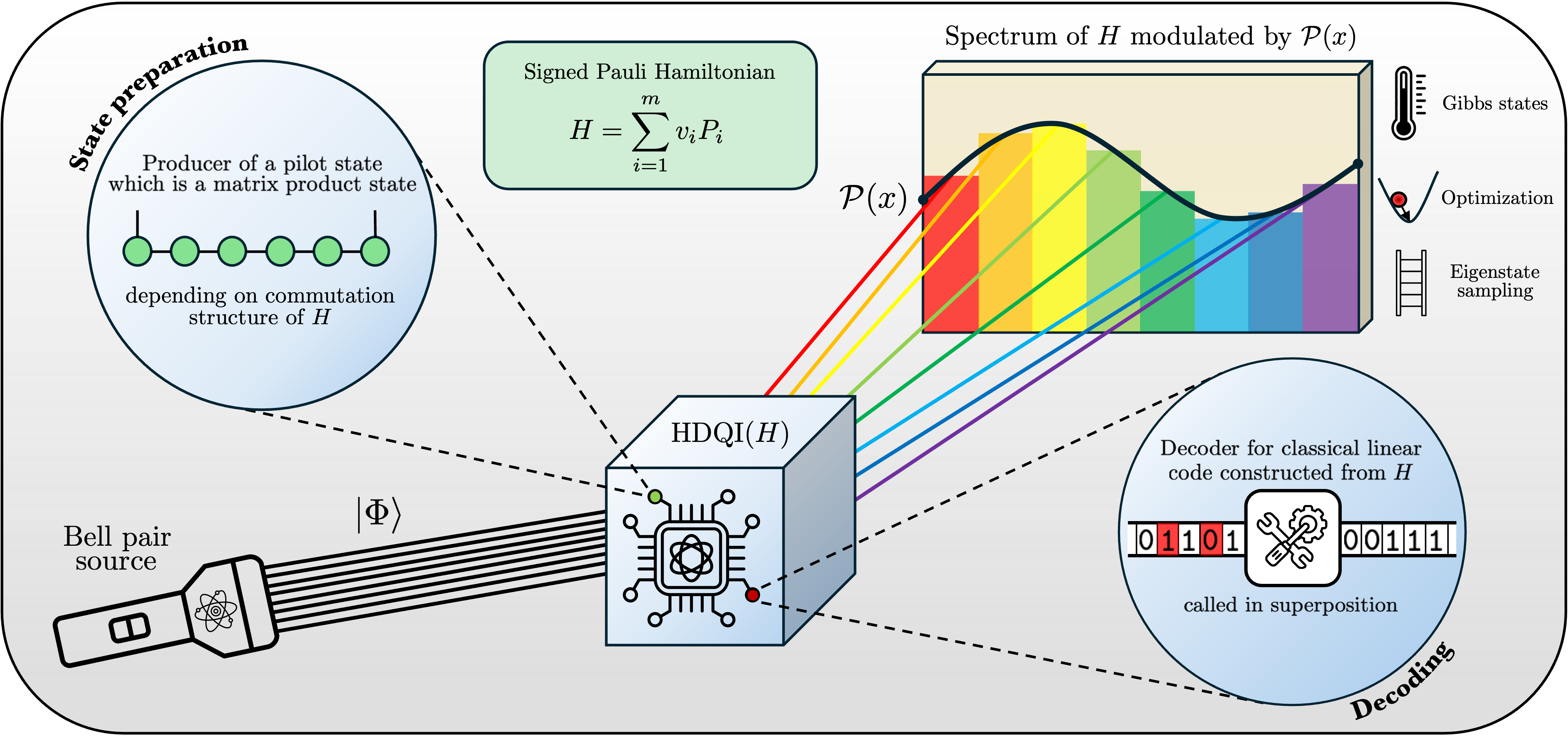}
    \caption{%
    Schematic illustration of Hamiltonian Decoded Quantum Interferometry (HDQI), a quantum algorithm for modulating the spectral distribution of a given Hamiltonian. 
    For a Pauli Hamiltonian $H = \sum_{i=1}^m v_i P_i$, HDQI transforms a source of Bell pairs $\ket{\Phi}$ into a state which is a superposition over the eigenstates $\ket{\lambda}$ of $H$, with amplitudes modulated by a high-degree polynomial $\calP(\lambda)$. HDQI implements this transformation by decoding, in superposition, errors on a certain classical linear code $\Symp(H)$ constructed from $H$.
    The output state of $\operatorname{HDQI}(H)$ has probability $\propto \calP^2(\lambda)$ of being $\ket{\lambda}$.
    In this sense HDQI acts as an algorithmic spectral filter, splitting a uniform incoming source of Bell pairs into a certain modulated spectrum of energies.
    This polynomial can be chosen to filter for desired eigenvectors.
    For example, high-degree monomials concentrate amplitudes toward high-energy eigenstates, and polynomial approximations of exponential functions yield approximate Gibbs states.}
    \label{fig:main_figure}
\end{figure}
\section{Technical overview}
\label{sec:overview}
Before we discuss our results in full detail, we provide in this section an overview of our quantum algorithm and its main applications.

\subsection{Polynomial Hamiltonian states}
Consider a Hamiltonian on $n$ qubits, \begin{equation}
\label{eq:HDQI_ham}
   H =  \sum_{i = 1}^m v_i P_i, 
\end{equation} 
where $m = \poly(n)$, each $P_i \in \{\Id,X,Y,Z\}^{\otimes n}$ is a $n$-qubit Pauli operator, and each $v_i \in \{\pm1\}$ is a phase.
Denote the eigenvalues and eigenvectors of $H$ by \begin{equation}
    \{\lambda \in \mathbb{R}, \ket{\lambda} \in (\mathbb{C}^{2})^{\otimes n}\}_\lambda, \qquad \text{ such that } \qquad H  = \sum_\lambda \lambda \ketbra{\lambda}{\lambda
    }.
\end{equation} 
Our goal in this paper is to prepare mixed\footnote{In reality, our quantum algorithm naturally prepares a minimal purification of these density matrices, providing additional applications.} quantum states that take the form of polynomials of Hamiltonians:
\begin{equation}
\label{eq:HDQI_state}
    \rho_\calP (H) = \frac{\calP^2(H)}{\Tr \left[\calP^2(H)\right]},
\end{equation}
where $\calP$ is an arbitrary degree-$\ell$ polynomial and $H$ is of the form in Eqn.~(\ref{eq:HDQI_ham}). 
Here, $\calP(H)$ refers to the linear combination of Paulis obtained by treating $\calP$ as a formal polynomial in the variable $H$. 
Before explaining how HDQI prepares the state in Eqn.~(\ref{eq:HDQI_state}), we summarize some immediate applications of $\rho_\calP(H)$.
\begin{itemize}
    \item \textbf{Optimization.} Measuring the density matrix in Eqn.~(\ref{eq:HDQI_state}) in the eigenbasis $\{\ket{\lambda}\}_\lambda$ of the Hamiltonian $H = \sum_\lambda \lambda \ketbra{\lambda}{\lambda}$ (e.g., via \emph{Quantum Phase Estimation}~\cite{kitaev1995quantum}) yields an eigenstate $\ket{\lambda}$ with probability $\sim \calP(\lambda)^2$. For appropriate polynomials, such as $\calP(x) = x^\ell$, this distribution is heavily biased towards eigenstates with large energy. 
    Hamiltonian DQI therefore naturally solves a \emph{quantum optimization problem}, where the objective function corresponds to the energy of $H$. 

    As mentioned before, if the Hamiltonian $H$ is diagonal, its eigenvectors are bitstrings in $\F_2^n$ and the corresponding optimization problem is classical. 
    In this case, the state in Eqn.~(\ref{eq:HDQI_state}) is precisely the classical probability distribution prepared by Decoded Quantum Interferometry (DQI)~\cite{JSW24}. 
    As such, Hamiltonian DQI includes DQI as a special case.

    \item \textbf{Gibbs sampling.} If the polynomial $\calP(x)$ is instead chosen to be the degree-$\ell$ Chebyshev approximation of $\exp(-\beta x/2 )$, the state in Eqn.~(\ref{eq:HDQI_state}) is an approximate \emph{Gibbs state} of the Hamiltonian $H$ at (inverse) temperature $\beta$. 
    We show that HDQI efficiently prepares Gibbs states up to an inverse temperature that scales linearly with the number of efficiently decodable errors on the symplectic code of $H$. 
    HDQI thus presents a new approach for Gibbs sampling, by mapping this problem to a decoding task.  
    
    More generally, HDQI prepares a purification of the Gibbs state known as the \emph{thermofield double state}, which appears in condensed matter theory \cite{umezawa1982thermo} and high energy physics \cite{maldacena2003eternal}.
    This state is strictly more informative than the Gibbs state itself, enabling, for example, quadratically faster detection of quantum phase transitions and a more efficient estimation of the partition function when combined with \emph{Amplitude Amplification}~\cite{brassard2000quantum}. 
    The preparation of thermofield double states has been a topic of recent interest \cite{zhu2020generation}
    
    \item \textbf{Microcanonical ensembles and spectral filters.} For certain Hamiltonians we refer to as \emph{nearly independent}, HDQI efficiently prepares \emph{arbitrary} non-negative normalized functions of $H$. 
    This class includes abelian Kitaev quantum double models such as the 2D toric code, as well as the 3D Haah code and other, non-local models. 
    For nearly independent Hamiltonians, HDQI efficiently prepares Gibbs states at \textit{any} temperature 0 $\leq \beta \leq \infty$, as well as any exact microcanonical ensembles, band-pass filters, and sector-projected ensembles.
\end{itemize} 
To prepare this density matrix, we coherently decode $\ell$ errors on a certain code defined by the structure of $H$, which we call the \emph{symplectic code} of $H$.

\subsection{Symplectic codes}

For $H =  \sum_{i = 1}^m v_i P_i$, rewrite each pauli $P_i$ in its \emph{symplectic representation} $\sym(P_i) \in \F_2^{2n}$, which is a bitstring whose first $n$ bits record the $Z$-support and the last $n$ bits record the $X$-support of $P_i$ (see Definition~\ref{def:symplectic_vectors} for details). 
This mapping is a group homomorphism:
\begin{equation}\label{eq:homomorphism}
    \sym(P_i P_j) = \sym(P_i) \oplus \sym(P_j) \in \F_2^{2n} ,
\end{equation}
where $\oplus$ denotes addition mod $2$.
The symplectic code of $H$, denoted $\Symp(H)$, is a classical linear code whose parity check matrix is given by binary matrix $B^\intercal \in \F_2^{2n \times m}$ such that \begin{align}
        B^\intercal = \begin{bmatrix}
        | & | &        & | \\
        \sym(P_1) & \sym(P_2) & \cdots & \sym(P_m) \\
        | & | &        & | 
        \end{bmatrix} .
    \end{align}
A precise definition is given in Definition~\ref{def:symplectic_codes}. 
Note that if the Hamiltonian $H$ is classical (diagonal), then only the first $n$ rows of $B^\intercal$ contain non-zero elements. 
We may thus ignore the last $n$ rows for classical Hamiltonians, in which case the resulting parity check matrix in $\F_2^{n \times m}$ precisely equals the parity check matrix studied in Decoded Quantum Interferometry \cite{JSW24}.

\subsection{Hamiltonian DQI}

Our main algorithmic result is a reduction from preparing $\rho_\calP (H)$ in Eqn.~(\ref{eq:HDQI_state}) to decoding the classical code $\Symp(H)$ combined with a state preparation task, which we now sketch.

\subsubsection{Maximally entangled state}\label{subsec:initial-state-prep}
The perhaps simplest approach to arrive at the distribution in Eqn.~(\eqref{eq:HDQI_state}) would be to prepare a ``q-sample'' given by the quantum state \begin{equation}
   \propto \calP(H)\sum_\lambda \ket{\lambda} = \sum_{\lambda}\calP(\lambda)\ket{\lambda},
\end{equation} 
and then measure in the eigenbasis of $H$. 
If $H$ is diagonal and hence characterized by a classical function $f(\bx) = \langle \mathbf{x} |H|  \mathbf{x} \rangle$, this state simplifies to \begin{equation}
     \propto \sum_{\mathbf{x} \in \mathbb{F}_2^n} \calP(f(\mathbf{x})) \ket{\mathbf{x}},
\end{equation} 
which is precisely the quantum state prepared by Decoded Quantum Interferometry \cite{JSW24}. 
In the more general setting of quantum Hamiltonians we consider here, two immediate issues arise. 
First, preparing the uniform superposition $\propto \sum_\lambda \ket{\lambda}$ is not efficient in general, without explicitly knowing the eigenbasis of $H$. 
Second, no state on $n$ qubits can be used to uniquely distinguish all Pauli errors (as is needed in the decoding procedure), since the Hilbert space has dimension $2^n$, whereas there are $4^n$ distinct Pauli operators. 
Both issues can be overcome simultaneously by replacing $\sum_\lambda \ket{\lambda}$ with the maximally entangled state 
\begin{equation}
    \ket{\Phi^n} := \left(\frac{1}{\sqrt{2}} (\ket{00} + \ket{11})\right)^{\otimes n} = \frac{1}{2^{n/2}}\sum_{\mathbf{x} \in \F_2^n} \ket{\mathbf{x}} \otimes \ket{\mathbf{x}} 
\end{equation} 
across two $n$-qubit systems. 
This state has the convenient property that it is maximally entangled in any basis. 
Indeed, for the eigenbasis $\{\ket{\lambda}\}_\lambda$ of $H$, it has the representation 
\begin{equation}
    \ket{\Phi^n} = \frac{1}{2^{n/2}}\sum_\lambda \ket{\lambda}\otimes \ket{\bar{\lambda}} 
\end{equation} 
where $\ket{\bar{\lambda}}$ is defined by $\bra{\mathbf{x}} \ket{\bar{\lambda}} = \overline{\bra{\mathbf{x}} \ket{\lambda}}$ and $\overline{z}$ denotes the complex conjugate of $z \in \C$. 
The first register of $\ket{\Phi^n}$ thus takes the desired form $\propto \sum_\lambda \ket{\lambda}$. 
We proceed by applying \begin{equation}\label{eq:preparethis_2}
    (\calP(H) \otimes \Id)  \ket{\Phi^n} = \frac{1}{2^{n/2}}\sum_{\lambda} \calP(\lambda) \ket{\lambda}\otimes \ket{\bar{\lambda}},
\end{equation} 
which, after tracing out the second register, yields the density matrix in Eqn.~(\ref{eq:HDQI_state}). 
In the above schema, the main challenge is to implement $\calP(H) \otimes \Id$ on the maximally-entangled state.
This is not a unitary operation in general. 
Nevertheless, following the philosophy of \cite{JSW24}, we may perform this operation by appending registers and applying classical decoding algorithms to uncompute those registers, as we next discuss.

\subsubsection{Implementing \texorpdfstring{$\calP(H)$}{P(H)} via decoding}\label{subsec:decoding-step}

For simplicity, let us assume in this overview that $H = \sum_{i=1}^m v_i P_i$ is a \emph{commuting} Hamiltonian such that $[P_i,P_j] = 0$. We discuss how to adapt the quantum algorithm in the non-commuting case in Section~\ref{subsec:intro-non-commuting}. The degree-$\ell$ polynomial $\calP(H)$ is a sum of signed products of Pauli terms of the form \begin{equation}
    P_{\mathbf{y}} := \prod_{i \in \supp(\mathbf{y})} P_i , 
\end{equation} 
where $\mathbf y \in \F_2^n$ with $|\mathbf y| \leq \ell$ labels which Paulis appear in $P_\mathbf{y}$. In the commuting case, the ordering of the Pauli operators in $P_{\mathbf y}$ does not matter, and $\calP(H)$ takes the simple form 
\begin{equation}
    \calP(H) = \sum_{k = 0}^\ell w_k \sum_{\mathbf y\in \F_2^m, |\mathbf y| = k} \left(\prod_{i\in \supp(\mathbf y)} v_i\right)P_{\mathbf y},
\end{equation} 
where $w_k$ are coefficients given by the expansion of $\calP$ into elementary symmetric polynomials. 
We show in Appendix~\ref{app:combinatorics} that these coefficients are efficiently computable. 
Building on this symmetric expansion, Hamiltonian DQI applies $\calP(H) \otimes \Id$ in three steps: 

\paragraph{Step 1:} Prepare the pilot state, which is a  weighted superposition of Dicke states:\footnote{We consider more general pilot states---beyond Dicke states---in Section~\ref{subsec:intro-non-commuting} and Section~\ref{sec:HDQIperfectlyDecodable}.} 
\begin{equation}
\label{eq:initial_state}
   \sum_{k = 0}^\ell w_k \sum_{\mathbf y \in \F_2^m ,|\mathbf y| = k}  \ket{\mathbf y} .
\end{equation} 
Tensor it with a maximally entangled state $\ket{\Phi^n}$.

\paragraph{Step 2:} Apply controlled Pauli and controlled phase operators on the first half of the maximally entangled state $\ket{\Phi^n}$, controlled on the pilot state, to obtain the overall state
\begin{equation}\label{eq:state_2}
   \sum_{k = 0}^\ell w_k \sum_{\mathbf y \in \F_2^m ,|\mathbf y| = k}  \ket{\mathbf y} \otimes \left(\prod_{i\in \supp(\mathbf y)} v_i\right) (P_{\mathbf y} \otimes \Id) \ket{\Phi^n}.
\end{equation} 

\paragraph{Step 3:} Uncompute the $\ket{\mathbf y}$ register given $(P_{\mathbf y} \otimes \Id) \ket{\Phi^n}$. That is, implement a unitary mapping \begin{equation}
    \calD_H^{(\ell)} \,:\, \ket{\mathbf y} \otimes (P_{\mathbf y} \otimes \Id) \ket{\Phi^n} \longmapsto \ket{0^m} \otimes (P_{\mathbf{y}} \otimes \Id) \ket{\Phi^n} \qquad \left( \text{for any } \mathbf{y} \in \F_2^m \text{ such that } |\mathbf{y}| \leq \ell\right). 
\end{equation} 
The resulting state, after discarding the $\ket{0^m}$ register, is \begin{equation}
\label{eq:purified_hdqi_sec2}
    \sum_{k = 0}^\ell w_k \sum_{\mathbf y\in \F_2^m, |\mathbf y| = k} \left(\prod_{i\in \supp(\mathbf y)} v_i\right) (P_{\mathbf{y}} \otimes \Id) \ket{\Phi^n} = (\calP(H) \otimes \Id)  \ket{\Phi^n},
\end{equation} as desired. 
Optionally, trace out the second register of $\ket{\Phi^n}$ to obtain Eqn.~(\ref{eq:HDQI_state}).

The main difficulty, in the case of commuting Hamiltonians, lies in the last step, which is a uncomputation or decoding step. 
For an arbitrary Pauli $P$, it is always possible to efficiently and unambiguously determine $P$ (up to a global phase) given $(P\otimes \Id) \ket{\Phi^n}$ with a single measurement. Such a measurement must exist information-theoretically, since for any two phase-free Paulis $P, Q$ we have the trace identity
\begin{equation}
    \langle \Phi^n |(Q^{\dagger}\otimes \Id)(P\otimes \Id) \ket{\Phi^n} = \frac{1}{2^n} \Tr(Q^{\dagger}P) = \delta_{P,Q}.
\end{equation}
Moreover, the unitary that maps $(P\otimes I) \ket{\Phi^n} \mapsto \ket{\sym(P)}$ (where $\sym(P)$ is the symplectic representation of $P$) is simply a coherent Bell measurement; that is, the unitary operation which maps the Bell basis to the computational basis.
This measurement is efficient and only requires a constant-depth Clifford circuit \cite{montanaro2017learningstabilizerstatesbell}.
Having performed this step, the actual decoding problem of interest is to recover $\mathbf y$ given the symplectic representation $\sym(P_\mathbf{y})$, so that we can set the register containing $\mathbf y$ to $0^m$. 

By Eqn.~(\ref{eq:homomorphism}), this task is equivalent to the problem of recovering $\mathbf y$, given
\begin{equation}
    \bigoplus_{i \in \supp(\mathbf y)} \sym(P_i) = B^\intercal \mathbf y, \qquad \left( \text{subject to } |\mathbf y| \leq \ell \right) .
\end{equation} 
This task is thus precisely the bounded distance syndrome decoding problem $B^\intercal \mathbf y \mapsto \mathbf y$ on the symplectic code $\Symp(H)$.
This problem can be attacked using classical decoders such as Belief Propagation, as we discuss in the following sections. 
Given such an efficient decoder for up to $\ell$ errors on $\Symp(H)$, we use the output of this decoder to prepare polynomial Hamiltonian states of $H$ up to degree $\ell$.

\begin{figure}
\centering
\quad\quad\quad\begin{minipage}[t]{0.4\textwidth}
\vspace{-1.7cm}
    \textbf{Non-commuting Hamiltonian}\\[1.2em]
    $H = X_1 Z_2 X_3 + Z_1 X_2 Z_3 - X_1$\\[1.4em]
    \begin{tikzpicture}[scale=0.8, every node/.style={scale=0.9}]
        \begin{pgfonlayer}{background}
            \fill[gray!15, rounded corners=6pt] (-3.3,-0.9) rectangle (2.2,2.3);
        \end{pgfonlayer}

        \node at (-2.0,1) {$G = (V,E)$};

        \node[draw,circle,fill=black,minimum size=4pt,inner sep=0pt,label=above:{$Z_1 X_2 Z_3$}] (top) at (0,1.5) {};
        \node[draw,circle,fill=black,minimum size=4pt,inner sep=0pt,label=below left:{$X_1 Z_2 X_3$}] (left) at (-1.2,0) {};
        \node[draw,circle,fill=black,minimum size=4pt,inner sep=0pt,label=below right:{$X_1$}] (right) at (1.2,0) {};

        \draw[line width=0.6pt] (top) -- (left);
        \draw[line width=0.6pt] (top) -- (right);
    \end{tikzpicture}
\end{minipage}
\begin{minipage}[t]{0.48\textwidth}
    \begin{tabular}{@{} l l @{}}
        Permutation: & Sign function: \vspace{3.5mm}\\
        $\pi = (1)(2)(3)$ & $\mathrm{sgn}_G(\pi) = +1$ \\
        $\pi = (1\,2)(3)$ & $\mathrm{sgn}_G(\pi) = -1$ \\
        $\pi = (1\,3)(2)$ & $\mathrm{sgn}_G(\pi) = +1$ \\
        $\pi = (1)(2\,3)$ & $\mathrm{sgn}_G(\pi) = -1$ \\
        $\pi = (1\,2\,3)$ & $\mathrm{sgn}_G(\pi) = -1$ \\
        $\pi = (1\,3\,2)$ & $\mathrm{sgn}_G(\pi) = -1$
    \end{tabular}
\ \\[1.5em]
    $\alpha_G = \underset{\pi \sim S^3}{\mathbf{E}}[\mathrm{sgn}_G(\pi)] 
    = \frac{1}{3!}(2-4) = -\frac{1}{3}$
\end{minipage}
\\[1.2em]
\caption{A simple non-commuting $3$-qubit Pauli Hamiltonian with \emph{anti-commutation graph} $G=(V,E)$, where $V=\{X_1 Z_2 X_3, Z_1 X_2 Z_3,X_1\}$ and $E=\{(X_1 Z_2 X_3,Z_1 X_2 Z_3), (Z_1 X_2 Z_3,X_1)\}$.
The quantity $\alpha_G$ is called the \emph{anti-symmetry character} of $G$ and is a coarse measure of the anti-commutativity of $H$
(see Definition~\ref{def:alpha_function}); in particular, it is equal to $1$ whenever $G$ is the empty graph and all of the Pauli operators in $V$ commute.
}
\label{fig:non-commuting}
\end{figure}

\subsubsection{The non-commuting case}
\label{subsec:intro-non-commuting}

The quantum algorithm we sketched in Sections~\ref{subsec:initial-state-prep} and \ref{subsec:decoding-step} so far crucially relies on the fact that $H = \sum_{i=1}^m v_i P_i$ is a \emph{commuting} Hamiltonian, i.e. $[P_i,P_j] = 0$ for all $i,j$. 
This property allowed us to easily expand $\calP(H)$ into elementary symmetric polynomials; in particular, as a sum of products $P_\mathbf{y} := \prod_{i \in y} P_i$ in which the ordering of the Pauli operators $\{P_i\}_{i \in y}$, for any $y \in \F_2^n$, does not matter.
A general Pauli Hamiltonian, however, may also contain \emph{anti-commuting} terms, as we illustrate for a simple case in Fig.~\ref{fig:non-commuting}. 
The anti-commuting terms render an expansion of $\calP(H)$ into elementary symmetric polynomials much more subtle because we must keep track of the precise anti-commutation properties of the given Pauli operators $\{P_i\}_{i=1}^m$. 
A natural way of representing these properties is through the notion of an \emph{anti-commutation graph} $G=(V,E)$, whereby
\begin{itemize}
    \item the vertex set $V$ consists of the Pauli operators $\{P_i\}_{i=1}^m$ which appear in the Hamiltonian $H$, and
    \item the edge set $E$ contains edges between any two Paulis $P,Q \in V$ if and only if they anti-commute.
\end{itemize}
Equipped with this representation, we tackle the expansion of $\calP(H)$ into elementary symmetric polynomials by enforcing a canonical ordering of the Pauli products. 
For any non-canonically ordered product in the expansion, we simply resort back to the canonical ordering by applying an appropriate permutation of the Pauli operators\footnote{To re-arrange the Pauli product into a canonical ordering, it is instructive to think of a \emph{bubble sort} which only uses swaps between adjacent elements; this occasionally incurs a sign if there is a corresponding edge in the anti-commutation graph $G$.}, which incurs an overall sign. 
For example, if we encountered a product of the first $q$ Pauli operators $\{P_i\}_{i=1}^q$ in some non-canonical ordering specified by a permutation $\pi \in S^q$, we can write
\begin{align}
    P_{\pi(1)} \cdots P_{\pi(q)} = \sgn_{G_q}(\pi) \,\, P_1 \cdots P_q \, ,
\end{align}
where $P_1 \cdots P_q$ is the canonical ordering (according to increasing index order), and the function $\sgn_{G_q}(\pi)$ keeps track of the overall sign of the permutation $\pi$. 
The latter depends on the anti-commutation subgraph $G_q$ in which the vertex set is restricted to $\{P_i\}_{i=1}^q$. 
Each of the $q!$ many possible orderings carries their own sign.  
To account for such symmetries in the polynomial expansion of $\calP(H)$, we introduce the notion of an \emph{antisymmetry character} $\alpha_{G_q}$ of a graph $G_q$, given by the average sign
\begin{align}
    \alpha_{G_q} := \underset{{\pi \sim S^q}}{\mathbf{E}} [\sgn_{G_q}(\pi)] = \frac{1}{q!} \sum_{\pi \in S^q} \sgn_{G_q}(\pi). \label{eq:into-anti-symmetry-character}
\end{align}
Note that an expansion of a degree-$\ell$ polynomial $\calP(H)$ can, in general, feature products of up to $\ell$-many Pauli operators in $\{P_i\}_{i=1}^m$
even with multiplicities, which ultimately requires us to generalize the notion of anti-symmetry characters even further.
In Section~\ref{sec:non-commuting-symmetric-expansion}, we show that a \emph{non-commutative} expansion of $\calP(H) = \sum_{k=0}^\ell c_k H^k$ into elementary symmetric polynomials takes the form
\begin{equation}
    \calP(H)=\sum_{k=0}^{\ell} c_k \sum_{\mathbf{y} \in \mathbb{F}_2^m, |\vec y| \leq k} \beta_G^{(k)}(\mathbf{y}) \cdot \left( \prod_{i \in \supp(\mathbf y)} v_i \right) P_{\vec y} \, ,
\end{equation}
where we call $\beta_G^{(k)}(\mathbf{y})$ the $\beta$-function of order $k$ (formally defined in Definition~\ref{def:beta_function}) and where $P_{\vec y}$ is the canonical ordering (in increasing order) of the Paulis $\{P_i : i \in \supp(\mathbf y)\}$.
At a high level, the aforementioned $\beta$-functions are combinatorial sums of  $\alpha$-functions (i.e. antisymmetry characters) across different subgraphs of $G$ which are labeled by $\mathbf y \in \mathbb{F}_2^m$. 
Both $\alpha$-functions and $\beta$-functions have a rich graph-theoretic structure.
In particular, in Lemmas~\ref{lemma:alpha_factorization} and \ref{lemma:beta_function_reduced_form}, we show that these functions can be decomposed according to the \emph{connected components} of the anti-commutation graph $G$. 

In Section~\ref{sec:main-non-commuting}, we fully describe how to implement the HDQI algorithm for general, non-commuting, Pauli Hamiltonians. 
The main difference to the commuting case is that the pilot state now carries amplitudes of the non-commutative expansion of $\calP$. 
In particular, instead of preparing the initial state in Eqn.~(\ref{eq:initial_state}), it suffices to prepare the state proportional to
\begin{align}
    \sum_{k = 0}^\ell c_k \sum_{\vec y \in \F_2^m} \beta_G^{(k)}(\mathbf{y})  \, \ket{\vec y} \otimes \ket{\Phi^n} \, .
\end{align} 
Importantly, $\beta_G^{(k)}(\mathbf{y})$ is equal to $0$, whenever the Hamming weight of $\vec y$ exceeds $k$ (see Definition~\ref{def:beta_function}). 
Due to the intricate complexity behind the $\alpha$- and $\beta$-functions, we do not expect the state preparation step to be computationally tractable in general. 
In many cases, however, it is possible to exploit the graph-theoretic structure of the $\beta$-function in order to prepare the state efficiently. 
In Theorem~\ref{thm:resource_state_algprithm_MPS}, we show that the general HDQI pilot state is in fact a qudit \emph{matrix product state} (MPS)~\cite{mps1,mps2} with $O(n)$ bond dimension, with qudit dimension $\poly(n)$ provided that the anti-commutation graph $G$  has small connected components.
By exploiting this structure, we are able to efficiently or nearly efficiently (in the sense of a quasipolynomial time classical pre-processing step to compute the matrices) prepare the pilot state when the Hamiltonian's anticommutation graph has small connected components.

\subsection{Overview of results}
\label{sec:teaser-applications}
We now present our main applications of Hamiltonian DQI to Gibbs sampling and optimization. 
In Fig.~\ref{fig:HDQI_spectrum}, we also summarize the Hamiltonian models for which we have obtained new results via HDQI and our improved classical algorithms developed in this work.

\begin{figure}[t!]
    \centering
    \includegraphics[width=0.9\linewidth]{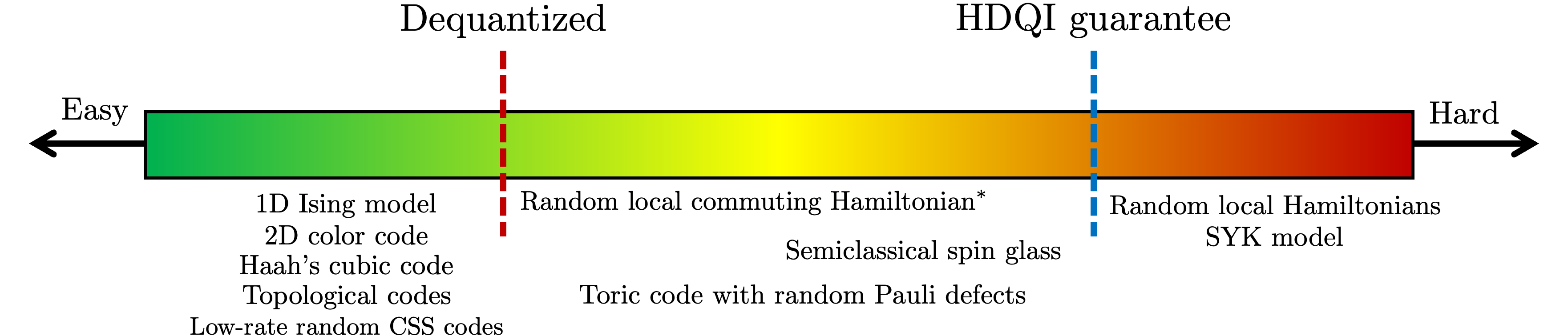}
    \caption{Spectrum of hardness for Hamiltonian models studied in this work. To the left are models which are dequantized in the sense that while HDQI can Gibbs sample them to arbitrary temperature, we develop a classical algorithm which efficiently samples stabilizer descriptions of eigenstates from the Gibbs distribution at any temperature. 
    In the middle are models for which HDQI prepares Gibbs states up to a constant temperature (perhaps with a quasipolynomial classical pre-processing step). 
    The asterisk denotes an assumption of high-distance decodability, which we show numerically but not analytically. 
    To the right are models for which we do not currently provide provable guarantees about the performance of HDQI. 
    }
    \label{fig:HDQI_spectrum}
\end{figure}

\subsubsection{Nearly independent Hamiltonians}

If the degree $\ell$ of $\calP$ exceeds the number of distinct eigenvalues of $H$, then any function of $H$ can be exactly interpolated by a degree-$\ell$ polynomial. 
In this case, HDQI prepares arbitrary non-negative functions of $H$. 
This includes \begin{itemize}
    \item the Gibbs state at any inverse temperature $0 \leq \beta \leq \infty$;
    \item any microcanonical ensemble $\rho_E(H) \propto \sum_{\ket{\psi} : H\ket{\psi} = E \ket{\psi}} \ketbra{\psi}{\psi}$, which is the uniform mixture within any fixed energy eigensector;
    \item arbitrary band-pass, Chebyshev, or Gaussian filters. 
\end{itemize} 
We identify several physically-motivated Hamiltonians for which HDQI prepares arbitrary non-negative spectral filters. 
Specifically, we show that this includes all Hamiltonians whose symplectic code has constant dimension, which we call ``nearly independent'' Hamiltonians.
\begin{restatable}[Nearly independent Hamiltonians]{theorem}{THMdecodable}
\label{thm:arbitrary_temperature}
    Let $H = \sum_{i=1}^m v_i P_i$ be a commuting Pauli Hamiltonian on $n$ qubits with $v_i \in \{\pm1\}$ and $m = \poly(n)$. 
    If the symplectic code of $H$ is of constant dimension (number of logical bits $k = \dim \Symp(H) = O(1)$), then HDQI efficiently prepares the density matrix \begin{equation}
        \rho_f(H) := f(H)/\Tr[f(H)]\label{eq:polynomialState}
    \end{equation} 
    for any non-zero function $f:\Spec(H) \to \R_{\geq 0}$ in time $\poly(n)$, where $\Spec(H)$ is the set of eigenvalues of $H$. 
    In particular, for such Hamiltonians, HDQI efficiently prepares Gibbs states at \emph{any} temperature $0 \leq \beta \leq \infty$.
\end{restatable}

Nearly independent Hamiltonians include essentially all 2D translation-invariant (TI) Pauli stabilizer Hamiltonians, the color code, cluster Hamiltonians, and 3D Hamiltonians like Haah's code. Some examples of such Hamiltonians (and notable counterexamples) are summarized in Table~\ref{tab:constant-relations}.
For the 2D toric code Hamiltonian, efficient quantum Gibbs sampling at arbitrary temperature has been shown recently in \cite{DLL24, HJ24, PSS25}.
These results are based on proving bounds on mixing times and associated spectral gaps.
Hamiltonian DQI provides a simple alternative approach that is mixing-free and based on coherent decoding, and moreover also prepares arbitrary non-negative spectral filters. 

Our other examples of nearly independent Hamiltonians, such as the color code or Haah's code, are not analyzed in \cite{DLL24, HJ24}. Recent work \cite{PSS25} provides numerical evidence of a poly-depth duality between these models and two decoupled Ising models, which would suggest efficient thermalization at arbitrary temperatures. 
However, it seems that HDQI is the first proven efficient Gibbs sampler for these models at arbitrary temperature. 

Interestingly, we show that it is not even necessary to use a quantum computer to Gibbs sample from these models, in the following sense of the word ``sample": if one is only interested in sampling stabilizer descriptions of eigenvectors from the distribution in Eqn.~(\ref{eq:polynomialState}) (as opposed to preparing the quantum state or a coherent quantum purification as in Eqn.~(\ref{eq:purified_hdqi_sec2})), there is a classical algorithm to do so. We state this as Theorem~\ref{thm:arbitrary_temperature_dequant}.
\begin{restatable}[Classical sampling for nearly independent Hamiltonians]{theorem}{THMDecodableDequant}
\label{thm:arbitrary_temperature_dequant}
    Under the same conditions on $H$ as Theorem~\ref{thm:arbitrary_temperature}, there is a classical algorithm which takes time $O(2^k m^2)$ and samples classical descriptions of stabilizer states such that the corresponding mixture of stabilizer states is $\rho_f(H)$.
\end{restatable}
We prove Theorem~\ref{thm:arbitrary_temperature_dequant} in Section~\ref{sec:perfect_classical}.
Although the proof is fairly straightforward, we did not find it in \cite{KB16, HJ24, DLL24, PSS25}, which similarly consider quantum algorithms for Gibbs sampling from quantum code Hamiltonians, a class that includes nearly independent Hamiltonians. 

We note that this classical algorithm does not recover the full power of HDQI for the same Hamiltonians, because given the state $\rho_f(H)$ one can use Kitaev's phase estimation algorithm to estimate the expectation value of an observable $M$ such that $e^{-i M t}$ can be efficiently implemented as a quantum circuit. 
By contrast, the above classical algorithm only directly yields efficient means to estimate observables within the Clifford group.

\begin{table}[t]
  \centering
  \setlength{\tabcolsep}{7pt}
  \renewcommand{\arraystretch}{1}
  \begin{tabular}{lccp{5.8cm}}
    \toprule
    \textbf{Model} & \textbf{Geometry} & $\boldsymbol{\dim\Symp(H)}$ & \textbf{Reason / remarks} \\
    \midrule
    Ising (ring) & 1D, periodic & $1$ & Unique cycle relation $\prod ZZ=I$. \\
    Surface/toric & 2D, closed & $2$ & Global star and plaquette products; homology on closed surfaces. \\
    Color code (stabilizer) & 2D, closed & $4$ & Two independent two-color products per Pauli sector. \\
    TI Pauli in 2D & 2D, closed & $2K$ & Local Clifford-equivalent to $K$ toric copies; $K$ constant. \\
    Haah's cubic code & 3D, periodic & $O(1)$ & Four non-local relations for generic CSS ``Type 1" codes. \\
    \midrule
    \multicolumn{4}{c}{\textbf{Counterexamples}} \\
    \midrule
    3D toric code & 3D, periodic & $\Theta(L^3)$ & Cube-local plaquette identities. \\
    X-cube / checkerboard & 3D, periodic & $\Theta(L^3)$ & Many local relations (type-I fracton). \\
    \bottomrule
  \end{tabular}
  \caption{Examples of nearly independent commuting Hamiltonians, i.e. commuting Hamiltonians whose symplectic code has constant dimension $k$. 
  HDQI efficiently prepares arbitrary non-negative functions of such Hamiltonians. 
  The last two entries are notable counterexamples, where the code dimension scales with system size.}
  \label{tab:constant-relations}
\end{table}

\subsubsection{Random local commuting Hamiltonians} 

For general Hamiltonians, our quantum algorithm efficiently prepares Gibbs states $\propto \exp(-\beta H)$ by choosing $\calP(x)$ to approximate the exponential function $\exp(-\beta x/2)$. 
In the generic case, we show that 
the achievable temperature $\beta$ is linearly related to the number of efficiently decodable errors in the symplectic code of $H$. 

\begin{restatable}[Gibbs state preparation]{theorem}{Gibbsthm}
    \label{thm:gibbs_explicit}
       For any Hamiltonian $H$ and approximation error $\delta >0$, there exists a polynomial $\calP$ of degree 
       
    \begin{equation}
        \ell \leq 1.12 \beta \lVert H\rVert + 0.648\ln \left(\frac{12}{\eps}\right)
    \end{equation}
    such that the polynomial Hamiltonian state $\rho_{\calP}(H) = \frac{\calP^2(H)}{\Tr[\calP^2(H)]}$ satisfies
    \begin{equation}
        \frac{1}{2} \left\lVert \rho_{\calP}(H) -\frac{e^{-\beta H}}{\Tr[e^{-\beta H}]} \right\rVert_1 \leq \delta.
    \end{equation}  
\end{restatable}
For many good code ensembles, including random LDPC codes, it is possible to efficiently decode errors of linear weight, which allows us to choose $\ell = \Theta(m)$. 
Since $\norm{H}$ is also linear in $m$ generically, HDQI prepares Gibbs states at constant temperature for such Hamiltonians.
To provide a concrete example, we focus on a particular class of random, commuting $k$-local Hamiltonians whose symplectic code can be decoded by Belief Propagation (BP). 
Beyond Gibbs states, an additional application of HDQI is \emph{optimization} or ground-state energy estimation. 
Indeed, this is the primary focus of Decoded Quantum Interferometry \cite{JSW24} for classical Hamiltonians. 
The performance of DQI was analyzed in \cite{JSW24} through a Fourier-analytic tool called the semicircle law, mapping the number of decodable errors $\ell$ to an approximation ratio. 
Here, we derive an analogous result for HDQI.
\begin{restatable}[Semicircle law for HDQI]{theorem}{semicircle}\label{thm:semicircle}
   Let $H = \sum_i v_i P_i$ be a commuting Pauli Hamiltonian, where $v_i \in \set{\pm 1}$ and the $P_i$ are $n$-qubit Paulis. 
   Given a degree-$\ell$ polynomial $\calP$, let $\langle E\rangle$ be the expected energy obtained upon measuring the corresponding polynomial Hamiltonian state $\rho_\calP(H)$ in the eigenbasis of $H$.
   Suppose $2 \ell+1<d^{\perp}$ where $d^{\perp}$ is the minimum distance of the symplectic code $\Symp(H)$. 
   In the limit $m \rightarrow \infty$, with $\ell / m$ fixed, the optimal choice of degree-$\ell$ polynomial $\calP$ to maximize $\langle E \rangle$ yields
    \begin{equation}
        \frac{\langle E\rangle}{m}=\left(\sqrt{\frac{\ell}{2m}}+\sqrt{\left(\frac{1}{2}-\frac{\ell}{2m}\right)}\right)^2
    \end{equation}
    if $\frac{\ell}{m} \leq \frac{1}{2}$ and $\frac{\langle E \rangle}{m}=1$ otherwise.
\end{restatable}
On this random local commuting model, we compare the performance of HDQI to a variant of the classical Simulated Annealing (SA) algorithm in Section~\ref{sec:commuting_dequantizations}. While HDQI does not reach lower energies than SA, we can choose $\calP$ to prepare Gibbs distributions\footnote{Up to a conjecture related to the analysis of HDQI beyond half the distance of the symplectic code, see Section~\ref{sec:random_local_commuting}.} at temperatures for which it is unclear whether SA efficiently converges to the Gibbs distribution.

\subsubsection{Semiclassical spin glasses}

We conclude with an application to a physically-relevant non-commuting Hamiltonian model for which HDQI is provably efficient. Since we cannot handle Hamiltonians which have strong non-commuting structure (e.g. a random $k$-local Hamiltonian), the primary usage of our results above which we explore is the optimization of Hamiltonians which begin as a commuting Hamiltonian, and then undergo a \emph{defect} process wherein some randomly-chosen terms are modified so that the Hamiltonian is no longer commuting.
We introduce an object which we call \emph{template graphs} to prove that for Hamiltonians which are local in a certain sense, there is a constant threshold value for the probability of a term becoming defected, under which the anticommutation graph has small connected components with high probability.

As an application, we prove that a local classical spin glass defected with some quantum terms---i.e. $Z$-type Hamiltonian, with each term becoming $X$-type with probability $p$---can be Gibbs sampled up to a constant temperature by HDQI if $p$ is a sufficiently small constant.
This model sits in between a completely classical spin glass, and a quantum spin glass wherein $p = 1/2$, so we refer to it as a semiclassical spin glass.

\begin{restatable}[Threshold theorem, $p$-semiclassical spin glass]{theorem}{THMspinglass}\label{thm:semiclassical_spin_glass_HDQI}
    Let $m = cn$ for a constant $c$, and let $K = a(b-1)$.
    Let $B^\intercal_0 \in \F_2^{n \times m}$ be a random matrix with row-weight $b$ and column-weight $a$, and let $\mathbf{h}_i$ be the $i$th column of $B^\intercal_0$.
    Define the Hamiltonian $H(0) = \sum_{i=1}^m v_i Z^{\mathbf{h}_i}$, where $v_i \in \set{\pm 1}$ and $Z^{\mathbf{h}_i}$ is a $Z$ Pauli supported on non-zero indices of $\mathbf{h}_i$.
    Define $H(p)$ by independently mapping each term $Z^{\mathbf{h}_i}$ in $H(0)$ to $X^{\mathbf{h}_i}$ with probability $p$.
    Then there exists a constant $d(a, b, c)$ such that with probability $1 - 1/n^{\Omega(1)}$, for any $p < 1/K^2$ and for any polynomial $\calP$ of degree $\ell \leq d(a, b, c) \cdot m$, there exists a quantum algorithm which outputs the HDQI state \begin{align}
        \rho_{\calP}(H(p)) = \frac{\calP^2(H(p))}{\Tr[\calP^2(H(p))]}
    \end{align} 
    using $\poly(n)$ time on a quantum device alongside a $n^{O(\log n)}$ time classical pre-processing step.
\end{restatable}

We moreover prove that nearly independent Hamiltonians whose symplectic code has high distance, such as the toric code, are also Gibbs-sampleable up to a constant temperature even if every term defects to a different non-commuting term with a small constant probability.
As stated, these samplers all require a quasipolynomial classical pre-processing step in addition to efficient quantum computation.
They can, however, also be made efficient if the success probability $1 - 1/n^{\Omega(1)}$ is replaced by a fixed constant which can be chosen to be arbitrarily close to 1.

\subsection{Organization of the paper}
The remainder of this paper is organized as follows. 
In Section~\ref{sec:commuting}, we present a complete description of Hamiltonian DQI (HDQI) in the special case of commuting Hamiltonians. 
We subsequently discuss our applications of HDQI to specific families of commuting Hamiltonians in Section~\ref{sec:commuting-applications}. 
In Section~\ref{sec:commuting_dequantizations}, we give classical algorithms for these families of Hamiltonians, for the task of incoherently sampling an eigenstate.
We then extend HDQI to the general setting of non-commuting Hamiltonians in Section~\ref{sec:general}, followed by corresponding applications to physically interesting Hamiltonians in Section~\ref{sec:applications_non-commuting}. 
Finally, we show that HDQI can be extended beyond signed Pauli Hamiltonians in \Cref{sec:beyond_paulis}.


\section{Commuting Pauli Hamiltonians: the simple case}\label{sec:commuting}

To begin, we study the restricted case when the Hamiltonian has pairwise commuting terms.
While commuting Hamiltonians are a special case of our main reduction in Section~\ref{sec:general}, the state preparation portion of our algorithm simplifies significantly when all terms commute.

\begin{definition}[Commuting Pauli Hamiltonians] \label{def:commuting_Hamiltonian}
    Let $H = \sum_i v_i P_i$, where $v_i \in \set{\pm 1}$ and the $P_i$ are $n$-qubit Paulis. We say that $H$ is a \emph{commuting Pauli Hamiltonian} if $[P_i, P_j] = 0$ for all $i, j$.
\end{definition}
We first show how to prepare states $\propto \calP(H)$, where $\calP$ is a degree-$\ell$ univariate polynomial, in the special case when $H$ is a commuting Hamiltonian. 
This process requires only a decoder which can correct weight-$\ell$ errors of a certain code depending on $H$, as we will show that the state preparation requirement of the reduction can always be performed efficiently.
Although most Pauli Hamiltonians do not have commuting terms, many Hamiltonians of broad interest do have this form. 
For example, any quantum stabilizer code given by stabilizers $S_1, \dots, S_r$, where each $S_i$ is a $n$-qubit Pauli, by definition satisfies $[S_i, S_j] = 0$. 
The parent Hamiltonian of such a code is given by $H = \sum_{i} S_i$ and is commuting. 
We may also add extra terms which are not independent of the first $r$ terms, producing Hamiltonians such as that of the two-dimensional toric code.
In these cases, while finding the ground state space is trivial---it is the stabilizer code space---preparing Gibbs states at low temperatures is not immediately obvious. 
Indeed, the task of preparing low-temperature Gibbs states for commuting Hamiltonians has proved challenging, being solved only for certain special cases~\cite{HJ24}. 
Our algorithm is by no means limited to stabilizer code Hamiltonians even in the commuting case, however, and in Section~\ref{sec:commuting-applications} we give additional examples based on physically motivated commuting Hamiltonians as well as random ensembles of commuting Hamiltonians for which we can prepare Gibbs states and sample low-energy states.

A key element of our algorithm is an expansion of any polynomial $\calP(x)$ on a single formal variable $x$ into a linear combination of symmetric terms, when $x$ is the sum of many formal variables $\sum_{i=1}^m z_i$. 
This expansion is well-known~\cite{JSW24}, but we give a derivation that computes all the coefficients exactly for purposes of comparison when we move to the general case. 
We show in Appendix~\ref{app:combinatorics} that these coefficients are efficiently computable.

\begin{theorem}[Univariate polynomial symmetric expansion] \label{thm:commuting_symmetric_polynomial_expansion}
    Let $\calP(x)$ be a univariate polynomial of degree $\ell$, and let $x = \sum_{i=1}^m z_i$ such that $z_i^2 = 1$ for all $i$. 
    Then there exists coefficients $w_0, \dots, w_\ell$ independent of the argument $x$, such that \begin{align}
        \calP\left( \sum_{i=1}^m z_i \right) = \sum_{j=0}^\ell w_j \sum_{\mathbf{y} \in \F_2^m : |\mathbf{y}| = j} z_1^{y_1} \cdots z_m^{y_m} .
    \end{align}
\end{theorem}
\begin{proof}
    Let $\calP(x) = \sum_{k=0}^\ell c_k x^k$, so that by definition \begin{align}
        \calP\left( \sum_{i=1}^m z_i \right) & = \sum_{k=0}^\ell c_k \sum_{1 \leq a_1, \dots, a_k \leq m} z_{a_1} \cdots z_{a_k} .
    \end{align}
    For $k \leq m$ and $\bmu \in \Z_{\geq 0}^m$, we say that $\bmu \vDash_m k$ if $|\bmu| := \sum_{i=1}^m \mu_i = k$. Now, to every term on the right-hand side we will associate a counting vector $\bmu \in \Z_{\geq0}^m$ such that index $i$ appears $\mu_i$ times. There are exactly $\binom{k}{\bmu} = \frac{k!}{\mu_1! \cdots \mu_m!}$ ways to choose $a_1, \dots, a_k$ such that the counting vector associated to $z_{a_1} \cdots z_{a_k}$ is $\bmu$. Hence, \begin{align}
        \calP\left( \sum_{i=1}^m z_i \right) & = \sum_{k=0}^\ell c_k \left( \sum_{i=1}^m z_i \right)^k = \sum_{k=0}^\ell c_k \sum_{\bmu : \bmu \vDash_m k} \binom{k}{\bmu} \cdot z_{1}^{\mu_1} \cdots z_{m}^{\mu_m} .
    \end{align}
    where the second line sorts the variables in order, since they all commute. Next, because $z_i^2 = 1$, only the parity of each $\mu_i$ matters. We use the notation $\bmu \text{ (mod 2)}$ to denote the element-wise mod 2 of a vector $\bmu$. We group the sum over $\bmu$ by their parity $\mathbf{y} = \bmu \text{ (mod 2)} \in \F_2^m$, and then perform a second grouping of the $\mathbf{y}$ vectors by their weight. In total,
    \begin{align}  \label{eq:commutative_expansion_divergence_point}
        \sum_{\bmu \vDash_m k} \binom{k}{\bmu} \cdot z_{1}^{\mu_1} \cdots z_{m}^{\mu_m} 
        &= \sum_{j=0}^k \, \,\sum_{\substack{\mathbf{y} \in \F_2^m\\|\mathbf{y}| = j}} \left( \sum_{\substack{\bmu : \bmu \vDash_m k\\ \bmu \text{ (mod 2)} = \mathbf{y}
        }} \binom{k}{\bmu} \right)
        z_{1}^{y_1} \cdots z_{m}^{y_m}.
    \end{align}
    The combinatorial factor $\binom{k}{\bmu}$ is invariant under permutations of the entries of $\bmu$. 
    Therefore, the term in parentheses above depends only on the weight of $\mathbf{y}$. 
    Consequently, we may exchange that term with the sum over $\mathbf{y}$ to obtain
    \begin{align}
        \sum_{\bmu \vDash_m k} \binom{k}{\bmu} \cdot z_{1}^{\mu_1} \cdots z_{m}^{\mu_m} 
        & = \sum_{j=0}^\ell b_k^{(j)} \sum_{\substack{\mathbf{y} \in \F_2^m\\|\mathbf{y}| = j}} 
        z_{1}^{y_1} \cdots z_{m}^{y_m} ,
    \end{align}
    where for any fixed choice of $\mathbf{y}$ such that $|\mathbf{y}| = j$, \begin{align}
        b_k^{(j)} = \sum_{\substack{\bmu : \bmu \vDash_m k\\ \bmu \text{ (mod 2)} = \mathbf{y}}} \binom{k}{\bmu} \, , \, \text{ if $0 \leq j \leq k$, \, and \, $b_k^{(j)}=0$, otherwise},
    \end{align}
    is a combinatorial factor. Now define $w_j = \sum_{k=0}^\ell c_k b_k^{(j)}$. 
    In sum, we obtain the claimed expansion: \begin{align}
        \calP\left( \sum_{i=1}^n z_i \right) & = \sum_{k=0}^\ell c_k \sum_{j=0}^\ell b_k^{(j)} \sum_{\mathbf{y} \in \F_2^n : |\mathbf{y}| = j} z_{1}^{y_1} \cdots z_{n}^{y_n} \\
        & = \sum_{j=0}^\ell \left(\sum_{k=0}^\ell  c_k b_k^{(j)} \right) \sum_{\mathbf{y} \in \F_2^n : |\mathbf{y}| = j} z_{1}^{y_1} \cdots z_{n}^{y_n} \\
        & = \sum_{j=0}^\ell w_j \sum_{\mathbf{y} \in \F_2^n : |\mathbf{y}| = j} z_{1}^{y_1} \cdots z_{n}^{y_n} .
    \end{align}
\end{proof}

\subsection{Reduction to decoding}
Using Theorem~\ref{thm:commuting_symmetric_polynomial_expansion}, we give a reduction from preparing the quantum state proportional to $\calP(H)$ to decoding a corresponding \textit{classical} code.
The reduction also requires the production of a certain \emph{pilot state}, which is for a general Hamiltonian not necessarily efficiently preparable.
However, the pilot state for a commuting Hamiltonian is always efficiently preparable, and hence the reduction is formally directly from preparing $\calP(H)$ to a decoding problem.
We first define the classical code which must be decoded in this reduction.

\begin{definition}[Symplectic vectors] \label{def:symplectic_vectors}
    Let $P$ be a $n$-qubit Pauli operator. 
    Define $\mathbf{h} := \sym(P) \in \F_2^{2n}$ to be the \emph{symplectic representation} of $P$. 
    That is, if $P$ on the $j$th qubit is $X$, then $h_j = 1, h_{j+n} = 0$; if it is $Z$, then $h_j = 0, h_{j+n} = 1$; and if it is $Y$, then $h_j = h_{j+n} = 1$. 
\end{definition}
The symplectic vector representation $\mathbf{b}$ completely determines a Pauli $P$ up to an overall phase.
Next, a classical lin ear code $\calC$ with $m$ physical bits is a linear subspace of $\F_2^m$. 
Such a code can be specified as the kernel of a $q \times m$ binary matrix $B^\intercal$, where $q \leq m$. 
In that case, $B^\intercal$ is known as the \emph{parity check matrix} of $\calC$.

\begin{definition}[Hamiltonian symplectic code] \label{def:symplectic_codes}
    Let $H = \sum_{i=1}^m v_i P_i$, where $v_i \in \set{\pm 1}$ and the $P_i$ are $n$-qubit Paulis. Let $\mathbf{h}^{(i)} = \sym(P_i)$ be the symplectic vector of $P_i$ as in Definition~\ref{def:symplectic_vectors}. 
    The \emph{symplectic code} of $H$, denoted $\Symp(H)$, is a classical linear code whose parity check matrix is given by binary matrix $B^\intercal \in \F_2^{2n \times m}$ such that \begin{align}
        B^\intercal = \begin{bmatrix}
        | & | &        & | \\
        \mathbf{h}^{(1)} & \mathbf{h}^{(2)} & \cdots & \mathbf{h}^{(m)} \\
        | & | &        & | 
        \end{bmatrix} .
    \end{align}
\end{definition}
Note that $H$ is a completely general Hamiltonian in Definition~\ref{def:symplectic_codes}; it need not be commuting. 
Our reduction will formally be from preparing $\rho_{\calP}(H) \propto \calP(H)$ to decoding all errors in $\Symp(H)$ with weight up to $\deg(\calP)$.
First, we show that the pilot state, which we can think of as a resource state for the reduction, can be efficiently prepared.

\begin{lemma}[Efficient preparation of pilot state for commuting Hamiltonians] \label{lemma:commuting_efficient_resource_state}
    Let $w_0, \dots, w_{\ell} \in \R$, where $\ell \leq n$. There exists a quantum algorithm running in time $\poly(m)$ which prepares the state \begin{align}
        \ket{\calR^\ell_{\text{comm}}(H)} \propto \sum_{j=0}^{\ell} w_j \sum_{\mathbf{y} \in \F_2^m : |\mathbf{y}| = j} \ket{\mathbf{y}} .
    \end{align}
\end{lemma}
\begin{proof}
    The normalization of $\ket{\calR^\ell_{\text{comm}}(H)}$ is explicitly given by $1/\calN = \left( \sum_{j=0}^{\ell} w_j^2 \binom{m}{j} \right)^{-1/2}$. 
    To prepare $\ket{\calR^\ell_{\text{comm}}(H)}$, we first prepare the state $\ket{\psi_1} = \frac{1}{\calN} \sum_{j=0}^{\ell} w_j \sqrt{\binom{m}{j}} \ket{j}$. 
    This can be done in $\poly(m) = \poly(n)$ time because the superposition is over only $O(\ell)$ states, which can be represented in $O(\log \ell)$ qubits. 
    An explicit algorithm to execute this preparation algorithm is given in~\cite{low2024trading}. 
    Next, each term in the pilot state is proportional to a Dicke state \begin{align}
        \ket{D_j^m} = \frac{1}{\sqrt{\binom{m}{j}}} \sum_{\mathbf{y} \in \F_2^m : |\mathbf{y}| = j} \ket{\mathbf{y}} ,
    \end{align}
    for which efficient quantum state preparation algorithms mapping $\ket{j} \mapsto \ket{D_j^m}$ in superposition are well-known~\cite{bartschi2022short}. 
    By applying any one such algorithm, we can prepare \begin{align}
        \ket{\psi_2} = \frac{1}{\calN} \sum_{j=0}^{m} w_j \sqrt{\binom{m}{j}} \frac{1}{\sqrt{\binom{m}{j}}} \sum_{\mathbf{y} \in \F_2^m : |\mathbf{y}| = j} \ket{\mathbf{y}} = \frac{1}{\calN} \sum_{j=0}^{m} w_j \sum_{\mathbf{y} \in \F_2^m : |\mathbf{y}| = j} \ket{\mathbf{y}} = \ket{\calR^\ell_{\text{comm}}(H)} .
    \end{align}
\end{proof}

We next make formal the notion of a classical decoding algorithm which can be called by a quantum algorithm. The following definition concerns an idealized version of a decoder, but we show in Section~\ref{sec:decoding_failures} that our results extend to approximate decoders as well. 
\begin{definition}[Decoding oracle] \label{def:decoding_oracle}
    Let $B^\intercal \in \F_2^{2n \times m}$ be the parity check matrix of a classical linear code which is the symplectic code (as in Definition~\ref{def:symplectic_codes}) $\Symp(H)$ of a Pauli Hamiltonian $H = \sum_{i=1}^m v_i P_i$. 
    A weight-$\ell$ \emph{decoding oracle} $\calD_{H}^{(\ell)}$ for $H$ is a unitary operation such that \begin{align}
        \calD_{H}^{(\ell)} \ket{\mathbf{y}} \ket{B^\intercal \mathbf{y}} =  \ket{0} \ket{B^\intercal \mathbf{y}}
    \end{align} 
    for any $\mathbf{y} \in \F_2^m$ such that $1\leq |\mathbf{y}| \leq \ell$.
\end{definition}
The decoding oracle can be implemented by a classical decoder which can correct any error of weight up to $\ell$ on $\Symp(H)$: use the classical decoder to compute $\mathbf{y}$ from $B^\intercal \mathbf{y}$ and coherently add (via bitwise XOR) that to the first register so that $\mathbf{y} \oplus \mathbf{y} = 0$ as desired above. 
We now use this decoding oracle to produce the reduction.

\begin{theorem}[Main reduction, commuting case] \label{thm:commuting_reduction_to_decoding}
    Let $H = \sum_{i=1}^m v_i P_i$ be a commuting Hamiltonian, where $v_i \in \set{\pm 1}$ and the $P_i$ are $n$-qubit Paulis. 
    Assume that $m = \poly(n)$. 
    Suppose that $\calD_H^{(\ell)}$ is a weight-$\ell$ decoding oracle for $H$ (as in Definition~\ref{def:decoding_oracle}).
    Then Algorithm~\ref{alg:amplitude_transform} runs in time $\poly(n)$ and prepares the state \begin{align}
        \rho_{\calP}(H) = \frac{\calP^2(H)}{\Tr[\calP^2(H)]}
    \end{align}
    for any polynomial $\calP$ such that $\deg(\calP) \leq \ell \leq m$. 
    The algorithm makes a single call to $\calD_H^{(\ell)}$.
\end{theorem}
\begin{proof}
    Let $z_i = v_i P_i$ and note that $z_i^2 = 1$. By Theorem~\ref{thm:commuting_symmetric_polynomial_expansion}, \begin{align}
        \calP(H) & = \calP\left( \sum_{i=1}^m z_i \right) = \sum_{j=0}^{\ell} w_j \sum_{\mathbf{y} \in \F_2^m : |\mathbf{y}| = j} z_1^{y_1} \cdots z_m^{y_m} ,
        \label{eq:commuting_reduction_matching_expansion}
    \end{align}
    where the $w_j$ depend only on $\calP$. 
    Without loss of generality, assume that $\ell \in 2\Z$. 
    We initialize three registers $A, B, C$. 
    Register $A$ has $m$ qubits and holds the pilot state $\ket{\calR^\ell_{\text{comm}}(H)}_A$ from Lemma~\ref{lemma:commuting_efficient_resource_state}, while registers $B$ and $C$ each hold $n$ qubits and are initialized jointly as a maximally entangled state $\ket{\Phi^n}_{BC} = \sum_{i=0}^{2^n-1} \ket{i}_B \ket{i}_C$.
    \begin{align}
        \ket{\psi_1}_{ABC} = \ket{\calR^\ell_{\text{comm}}(H)}_A \otimes \ket{\Phi^n}_{BC} = \frac{1}{\calN} \sum_{j=0}^{m} w_j \sum_{\mathbf{y} \in \F_2^m : |\mathbf{y}| = j} \ket{\mathbf{y}}_A \otimes \ket{\Phi^n}_{BC} ,
    \end{align}
    where $\calN$ is given in Lemma~\ref{lemma:commuting_efficient_resource_state}.
    We next apply $\bigotimes_{i=1}^m Z^{(1 - v_i)/2}_i$ to register $A$, i.e., we apply a phase of $-1$ each time $v_i = -1$. This yields the state \begin{align}
        \ket{\psi_2}_{ABC} = \frac{1}{\calN} \sum_{j=0}^{\ell} w_j \sum_{\mathbf{y} \in \F_2^m : |\mathbf{y}| = j} \left( \prod_{i \in \supp(\mathbf{y})} v_i \right) \ket{\mathbf{y}}_A \otimes \ket{\Phi^n}_{BC} ,
    \end{align}
    where $\supp(\mathbf{y})$ is the support of $\mathbf{y}$, i.e. the set of indices $i$ such that $y_i = 1$.
    Now define $P_\mathbf{y} := \prod_{i \in \supp(\mathbf{y})} P_i$ and $(C^{(i)} P)_{A \to B}$ to be the controlled-$P$ operation where the $i$th qubit of register $A$ is the control, and all of register $B$ is the target.
    Apply the operation $\prod_{i=1}^m (C^{(i)}P_i)_{A \to B}$.
    This operation transforms our state into \begin{align}
    \ket{\psi_3}_{ABC} & = \frac{1}{\calN} \sum_{j=0}^{\ell} w_j \sum_{\mathbf{y} \in \F_2^m : |\mathbf{y}| = j}  \ket{\mathbf{y}}_A \otimes \left( \prod_{i \in \supp(\mathbf{y})} v_i \right)  (P_{\mathbf{y}} \otimes I)_{BC} \ket{\Phi^n}_{BC} . 
    \end{align}
    This state is {\em almost} the target state: if the state in register $A$ were not $\ket{\bf{y}}$ but some fixed state, say $\ket{0}$, $\ket{\psi_3}$ would be a tensor product over the state $\ket{0}_A$ and the state in registers $BC$, which equals
    \begin{align}
        \frac{1}{\calN} \sum_{j=0}^\ell w_j \sum_{\mathbf{y} \in \F_2^m : |\mathbf{y}| = j} \left( \prod_{i \in \supp(\mathbf{y})} v_i P_i \otimes I \right) \,  \ket{\Phi^n}_{BC} &=  \frac{1}{\calN} \sum_{j=0}^{\ell} w_j \sum_{\mathbf{y} \in \F_2^m : |\mathbf{y}| = j} z_1^{y_1} \cdots z_m^{y_m} \otimes I \, \ket{\Phi^n}_{BC}\\
        &\propto (\calP(H) \otimes I ) \, \ket{\Phi^n}_{BC}
    \end{align}
    where the equality follows from Eqn.~(\ref{eq:commuting_reduction_matching_expansion}).
    Therefore, the final step is to uncompute register $A$ coherently, which we carry out in two steps.
    First, in superposition, we apply a coherent Bell measurement, which maps $(P_{\mathbf{y}} \otimes I) \ket{\Phi^n}$ to $\sym(P_{\mathbf{y}})$ for all $\mathbf{y}$.
    The construction follows from the unique property of the maximally entangled state that every Pauli applied to $\ket{\Phi^n}$ maps it to a distinct orthonormal state; the details of the coherent Bell measurement are given in Appendix~\ref{app:bell_basis}.
    The coherent measurement produces the state \begin{align}
        \ket{\psi_4}_{ABC} = \frac{1}{\calN} \sum_{j=0}^{\ell} w_j \sum_{\mathbf{y} \in \F_2^m : |\mathbf{y}| = j}  \ket{\mathbf{y}}_A \otimes \left( \prod_{i \in \supp(\mathbf{y})} v_i \right) \ket{\sym(P_{\mathbf{y}})}_{BC} .
    \end{align}
    The map $\sym(\cdot)$ is an isomorphism between the group of $n$-qubit Paulis (modulo overall phase) under multiplication and the group of $\F_2^{2n}$ under addition mod 2, which we denote by $\oplus$. 
    Therefore, $\sym(P_{\mathbf{y}}) = \bigoplus_{i=1}^n y_i \sym(P_{i}) = B^\intercal \mathbf{y}$, where $B^\intercal$ is the check matrix of $\Symp(H)$ from Definition~\ref{def:symplectic_codes}. 
    Consequently, by applying $\calD_H^{(\ell)}$ on registers $A$ and $D$, we map register $A$ to $\ket{\mathbf{0}}$ and may discard it. 
    The remaining state is \begin{align}
        \ket{\psi_5}_{BC} = \frac{1}{\calN} \sum_{j=0}^\ell w_j \sum_{\mathbf{y} \in \F_2^m : |\mathbf{y}| = j} \left( \prod_{i \in \supp(\mathbf{y})} v_i \right) \ket{\sym(P_{\mathbf{y}})}_{BC} .
    \end{align}
    We now undo the coherent Bell measurement, giving the final state \begin{align}
        \ket{\psi_6}_{BC} = \frac{1}{\calN} \sum_{j=0}^\ell w_j \sum_{\mathbf{y} \in \F_2^m : |\mathbf{y}| = j} \left( \prod_{i \in \supp(\mathbf{y})} v_i \right)  (P_{\mathbf{y}} \otimes I)_{BC} \ket{\Phi^n}_{BC} = \frac{1}{\calN} (\calP(H) \otimes I) \ket{\Phi^n}_{BC} .
    \end{align}
    In the eigenbasis of $H$, the maximally entangled state has the representation \begin{equation}
        \ket{\Phi^n}_{BC} = \frac{1}{2^{n/2}}\sum_\lambda \ket{\lambda}_B\otimes \ket{\bar{\lambda}}_C 
    \end{equation} 
    where $\ket{\bar{\lambda}}$ is defined by $\bra{\mathbf{x}}\ket{\bar{\lambda}} = \overline{\bra{\mathbf{x}}\ket{\lambda}}$ and $\overline{z}$ denotes the complex conjugate of $z \in \C$.
    This is a consequence of the unique property of the maximally entangled state, namely that for any matrix $M$, $(M \otimes I) \ket{\Phi^n} = (I \otimes M^\intercal) \ket{\Phi^n}$.
    Hence, in the eigenbasis representation, \begin{equation}
        (\calP(H) \otimes \Id)  \ket{\Phi^n} = \frac{1}{2^{n/2}}\sum_{\lambda} \calP(\lambda) \ket{\lambda}_B\otimes \ket{\bar{\lambda}}_C .
    \end{equation}
    By tracing out register $C$, we obtain $\rho_{\calP}(H)$. 
    Up to the oracle call, each step is accomplished in $\poly(n)$ time.
\end{proof}

\begin{algorithm}[H]        
  \caption{Hamiltonian DQI for Commuting Hamiltonians}
  \label{alg:amplitude_transform}
  \begin{algorithmic}[1]    
  \Require{$H = \sum_{i=1}^m v_iP_i = \sum_i \lambda_i \ketbra{\lambda_i}{\lambda_i}$ (commuting Pauli Hamiltonian on $n$ qubits), $|\Phi^n\rangle$ (maximally entangled state of $2n$ qubits), $\calD^{(\ell)}_H$ (decoding oracle), $\calP$ (polynomial of degree $\ell$).}
    \Ensure{$\rho_{\calP}(H) \propto \calP^2(H)$ ($n$-qubit mixed state).}
      \State Decompose $\calP(\sum_{i=1}^m z_i)$ as a sum of reduced elementary symmetric polynomials in the variables $z_i = v_i P_i$ by applying Theorem~\ref{thm:commuting_symmetric_polynomial_expansion}, obtaining
      \begin{equation}
          \calP\left( \sum_{i=1}^m z_i \right) = \sum_{j=0}^\ell w_j \sum_{\mathbf{y} \in \Z_2^{m} \,:\, |\mathbf{y}| = j} z_1^{y_1} \cdots z_m^{y_m} .
      \end{equation}
      \State Prepare the $m$-qubit pilot state in register $A$, using Lemma~\ref{lemma:commuting_efficient_resource_state}.
      \begin{equation}
         \ket{\calR_{\text{comm}}^{\ell}(H)} \propto \sum_{j=0}^\ell w_j \sum_{\mathbf{y}\in \mathbb{F}_2^m:|\mathbf{y}|=j} \ket{\bf y}_A .
      \end{equation}
      Append this state to $\ket{\Phi}_{BC}$.
      \State Apply $\bigotimes_{i=1}^m Z_i^{\frac{1-v_i}{2}}$ on the $A$ register. 
      Apply $\prod_{i=1}^m (C^{(i)} P_i)_{A\rightarrow B}$, where each operation is a Pauli $P_i$ on register $B$, controlled by the $i$th qubit of register $A$.
      \State Apply the coherent Bell measurement to registers $BC$, mapping the Bell basis to its symplectic representation as specified in Theorem~\ref{thm:app-Bell_measurement}.
      \State Uncompute register $A$ by applying $\calD_{H}^{(\ell)}$ to registers $A$ and $BC$. 
      Undo Step 4 by applying the inverse coherent Bell measurement to registers $BC$.
      \State The reduced density matrix on $B$ now equals $\rho_{\calP}(H)$. Output register $B$.
  \end{algorithmic}
\end{algorithm}

\subsection{Robustness against decoding failures}
\label{sec:decoding_failures}
A natural question arising from Theorem~\ref{thm:commuting_reduction_to_decoding} is whether the reduction is robust against decoding imperfections. 
More precisely, if the decoding oracle $\calD^{(\ell)}_H$ can correctly decode all but a $\epsilon$ fraction of errors up to weight $\ell$, then the state produced $\rho_{\calP, \epsilon}(H)$ should be close to the desired state $\rho_{\calP, \epsilon}(H)$.
We next prove that this robustness indeed holds, at least when all errors of weight $\ell$ are information-theoretically correctable.

\begin{theorem}[Robustness of commuting reduction] \label{thm:robustness_commuting}
Let $H = \sum_{i=1}^m v_i P_i$ be a commuting Hamiltonian, where $v_i \in \set{\pm 1}$ and the $P_i$ are $n$-qubit Paulis. 
Let $\calP(x)$ be an arbitrary polynomial of degree at most $\ell$.
Suppose that the Hamiltonian symplectic code $B^\intercal = \Symp(H)$ has distance at least $2 \ell + 1$,
and that there is a decoding oracle $\calD^{(\ell, \epsilon)}_H$ which for each $j$ from $1$ to $\ell$ satisfies $\calD^{(\ell, \epsilon)}_H \ket{\mathbf{y}} \ket{B^\intercal \mathbf{y}} = \ket{0} \ket{B^\intercal \mathbf{y}}$ for at least a $1 - \epsilon$ fraction of all errors $\mathbf{y}$ with weight $j$.
On the remaining $\epsilon$ fraction, $\calD^{(\ell, \epsilon)}_H$ may behave arbitrarily so long as unitarity is maintained.
Then the state $\rho_{\calP, \epsilon}(H)$ produced by Algorithm~\ref{alg:amplitude_transform} which uses $\calD^{(\ell, \epsilon)}_H$ in place of a perfect decoding oracle $\calD^{(\ell)}_H$, satisfies \begin{align}
    \calF(\rho_{\calP, \epsilon}(H), \rho_{\calP}(H)) \geq 1 - \epsilon ,
\end{align}
where $\calF$ is the fidelity, and \begin{align}
    \frac{1}{2} \lVert \rho_{\calP, \epsilon}(H) - \rho_{\calP}(H) \rVert_1 \leq \sqrt{2 \epsilon} ,
\end{align}
where $\frac{1}{2} \lVert \cdot \rVert_1$ is the trace distance.
\end{theorem}
\begin{proof}
    By Uhlmann's theorem, \begin{align}
        \calF(\rho_{\calP, \epsilon}(H), \rho_{\calP}(H)) \geq \langle \psi_{\calP, \epsilon}(H) | \psi_{\calP}(H) \rangle ,
    \end{align}
    where \begin{align}
        \ket{\psi_{\calP}(H)} & = \frac{1}{\calN} \sum_{j=0}^\ell w_j \sum_{\mathbf{y} \in \F_2^m : |\mathbf{y}| = j} \ket{0^n} \otimes  (v_{\mathbf{y}} P_{\mathbf{y}} \otimes I) \ket{\Phi^n} , \\
        \ket{\psi_{\calP, \epsilon}(H)} & = \frac{1}{\calN} \sum_{j=0}^\ell w_j \sum_{\mathbf{y} \in \F_2^m : |\mathbf{y}| = j} \ket{\mathbf{z}(\mathbf{y})} \otimes  (v_{\mathbf{y}} P_{\mathbf{y}} \otimes I) \ket{\Phi^n} 
    \end{align} 
    are purifications of the states in question.
    Here $\mathbf{z}(\mathbf{y}) = 0$ for a $(1-\epsilon)$ fraction each of the sets $\set{\mathbf{y} \in \F_2^m \,:\, |\mathbf{y}| = j}$ (for $j$ from $1$ to $\ell$), and otherwise can be an arbitrary bitstring in $\F_2^m$.
    We have defined $v_{\mathbf{y}} = \prod_{i \in \supp(\mathbf{y})} v_i \in \set{\pm 1}$ in the same manner as $P_{\mathbf{y}} = \prod_{i \in \supp(\mathbf{y})} P_i$.
    \begin{align}
         \langle \psi_{\calP, \epsilon}(H) | \psi_{\calP}(H) \rangle & = \frac{1}{\calN^2} \sum_{j, k = 0}^{\ell} w_j w_k \sum_{\substack{|\mathbf{y}| = j \\|\mathbf{x}| = k}} \langle 0^n | \mathbf{z}(\mathbf{y}) \rangle v_{\mathbf{x} \oplus \mathbf{y}} \langle \Phi^n | P_{\mathbf{x}} P_{\mathbf{y}} \otimes I | \Phi^n \rangle \\
         & = \frac{1}{\calN^2} \sum_{j, k = 0}^{\ell} w_j w_k \sum_{\substack{|\mathbf{y}| = j \\|\mathbf{x}| = k}} \langle 0^n | \mathbf{z}(\mathbf{y}) \rangle v_{\mathbf{x} \oplus \mathbf{y}} \delta_{\mathbf{x}, \mathbf{y}} \\
         & = \frac{1}{\calN^2} \sum_{j = 0}^{\ell} w_j^2 \sum_{|\mathbf{y}| = j} \langle 0^n | \mathbf{z}(\mathbf{y}) \rangle \\
         & = \frac{1}{\calN^2} \sum_{j = 0}^{\ell} w_j^2 (1-\epsilon) \binom{m}{j} \\
         & = 1 - \epsilon .
    \end{align}
    In the second line, we used the fact that $P_{\mathbf{x}} P_{\mathbf{y}} = I$ if and only if $\mathbf{x} = \mathbf{y}$, since the weights are each below half the distance of $\Symp(H)$.
    Then, unless $P_{\mathbf{x}} P_{\mathbf{y}} = I$, the inner product $\langle \Phi^n | P_{\mathbf{x}} P_{\mathbf{y}} \otimes I | \Phi^n \rangle$ vanishes due to the orthonormality of the Bell basis (see Section~\ref{app:bell_basis}).
    The normalization is given by $\calN^2 = \sum_{j=0}^\ell w_j^2 \binom{m}{j}$, as shown in Lemma~\ref{lemma:commuting_efficient_resource_state}. 
    Application of this fact yields the final line.
    By the Fuchs-van de Graaf inequality~\cite{fuchs2002cryptographic}, \begin{align}
        \frac{1}{2} \lVert \rho_{\calP, \epsilon}(H) - \rho_{\calP}(H) \rVert_1 \leq \sqrt{1 - \calF^2(\rho_{\calP, \epsilon}(H), \rho_{\calP}(H))} = \sqrt{1 - (1 - \epsilon)^2} \leq \sqrt{2 \epsilon} .
    \end{align}
\end{proof}

We remark that since Decoded Quantum Interferometry (DQI) is a special case of HDQI, a natural question is the extent to which HDQI strictly generalizes DQI.
In Appendix~\ref{sec:app-cDQI_simulation}, we show that at the level of eigenvalue sampling, the reduction in Theorem~\ref{thm:commuting_reduction_to_decoding} can be simulated by calls to an oracle implementing DQI, efficient classical computation, and Clifford operations.
This simulation holds only if $H$ is a commuting Hamiltonian, and even then does not full recover the scope of HDQI. 
It does not assert, for example, that DQI directly
enables estimation of general non-Clifford observables on the mixed state $\rho_{\mathcal{P}}(H)$. 
It also does not recover the preparation of coherent states like thermofield doubles.

\section{Applications of HDQI to commuting Hamiltonians} \label{sec:commuting-applications}
As first applications of HDQI, we instantiate the reduction of Theorem~\ref{thm:commuting_reduction_to_decoding} to certain commuting Hamiltonians whose symplectic codes $\Symp(H)$ have high distance and are efficiently decodable.
There are two conceivable approaches to search for target Hamiltonians that could be amenable to HDQI. First, we could construct a Hamiltonian directly (e.g. by studying physical Hamiltonians) and show that its symplectic code is efficiently decodable to high distance.
Alternatively, we could start with a classical code (corresponding to a classical $Z$-type Hamiltonian) with known high distance, and apply a transformation to it which preserves both distance and symplectic inner product in hopes of obtaining an interesting Hamiltonian.
To decode, we could reverse this transformation and apply the known decoder for the classical code.
The latter strategy does not yield any interesting quantum Hamiltonians, in the sense that any Hamiltonian produced by this approach is related to the original classical Hamiltonian by a simple change of basis. 
We prove this claim more formally in Appendix~\ref{app:symplectic_extension}.
As a consequence, we focus exclusively on the former approach.
We specifically consider two models: an explicit family which we call \emph{nearly independent Hamiltonians} that encompasses many physical Hamiltonians, and a random local commuting Hamiltonian, which may be a suitable model in the study of quantum spin glasses.

\subsection{Nearly independent Hamiltonians} \label{sec:HDQIperfectlyDecodable}
Many interesting Hamiltonians correspond to symplectic codes that are easy to decode. We now show that for a certain family of physically-motivated Hamiltonians (examples of which are given in Section~\ref{tab:constant-relations} and Section~\ref{sec:examples_constant_relations}), HDQI efficiently prepares the Gibbs state at \emph{any} inverse temperature $0 \leq \beta \leq \infty$. In fact, HDQI can prepare arbitrary (normalized) positive functions of the Hamiltonian, which also includes any exact microcanonical ensemble, band‑pass filter, or sector‑projected ensemble. We call such Hamiltonians \emph{nearly independent}.
\THMdecodable*
This task has been the focus of prior work using quantum algorithms (see Section~\ref{sec:examples_constant_relations}) and in some cases classical algorithms were known.
Surprisingly, we show in Section~\ref{sec:perfect_classical} that an even more general task---namely with $k = O(\log n)$---can be achieved with a purely classical algorithm.

For concreteness, we first discuss HDQI specifically on the Hamiltonian of a 2D toric code, which is nearly independent. 
We also give an array of further examples of nearly independent Hamiltonians.
We then conclude by proving the above theorem in full generality.
In what follows, we apply the following lemma, which is a direct consequence of the fact that degree-$\ell$ polynomials are uniquely specified by $\ell+1$ points.
\begin{lemma}[Exact polynomial interpolation] \label{lem:interpolation}
Let $H = \sum_{i=1}^m v_i P_i$ be a commuting Pauli Hamiltonian with $v_i \in \{\pm1\}$. For any non-negative function $f: \mathbb{R} \rightarrow \mathbb{R}_{\geq 0}$  there exists a real polynomial $\calP$ with $\operatorname{deg} \calP \leq m$ such that $\calP^2(H)=f(H)$. 
\end{lemma} \begin{proof}
    $H$ has at most $m+1$ distinct eigenvalues $\{\lambda_i \}_{i = 0}^m$ because every eigenvalue is necessarily an integer between $-m$ and $m$ with the same parity as $m$.
    Hence, any function $f(H)$ can be exactly interpolated by a degree-$m$ polynomial.
    If moreover $f$ is non-negative, then it has a unique point-wise square root $\sqrt{f}$. Choosing $\calP$ to be the exact interpolation of $\sqrt{f}$ yields the statement.
\end{proof}

\subsubsection{Case study: 2D toric code}
Consider the 2D toric code Hamiltonian
\begin{equation}
    H_{\mathrm{TC}} =- \sum_v A_v- \sum_p B_p
\end{equation}
on a $L \times L$ periodic lattice with vertex and plaquette stabilizers \begin{equation}
A_v=\prod_{i \in v} X_i, \qquad B_p=\prod_{i \in p} Z_i.
\end{equation}
The edges of the lattice correspond to qubits and there are $n = |E| = 2L^2$ of them. There are $L^2$ vertices and equally many plaquettes, hence a total number of $m = 2L^2$ terms in the Hamiltonian $H_\mathrm{TC}$. Only $2L^2 - 2$ of them are independent, which is why the toric code supports 2 logical qubits. 

The fact that $H_{\text{TC}}$ corresponds to a quantum error-correcting code bears no relevance here.
Instead focus on its corresponding symplectic code $\Symp(H_{\text{TC}})$. 
As $H_{\text{TC}}$ is a commuting Hamiltonian, HDQI for $H_{\text{TC}}$ will be efficient if $\Symp(H_{\text{TC}})$ is efficiently decodable up to the desired number of errors. $\Symp(H_{\text{TC}})$ is specified by a $2n \times m$ parity check matrix in $\F_2$ whose columns each have weight 4. 
The dimension of the symplectic code is the number of dependencies among toric-code stabilizers. 
There are two such dependencies: the product of all vertex stabilizers is $1$, as is the product of all plaquette stabilizers.
\begin{equation}
\label{eq:TC_relations}
    \prod_{v\in V} A_v = \prod_{p \in F} B_p = \Id. 
\end{equation} 
There are therefore two corresponding logical bits.
The same dependencies show that the distance of $\Symp(H_{\text{TC}})$ is $L^2 = n/2 = m/2$. 

By Lemma~\ref{lem:interpolation}, to prepare arbitrary positive functions of the toric code Hamiltonian, it thus suffices to prepare arbitrary degree $2L^2$ polynomials. 
However, the distance is only $L^2$, and thus it is information-theoretically impossible to correct more than all errors with up to any weight beyond $L^2/2$. 
We circumvent this problem by manually removing the linear dependencies, which significantly boosts the distance, and then adding the dependencies back in a different way such that the distance of the symplectic code does not degrade.

To build intuition for the technique of dependency removal, consider the \emph{defected} 2D toric code where the stabilizers have been removed for a single arbitrary fixed vertex $v_0$ and plaquette $p_0$. 
The corresponding Hamiltonian 
\begin{equation}
    H_{\mathrm{TC}}' =- \sum_{v \not = v_0} A_v- \sum_{p \not = p_0} B_p
\end{equation} 
now has no dependencies between the individual stabilizer terms. 
In other words, the symplectic code $\Symp(H')$ (with block length $m-2 = 2(L^2-1)$) has no non-zero codewords. Thus any syndrome corresponds to a unique error, and \emph{any} error on this code can be efficiently decoded via Gaussian elimination. The number of correctable errors is thus $\ell = 2(L^2-1)$, which implies via Lemma~\ref{lem:interpolation} that HDQI efficiently prepares any non-negative function of $H'_{TC}$. 

We now show how to modify HDQI so as to prepare arbitrary functions even for the 2D toric code.
To prepare arbitrary non-negative functions of $H_{\text{TC}}$, we decode errors not on $\Symp(H_{\text{TC}})$, but on a slightly smaller code that is \emph{completely} independent.
More precisely, we use the relations in Eqn.~(\ref{eq:TC_relations}) to substitute the variables $A_{v_0}$ and $B_{v_0}$ in the Hamiltonian $H_{\text{TC}}$ with \begin{equation}
    A_{v_0} \mapsto \prod_{v \not = v_0} A_v, \quad B_{p_0} \mapsto \prod_{p \not = p_0} B_p.
\end{equation} 
This does not change the Hamiltonian, but it is now written as a function of $m-2$ independent stabilizers. 
Any polynomial $\calP(H_{\text{TC}})$, which we interpret as a multilinear function in the $m$ variables $A_{v_0}, \dots , A_{v_{L^2-1}}, B_{p_0}, \dots , B_{p_{L^2-1}}$ can be re-expressed as a different multilinear function $\calP'$ in the $m-2$ variables $A_{v_1}, \dots , A_{v_{L^2-1}}, B_{p_1}, \dots , B_{p_{L^2-1}}$. 
The primary issue is that $\calP'$ no longer admits a simple interpretation as a univariate polynomial acting on a sum of variables, which upends the usage of the symmetric polynomial expansion from Theorem~\ref{thm:commuting_symmetric_polynomial_expansion}.
If this can be dealt with by modifying the exact pilot state prepared by HDQI, then we would be able to decode using the symplectic code $\Symp(H'_{\text{TC}})$ instead of $\Symp(H_{\text{TC}})$.
Since $\Symp(H'_{\text{TC}})$ has maximal distance and is therefore decodable by Gaussian elimination, we can prepare arbitrary degree-$(m-2)$ polynomials in $A_{v_1}, \dots , A_{v_{L^2-1}}, B_{p_1}, \dots , B_{p_{L^2-1}}$, which also includes any degree $m$ polynomial in the original variables $A_{v_0}, \dots , A_{v_{L^2-1}}, B_{p_0}, \dots , B_{p_{L^2-1}}$. By Lemma~\ref{lem:interpolation}, this is enough to prepare arbitrary non-negative functions of $H_{\text{TC}}$.

\subsubsection{Examples of nearly independent Hamiltonians}
\label{sec:examples_constant_relations}
The above strategy can be generalized to any Hamiltonian whose symplectic code has constant dimension, resulting in Theorem~\ref{thm:arbitrary_temperature}. 
Before we provide the full proof in Section~\ref{sec:proof_of_arbitrary_temperature}, we list notable examples of such Hamiltonians. A summary is also given in Table~\ref{tab:constant-relations}.
We emphasize that many of the Hamiltonians which we describe below correspond to widely studied quantum stabilizer codes.
However, their candidacy as nearly independent Hamiltonians depends on the logical dimension of their \emph{symplectic codes}, which has no direct relation to their logical dimension as quantum stabilizer codes.
Instead, the number of logical bits in $\Symp(H)$ is instead equal to the number of \emph{dependent stabilizers} in the description of the quantum code.

\paragraph{1D Ising model (repetition code) on a ring.}
On $L$ qubits with periodic boundary conditions,
$H=-\sum_{i=0}^{L-1} Z_i Z_{i+1}$.
There is a single relation
$\prod_{i=0}^{L-1} (Z_i Z_{i+1}) = \Id$, and no other product of a nonempty
proper subset equals $\Id$. Thus the symplectic code of this model has dimension $k=1$.
(With open boundary conditions there is no relation, i.e., $k=0$, so the Hamiltonian is completely independent and our theorem also applies.)


\paragraph{2D surface/toric codes on arbitrary closed cellulations.}
Place qubits on edges $e$ of a closed 2D cell complex $\Sigma$, with
vertex checks $A_v = \prod_{e\ni v} Z_e$ and
plaquette checks $B_p = \prod_{e\in\partial p} X_e$~~\cite{kitaev2003anyons}.
On any closed, connected $\Sigma$ there are exactly two relations,
$\prod_v A_v=\Id$ and $\prod_f B_f=\Id$, and no others.
Thus $k=2$, independent of lattice size or curvature,
including Euclidean tilings (which are translation invariant (TI)) and hyperbolic tilings (non-TI).
The same conclusion holds for prime-dimensional qudits
(and, more generally, 2D Kitaev Abelian $D(\mathbb{Z}_q)$ quantum double models with Pauli stabilizers):
on closed $\Sigma$ there are two global relations, $k=2$.

\paragraph{Homological CSS codes from two–term chain complexes.}
More generally, for a cellular
two–term chain complex
\begin{align}
    C_2 \xrightarrow{\;\partial_2\;} C_1 \xrightarrow{\;\partial_1\;} C_0,
    \qquad \partial_1\partial_2=0,
\end{align}
with qubits on $C_1$ (edges), $X$–checks on $C_2$ (faces), and $Z$–checks on $C_0$
(vertices), the number of independent
multiplicative relations among the checks is
\begin{align}
    k \;=\; \dim H_0(\,\cdot\,;\F_2) \;+\; \dim H_2(\,\cdot\,;\F_2),
\end{align}
i.e., the sum of the $\F_2$–Betti numbers in degrees $0$ and $2$ (for manifolds
with boundary, the corresponding relative groups apply). In particular, for any
\emph{closed connected} surface, $k=2$ (one global product among all vertex
checks; one among all face checks), independently of the tessellation; with
boundaries, $k$ remains $O(1)$ (often $k=0$ on a simply connected planar patch).
Thus this homological formulation provides a compact and very general topological
criterion for constant $r$ that covers arbitrary (even non‑TI) tessellations and
all standard boundary types (toric, cylindrical, Möbius/Klein, planar, etc.).

\paragraph{2D color codes (stabilizer version).}
On a 3-valent, 3-colorable tiling of a closed surface, with qubits on
vertices and two checks per face $X_f=\prod_{v\in\partial f}X_v$,
$Z_f=\prod_{v\in\partial f}Z_v$ (see, e.g.,~\cite{bombin2007optimal}),
each Pauli sector has the three ``two-color'' products equal to $\Id$;
exactly two of them are independent. Hence $k=4$ in total (two in
the $X$-sector and two in the $Z$-sector), again independent of size
or the choice of 3-colorable tiling. 
With boundaries, the number of relations remains $O(1)$ and is
dictated by the boundary pattern (relative homology), typically taking values in a small set (e.g., $k\in\{0,1,2\}$ per sector for standard color‑code patches).

\paragraph{General translation-invariant commuting-Pauli models on closed 2D surfaces.}
By the 2D classification of TI Pauli stabilizer
Hamiltonians (up to local Clifford circuits and addition of ancillas),
any such model is equivalent to a finite stack of $K$ decoupled toric
codes~\cite{Haah2013Modules,haah2011local,bombin2014structure}. 
Consequently,
$k=2K$ is a constant determined by the unit cell, independent of the
system size. This umbrella covers rotated toric codes, Wen's plaquette
model (equivalent to toric code up to local Cliffords), and many other TI variants. 
More generally, any constant number $K$ of decoupled 2D layers (any geometry) has $k=2K$, constant in system size.

\paragraph{Cluster-state Hamiltonians on an arbitrary graph.}
On an arbitrary graph with vertices $V$ and neighbor function $N(v)$, the cluster stabilizers
$K_v=X_v\prod_{u\in N(v)} Z_u$ are always independent, so $k=0$.
This includes any type of $d$-dimensional lattice (with any choice of boundary conditions) and any tesselation of a closed surface.

\paragraph{Haah's cubic code (``Code 1'').}
Haah's code is a three-dimensional stabilizer code without string-like logical operators, and one of the most famous examples of a fracton stabilizer code. 
We focus on the usual CSS ``Code 1" on a periodic $L \times L \times L$ lattice. 
This 3D CSS stabilizer model~\cite{haah2011local} has no \emph{local} stabilizer constraints. 
The ground state degeneracy of Haah's code, i.e. the number of logical \emph{qubits} that the stabilizer code admits, is intricate: unless $L$ is a multiple of $4p - 1$ for $p \geq 2$, this ground state degeneracy is 4 for odd $L$ \cite{haah2013lattice}, in which case this quantum code supports two logical qubits. 
We focus on these models as they are the most relevant for quantum error correction and have been initially proposed as candidates for self-correcting quantum memory.
In particular, the criterion of being nearly independent holds for those $L$ where the number of global relations is constant (i.e. odd $L$ which are not multiples of $4p-1$) \cite{haah2013lattice}.

\paragraph{Non-local deterministic CSS code.}
Fix even $n$ and a bipartition $[n] = A \cup B$ with $|A| = |B| = n/2$.
Define $S_i^{(X)} = \bigotimes_{\ell \neq i} X_\ell$ for $i \in A$ and
$S_j^{(Z)} = \bigotimes_{\ell \neq j} Z_\ell$ for $j \in B$.
Then $[S_i^{(X)},\, S_j^{(Z)}] = 0$ because
$|\supp S_i^{(X)} \cap \supp S_j^{(Z)}| = n - 2$ is even.
The sets $\{S_i^{(X)}\}_{i \in A}$ and $\{S_j^{(Z)}\}_{j \in B}$
are linearly independent over $\mathbb{F}_2$; hence $k = 0$. Every term has weight $n - 1 = \Theta(n)$, so the Hamiltonian is highly non-local.

\paragraph{Non-local random CSS code.}
Let $\alpha\in(0,1)$ be a constant.
Choose according to any desired distribution $d_X=\lfloor \alpha n\rfloor$ high‑weight vectors $x_1,\dots,x_{d_X}\in\F_2^n$ and let $J := \mathrm{span}\{x_1, \dots, x_{d_X}\}$. 
Pick $d_Z=n-d_X-c$ vectors $z_1,\dots,z_{d_Z}\in J^\perp$ according to any desired distribution, optionally adding $c\in\{0,1,2\}$ sector‑wise redundancies by setting one $x$ (resp.\ $z$) to the sum of the others. 
Define $S^{(X)}_i=X^{x_i}$ and $S^{(Z)}_j=Z^{z_j}$. By construction $x_i\cdot z_j=0$, so all X/Z checks commute.
Moreover, with high probability the only linear dependencies are the intentionally added optional ones, giving a \emph{constant} number of relations $k = O(1)$ and minimal dependency relation size $\Theta(n)$ in each affected sector.

\paragraph{Models that \emph{do not} qualify.}
The following well-known commuting-Pauli Hamiltonians fail the constant-relations property because they possess extensive families of \emph{local} multiplicative constraints.
Therefore, they are not nearly independent Hamiltonians.
\begin{itemize}[leftmargin=*]
  \item \textbf{3D and 4D toric code and related $\Z_2$ lattice gauge theories}:
    For the 3D toric code, the product of the six plaquette checks around every cube equals $\Id$,
    yielding $k = \Theta(L^3)$ independent relations. Similarly, for the 4D toric code the product of the 8 parity checks corresponding to the edges attached to a single vertex equals $\Id$.
  \item \textbf{Type-I fracton models} such as the $X$-cube and checkerboard codes: numerous cube- and vertex-local identities imply a constant fraction of dependent checks.
  \item \textbf{Dense graph-Ising models in $d\ge 2$} (e.g., $ZZ$ on
    all edges of a 2D torus): checks along every edge of a closed loop (e.g. every square in 2D) give a dependent relation, so the number of dependencies scales with the system size for any dimension.
\end{itemize}

\subsubsection{HDQI on nearly independent Hamiltonians}
\label{sec:proof_of_arbitrary_temperature}
We now prove Theorem~\ref{thm:arbitrary_temperature}. 
As discussed in the case study of the 2D toric code, our proof strategy is to eliminate the constant number of codewords in $C = \Symp(H)$ to obtain a new code $C'$ with block length $m-k$, for which any error up to the maximal possible weight $m-k$ can be efficiently decoded. 
We then express any polynomial of $H$ in these $m-k$ new variables, and prove that it is still block-symmetric within a constant number of blocks.
Therefore, instead of expanding into a single sum of symmetric polynomials, we will expand independently within each block.
Operationally, this affects HDQI in that the pilot state instead must become \emph{the product of a constant number of Dicke states}, rather than a single Dicke state. 
Our proof makes use of the following theorem about symmetric expansions of polynomials, of which we give a self-contained proof in Appendix~\ref{app:block_expansion}.
Here, the degree-$j$ elementary symmetric polynomial in variables $x_1, \dots, x_s$ is given by $e_{j}(x_1, \dots, x_s) = \sum_{1 \leq b_1 <  \dots <  b_{j} \leq s} x_{b_1} \cdots x_{b_{j}} = \sum_{\mathbf{y} \in \F_2^s : |\mathbf{y}| = j} x_1^{y_1} \cdots x_s^{y_s}$.

\begin{restatable}[Blockwise symmetric expansion for constant $k$]{theorem}{blockexp}\label{thm:blockwise_decomposition_theorem}
Let $\calP$ be a univariate polynomial of degree $m$ and $z_i$ formal variables satisfying $z_i^2 = 1$.
Moreover, assume that there are exactly $k$ independent relations among the $z_i$, i.e. $k$ independent subsets of the $z_i$ multiply to the identity.
Assume that the first $d \leq n$ variables are independent.
Let $S := \sum_{i=1}^m z_i$.
There exists a partition of these $d$ variables into $r$ blocks
$V_1,\dots,V_r$ with $r\le 2^k$ such that:
\begin{enumerate}
\item[(i)] $S$ can be written in the form
\begin{equation}\label{eq:S-decomposition}
    S\ =\ \sum_{t=1}^r T_t\ +\ \sum_{j=1}^k p_{U_j},
    \qquad
    T_t:=\sum_{i\in V_t} z_i,\quad
    p_{U_j}:=\prod_{i\in U_j} z_i,
\end{equation}
where each $U_j$ is a union of full blocks $V_t$.
\item[(ii)] For $0\le a_t\le d_t:=|V_t|$, let $e_{a_t}(V_t)$ denote the degree-$a_t$
elementary symmetric polynomial in the variables of $V_t$, and set
\begin{align}
    E_{\boldsymbol a}\ :=\ \prod_{t=1}^r e_{a_t}(V_t),\qquad
\boldsymbol a=(a_1,\dots,a_r).
\end{align}
Then $\calP(S)$ admits a unique expansion
\begin{equation}\label{eq:block-expansion}
    \calP(S)\ =\ \sum_{\mathbf{a} \,:\, 0\le a_t\le d_t}\ \gamma_{\boldsymbol a}\,E_{\boldsymbol a},
\end{equation}
i.e.\ it belongs to the subspace of polynomials that are separately symmetric within each block.
\item[(iii)] If $k$ is a fixed constant, the coefficients
$\{\gamma_{\boldsymbol a}\}$ can be computed using $O\!\big((r+k)D\,\ell\big)$ arithmetic
operations and $O(D)$ memory, where
\begin{align}
    D\ :=\ \prod_{t=1}^r (d_t+1)
\ \ \leq\ \ \Big(\tfrac{d}{r}+1\Big)^{\!r}
\ \ =\ d^{O(1)} = n^{O(1)}
\end{align}
since $r\le 2^k=O(1)$.
Consequently, for $k = O(1)$ the expansion in Eqn.~(\ref{eq:block-expansion}) is computable in time $O(\poly(n, \ell))$.
\end{enumerate}
\end{restatable}

\begin{proof}[Proof of Theorem~\ref{thm:arbitrary_temperature}]
Let $H = \sum_{i=1}^m v_i P_i$ be a Hamiltonian for which $k =\dim C = \dim \Symp(H)$ is constant, and let $\calP$ be any non-negative function of $H$. 
By Lemma~\ref{lem:interpolation}, we can assume without loss of generality that $\calP$ is a degree-$m$ polynomial. 
We identify the Paulis $v_i P_i$ with formal variables $z_i \in \{\pm1\}$, with the relations $z_i^2 = 1$ for all $i$.
Any independent basis of $\Symp(H)$ gives us $k$ (mod-$2$ independent) product relations
\begin{equation}
    \prod_{i\in B_j} z_i \;=\; 1\qquad (j=1,\dots,k),
\end{equation}
where $B_j\subseteq[m]$ and ``independence'' means that no relation is the
(mod-$2$) product of the others. 
We interpret the univariate polynomial  $\calP(x) = \sum_{j=0}^{m} c_j x^j$ applied to $S = \sum_{i=1}^m z_i$ as a multilinear function in $z_1, \dots, z_m$. 
By Gaussian elimination over $\F_2$ we may eliminate $k$ variables and rewrite
every $z_i$ as a monomial in $d := m - k$ \emph{independent} variables. 
After this reduction has been performed, we will denote the independent variables by $z_1,\dots,z_d$. 

We now apply Theorem~\ref{thm:blockwise_decomposition_theorem}.
Let $\{\gamma_{\ba}\}$ for $\ba=(a_1,\dots,a_r)$ and $0\le a_t\le d_t$ be the unique coefficients of the expansion of $\calP(S)$ into block-symmetric elementary polynomials on the $d = m-k$ independent variables $z_1, \dots, z_d$. 
That is, \begin{equation}
    \calP(S) =\ \sum_{\boldsymbol a} \gamma_{\boldsymbol a}\,E_{\boldsymbol a},\qquad
E_{\boldsymbol a}\ :=\ \prod_{t=1}^r e_{a_t}\big(\{z_i:i\in V_t\}\big).
\end{equation}
Since $r \leq 2^k=  O(1)$, the coefficient array $\{\gamma_{\boldsymbol a}\}$ can be computed from the coefficients of $\calP$ in time $O(\poly(n,\ell))$
and memory $O\!\big(\prod_t (d_t{+}1)\big) = O(\poly(n))$ because $r = O(1)$.
The block-wise basis elements $E_{\boldsymbol a}$ are products of block-elementary symmetric sums, i.e.
\begin{align}
    E_{\boldsymbol a}\ =\ \prod_{t=1}^r \ \sum_{\substack{J_t\subseteq V_t\\ |J_t|=a_t}} \ \prod_{i\in J_t} z_i .
\end{align}
Our commuting reduction in Theorem~\ref{thm:commuting_reduction_to_decoding}
exploited the decomposition into elementary symmetric polynomials by preparing a pilot state comprised of a superposition over Dicke states.
Now, since the decomposition is instead into a product of $r=O(1)$ elementary symmetric polynomials, we modify the pilot state by instead preparing a superposition (with coefficients $\gamma_{\ba}$) over \emph{products} of $r$ Dicke states of sizes $n_t$ that coherently indexes the pairs $(J_1,\dots,J_r)$.
From there, we apply the corresponding controlled Paulis as before to produce the correct pilot state. 

In more detail, this modified HDQI prepares the pilot state representing $\calP(S) =\ \sum_{\boldsymbol a} \gamma_{\boldsymbol a} E_{\ba}$ by replacing the initial weight-varying Dicke state in Eqn.~(\ref{eq:initial_state}) as:
\begin{equation}
    \sum_{j = 0}^m w_j \ket{D_j^m} = \sum_{j = 0}^m w_j' \sum_{\mathbf{y} \in \F_2^m ,|\mathbf{y}| = j} \ket{\mathbf{y}} \longrightarrow
    \sum_{\boldsymbol a} \gamma_{\boldsymbol a} \ket{D_{\ba}} = \sum_{\boldsymbol a} \gamma_{\boldsymbol a}  \bigotimes_{t = 1}^r \sum_{\substack{J_t\subseteq V_t\\ |J_t|=a_t}} \ket{J_t} .
\end{equation}
Here the state $\ket{J_t}$ is a computational basis state on $|V_t|$ qubits and Hamming weight $|J_t|$, whose non-zero entries indicate the elements of $J_t$.
The state $\ket{D_{\ba}}$ represents a product over $t$ of Dicke states, each over the variables in $V_t$.
With this modified pilot state in hand, we proceed with HDQI as in Section~\ref{sec:commuting}. Performing the associated controlled monomials implements $\sum_{\boldsymbol a}\gamma_{\boldsymbol a}E_{\boldsymbol a}$ up to a known global scaling, after which normalization yields the desired state $\rho_{\calP}(H) \propto \calP(H)$. 
Since $r = O(1)$ and $\deg(\calP) \leq m = \poly(n)$, the number of Dicke registers is constant, each register size is at most $d$, and hence the total HDQI runtime remains $\poly(n)$.
By construction, the reduced code $C'$ of block length $d=m-k$ has dimension $0$, so every error on $C'$ is trivially decodable up to the maximal weight $d$; this is the only requirement in Theorem~\ref{thm:commuting_reduction_to_decoding} for the reduction to succeed.
\end{proof}

\subsection{Random \texorpdfstring{$k$-}\ local commuting Hamiltonians}
\label{sec:random_local_commuting}
In this section, we study a natural class of random, commuting $k$-local Hamiltonians; namely, we consider
\begin{equation}
    H\;=\;\sum_{i=1}^{m} v_i P_i, \qquad v_i \sim \set{\pm 1},  P_i \in \Pauli_n^{(k)}:=\{\text{$k$--local Paulis on $n$ qubits}\}.
\end{equation}
That is, $H$ is comprised of a sum of  $m=\Theta(n)$ signed random $k$-local Pauli terms, subject to the constraint that all terms commute.
Such a Hamiltonian can be generated by an iterative rejection sampling algorithm, given in Algorithm~\ref{alg:greedy-commuting-sampler}.
It is not immediately clear that this algorithm is efficient; however, our numerics suggest that it does in fact terminate quickly, even for large instances. 
We leave analysis of the precise runtime of the greedy sampler in Algorithm~\ref{alg:greedy-commuting-sampler} to future work.

\begin{algorithm}[H]        
  \caption{Greedy sampling of random $k$-local commuting Pauli operators}
  \label{alg:greedy-commuting-sampler}
  \begin{algorithmic}[1]    
  \Require{Integers $n,k$ and $m$.}
    \Ensure{List $\mathcal{L} = \{P_1,\dots,P_m\}$ of $m$ commuting $n$-qubit Pauli operators of weight $k$.}
\State Initialize $\mathcal{L} = \emptyset$.
\State Repeat until $\mathcal{L}$ contains $m$ commuting Pauli operators: 
\begin{itemize}
    \item Draw a random $n$-qubit Pauli $P\sim \Pauli_n^{(k)}$ of weight $k$ (repeat if $P$ is already in $\mathcal{L}$).
    \item \textbf{If} $[P, Q]=0$, for all $Q \in \mathcal{L}$, add $P$ to $\mathcal{L}$; otherwise discard $P$.
\end{itemize}
  \end{algorithmic}
\end{algorithm}

In this setting, since $H$ is not nearly independent in general, we will be interested in preparing approximate Gibbs states by choosing a good approximating polynomial.
As discussed in Section~\ref{sec:overview}, HDQI reduces the preparation of the Gibbs state $\propto \exp{-\beta H}$ for any such Hamiltonian to the decoding of the symplectic code $\Symp(H) = \{\mathbf{u} \in \mathbb{F}_2^{m} : B^\intercal \mathbf{u} = 0\}$, where $B^\intercal \in \mathbb{F}_2^{2n \times m}$ is a parity check matrix in which the $i\th$ column is the symplectic representation of the Pauli group element $P_i$. The achievable $\beta$ depends on the degree of $\calP$ admissible to HDQI, which in turn depends on the maximum weight of errors on $\Symp(H)$ that are efficiently correctable.
The precise relation between $\beta$ and the correctable error weight is given by Theorem~\ref{thm:gibbs_explicit}, which we re-state below.
\Gibbsthm*
In addition, choosing an optimal form of the polynomial derived in \cite{JSW24}, HDQI samples high energy states of this random commuting model. The achievable approximation ratio is given by a semicircle law, which we also re-state below.
\semicircle*
\begin{proof}
This is a direct corollary of Theorem 4.1 from \cite{JSW24} and Lemma~\ref{lem:simul-diag}, which relates HDQI on a commuting Hamiltonians to DQI on a classical Hamiltonian. 
$\Symp(H)$ remains unchanged in the new basis, so its distance remains the same.
\end{proof}
For commuting Pauli Hamiltonians $H = \sum_{i=1}^m v_i P_i$ (noting that we adopt an extensive normalization convention, i.e. $\lVert H \rVert = \Theta(m)$), HDQI therefore prepares Gibbs states at constant temperature $\beta$ provided the symplectic code $\Symp(H)$ can be efficiently decoded up to linearly many errors. 
On certain ensembles of low-density parity check (LDPC) codes, local algorithms like Sipser-Spielman or Belief Propagation efficiently decode linearly many errors. 

By assumption, $H$ is a $k$-local Hamiltonian, which implies that the parity check matrix $B^\intercal$ of $\Symp(H)$ has at most $2k$ nonzero entries in any column. 
Since each column is essentially random, up to commuting constraints, each row is also typically very sparse, and hence $\Symp(H)$ is an LDPC code. Consequently, Belief Propagation can be expected to successfully recover the original codeword even when some constant fraction of its $m$ bits have been flipped. 
For some standard ensembles of sparse parity check matrices, the ability of Belief Propagation-type algorithms to decode errors up to a constant fraction of maximum weight has been analytically proven \cite{gallager1963low,richardson2008modern}. However, the ensemble of parity check matrices that arises by sampling a random commuting $k$-local Hamiltonian and then writing down its symplectic representation is not standard, and the constraint of pairwise commutation obstructs standard analytical techniques for lower-bounding the distance. 
Therefore, we have carried out computer experiments with Belief Propagation to verify empirically that our ensemble of LDPC codes shares the property with generic LDPC codes of correcting a constant fraction of bit flips. 
Further, such computer experiments yield quantitative information about what this fraction is in practice. 
In turn, this fraction determines the minimum temperature (i.e. maximum $\beta$) for which HDQI is able to efficiently produce Gibbs states via Belief Propagation decoding. 
The code and data used to generate the below results are publicly available~\cite{Schmidhuber_HDQI}.

\subsubsection{Numerical results}

We here describe some optimizations we made to our Belief Propagation decoder so as to tailor it for our model.
Relative to standard ensembles of LDPC code parity check matrices, such as Gallager's ensemble~\cite{gallager1962ldpc}, the ensemble arising here has the unusual feature that it typically contains a substantial number of linear dependencies between the parity checks. 
In other words, $B^\intercal$ frequently has rank lower than $\min(m, 2n)$. 
Experience suggests that Belief Propagation decoding works best in the absence of linear dependencies, and also works best when the parity checks are as sparse as possible. 
Therefore, we use the following decoding procedure, which we found to consistently correct a slightly larger number of bit flips than standard Belief Propagation decoding.

\begin{enumerate}
    \item Eliminate parity checks one by one, starting from those with highest Hamming weight, until the remaining parity check matrix has full rank.
    \item Apply standard Belief Propagation decoding to the resulting parity check matrix.
\end{enumerate}

\begin{figure}[t!]
    \begin{center}
    \includegraphics[width=\textwidth]{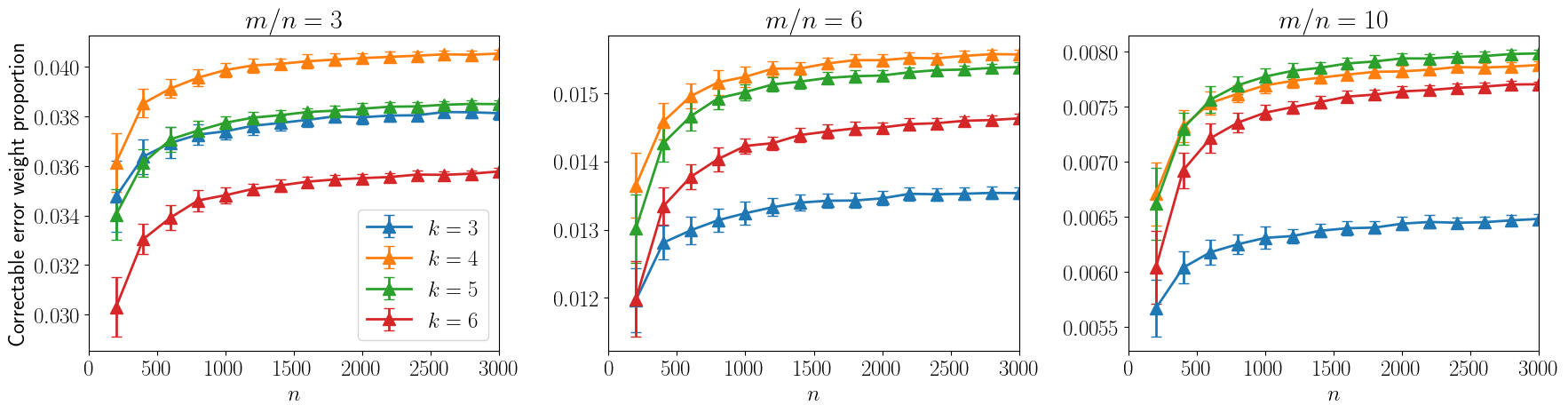}
    \end{center}
    \caption{\label{fig:bp_plot} 
    Maximum fraction of randomly-selected bits that can be flipped such that our modified Belief Propagation decoder remains able to recover the original codeword with probability $\geq 1/2$.
    For each set of parameters we independently sample 50 commuting Hamiltonians from Algorithm~\ref{alg:greedy-commuting-sampler}. 
    The points and error bars show the mean and standard deviation, respectively. 
    We consider three different ensembles of Hamiltonian with three different ratios $m/n$ of Pauli terms to qubits, and four different choices of $k$. 
    In each case one observes that the error fraction tolerated increases with $n$ and appears to converge toward a constant value as $n \to \infty$, suggesting that the decoder can correct randomly chosen errors up to a threshold linear weight with substantial probability.
    }
\end{figure}

Figure~\ref{fig:bp_plot} shows that the fraction of bits that can be flipped while still recovering the original codeword appears to asymptote to a constant in the limit of large block size. 
This is the same behavior that one sees in generic LDPC codes drawn from standard ensembles. 
The quirks of the ensemble arising from commuting $k$-local Hamiltonians are reflected only in the specific error fraction to which the curve asymptotes.

We next consider the implications of these results for the production of Gibbs states. The correctable error fraction displayed in Fig.~\ref{fig:bp_plot} is the maximum fraction of randomly-selected bits that can be flipped such that our modified Belief Propagation decoder remains able to recover the original codeword with at least $50\%$ probability. 
In an actual application of HDQI one would likely choose the polynomial degree $\ell$ such that $\ell$ bit flip errors can be decoded with high reliability, such as $99\%$. 
We nevertheless take our threshold at $50\%$ for the following reason. 
It is well-known that the success rate Belief Propagation versus number of bit flips displays a sigmoidal behavior at any given block size, but that this converges increasingly toward a step function in the limit of large block size \cite{richardson2008modern,luby2002efficient}, as illustrated schematically in Fig.~\ref{fig:waterfall}. 
Such behavior is frequently illustrated in the decoding under the name of ``waterfall plots.'' Thus, the intersection of the finite-block-size waterfall plot with the $50\%$ line should be a reasonable estimate of the location of the step function asymptotically.

\begin{figure}
    \begin{center}
        \includegraphics[width=0.5\textwidth]{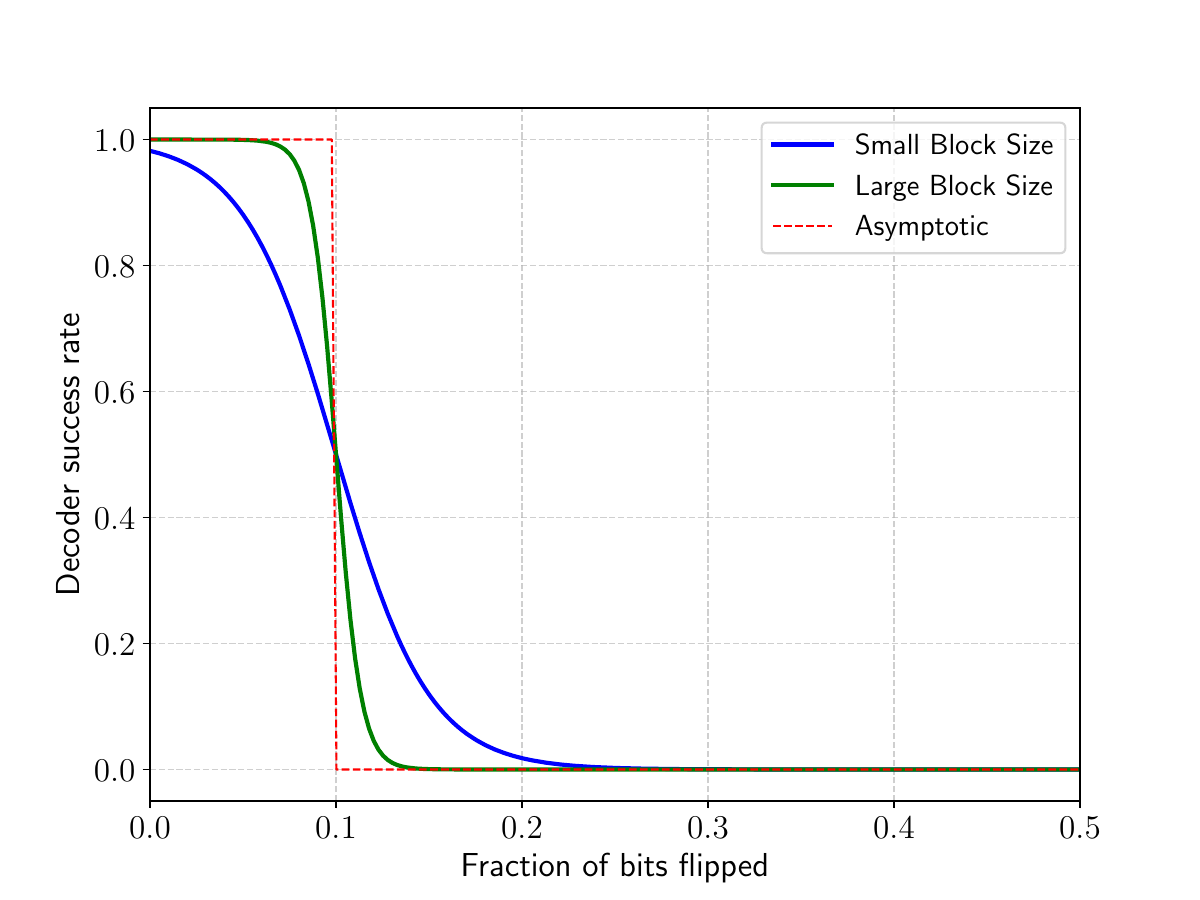}
    \end{center}
    \caption{Schematic of the qualitative behavior of typical ``waterfall plots'' from the decoding literature. Below some threshold error weight, the decoder succeeds with high probability; above the threshold, the decoder fails with high probability. The structure of the success probability curve implies that the exact choice of success probability to define decoder success does not make a significant impact on numerical analysis.}
    \label{fig:waterfall} 
\end{figure}

Next, recall that Theorem~\ref{thm:gibbs_explicit} states that Hamiltonian DQI can efficiently produce Gibbs states of any commuting Hamiltonian $H$ at inverse temperature $\beta \approx \frac{\ell}{1.12 \| H \|}$ (the $\approx$ suppresses dependence on the approximation error $\delta$), provided we can efficiently and exactly correct $\ell$ bit flip errors.
For a Hamiltonian consisting of $m$ Pauli terms, each with $\pm 1$ coefficients, we have $\|H\| \leq m$.
This is saturated only if the Hamiltonian is frustration-free. 
For a random $k$-local Hamiltonian the norm will instead be linear in $m$ but with a coefficient less than one. 
That is, for our ensemble of random $k$-local commuting Pauli Hamiltonians
\begin{equation}
\|H_{k,m/n}\| = \alpha_{k,m/n} m \quad \alpha_{k,m/n} < 1.
\end{equation}
Taking our empirical values of $e/m$, where $e$ is the maximum weight at which errors are correctable by Belief Propagation, at $n=3000$ to be our estimates of the asymptotically achievable values of $\ell/m$, we obtain estimates of the inverse temperatures efficiently reachable shown in Table~\ref{tab:all_results}.
Note that these estimates come with an important caveat:  we do not have a guarantee that the distance of our codes is linear in $n$, hence Theorem~\ref{thm:robustness_commuting} about the robustness against decoding failures does not directly apply.
In fact, there is a non-negligible (i.e. inverse-polynomial) probability that the distance is 3, since the probability that three randomly chosen commuting Paulis multiply to the identity is non-negligible.
An inverse-polynomial chance of decoding failure is not expected to degrade the performance of HDQI by too much; Indeed, a robustness theorem beyond half the distance of the symplectic code was shown for DQI in \cite{JSW24}, Theorem 10.1. Therefore, we estimate the performance of HDQI based on extrapolations beyond half the distance of the guarantees of the semicircle law and the Gibbs approximation theorem. These results are summarized in Table~\ref{tab:all_results}.

\begin{table}[ht!]
\centering
\begin{minipage}{0.3\textwidth}
    \centering
    $m/n = 3$ \\ \vspace{1mm}
    \begin{tabular}{|c|c|c|}
    \hline
    k & $\ell/m$ & $\beta$ \\
    \hline
    3 & 0.03813 & $0.03405 / \alpha_{3,3}$ \\
    4 & 0.04054 & $0.03620 / \alpha_{4,3}$ \\
    5 & 0.03850 & $0.03437 / \alpha_{5,3}$ \\
    6 & 0.03578 & $0.03194 / \alpha_{6,3}$ \\
    \hline
    \end{tabular}
\end{minipage}
\hfill 
\begin{minipage}{0.3\textwidth}
    \centering
    $m/n = 6$ \\ \vspace{1mm}
    \begin{tabular}{|c|c|c|}
    \hline
    $k$ & $\ell/m$ & $\beta$ \\
    \hline
    3 & 0.01353 & $0.01208/ \alpha_{3,6}$ \\
    4 & 0.01557 & $0.01390 / \alpha_{4,6}$ \\
    5 & 0.01538 & $0.01373 / \alpha_{5,6}$ \\
    6 & 0.01463 & $0.01306 / \alpha_{6,6}$ \\
    \hline
    \end{tabular}
\end{minipage}
\hfill 
\begin{minipage}{0.3\textwidth}
    \centering
    $m/n = 10$ \\ \vspace{1mm}
    \begin{tabular}{|c|c|c|}
    \hline
    $k$ & $\ell/m$ & $\beta$ \\
    \hline
    3 & 0.00648 & $0.00579 / \alpha_{3,10}$ \\
    4 & 0.00788 & $0.00704 / \alpha_{4,10}$ \\
    5 & 0.00799 & $0.00713 / \alpha_{5,10}$ \\
    6 & 0.00771 & $0.00688 / \alpha_{6,10}$ \\
    \hline
    \end{tabular}
\end{minipage}
\caption{Achievable polynomial degree $\ell$ to according to $\ell/m = e/m$, where $e/m$ is our estimate of the asymptotic fraction of bit flips decodable by our modified Belief Propagation decoder for the random $k$-local commuting Hamiltonian model. For this estimate we use our average value at $n=3000$. 
The corresponding achievable inverse temperature is $\beta = \frac{1}{\alpha_{k,m/n}} \frac{\ell/m}{1.12}$ in accordance with Theorem~\ref{thm:gibbs_explicit}. 
We do not have precise estimates of of $\alpha_{k,m/n}$ but we know rigorously that these are less than one, and thus $\frac{\ell/m}{\sqrt{2}}$ constitutes a lower bound on the achievable $\beta$.}
\label{tab:all_results}
\end{table}

It is illustrative also to study the capabilities of HDQI to perform approximate optimization---that is, approximately the ground state energy---on the random local commuting model.
We use the semicircle law as an upper bound on the approximation ratio.
On the other hand, there are many ways to estimate the ground state energy by simple classical optimization heuristics.
One such technique is to start with the $\ket{0}$ state, and then use simulated annealing with moves drawn from a uniform distribution over single-qubit Clifford operators.
We find, however, that even optimistically using the semicircle law to upper bound the performance of HDQI, this simulated annealing heuristic significantly outperforms HDQI.
As a consequence, HDQI may not be very useful in practice to estimate the ground state energy of random commuting Pauli Hamiltonians.
This comparison across various $m/n$ and $k$ is given in Table~\ref{tab:HDQI_vs_CliffordSA}.

\begin{table}[t!]
    \centering
    \begin{tabular}{|c|c|c|} \hline
        $(m/n, k)$ & HDQI and BP & Clifford SA \\ \hline
        $(3, 3)$ & $69.03 \pm 0.09\%$ & $78.20 \pm 0.31\%$ \\ \hline
        $(3, 4)$ & $69.61 \pm 0.07\%$ & $77.22 \pm 0.35\%$ \\ \hline
        $(3, 5)$ & $69.11 \pm 0.06\%$ & $75.84 \pm 0.39\%$ \\ \hline
        $(3, 6)$ & $68.40 \pm 0.06\%$ & $74.60 \pm 0.32\%$ \\ \hline
        $(6, 3)$ & $61.47 \pm 0.05\%$ & $70.18 \pm 0.25\%$ \\ \hline
        $(6, 4)$ & $62.30 \pm 0.04\%$ & $69.44 \pm 0.23\%$ \\ \hline
        $(6, 5)$ & $62.21 \pm 0.04\%$ & $68.56 \pm 0.25\%$ \\ \hline
        $(6, 6)$ & $61.86 \pm 0.04\%$ & $67.65 \pm 0.30\%$ \\ \hline
        $(10, 3)$ & $57.93 \pm 0.04\%$  & $65.59 \pm 0.22\%$ \\ \hline
        $(10, 4)$ & $58.76 \pm 0.04\%$ & $64.98 \pm 0.20\%$ \\ \hline
        $(10, 5)$ & $58.81 \pm 0.04\%$ & $64.33 \pm 0.21\%$ \\ \hline
        $(10, 6)$ & $58.63 \pm 0.03\%$ & $63.72 \pm 0.19\%$ \\ \hline
    \end{tabular}
    \caption{Estimated performance of HDQI with Belief Propagation (BP) decoding on a random $k$-local commuting Hamiltonian model, with $m$ terms and $n$ qubits.
    Compared with HDQI+BP is a simulated annealing heuristic consisting of random single-qubit Cliffords.
    In each case, we run the algorithm for $n = 1000$ and give the average ratio over 50 independent runs of Algorithm~\ref{alg:greedy-commuting-sampler}.
    The calculation of approximation ratio in HDQI uses the semicircle law, which is not exact here because the distance of the symplectic code is not guaranteed to be large.
    Instead, the semicircle law is an upper bound on performance.
    Nonetheless, the classical heuristic outperforms this upper bound of HDQI.}
    \label{tab:HDQI_vs_CliffordSA}
\end{table}

\section{Classical approaches for commuting Hamiltonians}\label{sec:commuting_dequantizations}

Commuting Hamiltonians are widely regarded as ``less quantum'' than general Hamiltonians, and many attempts have been made to formalize this intuition complexity-theoretically. 
In particular, the complexity of approximating ground energies of commuting Hamiltonians has been well-studied. 
A number of special cases of the problem have been shown to be contained in $\mathsf{NP}$~\cite{BV05, KKR06, AE11, Sch11, IJ23, AE15, AKV18} or even $\mathsf{P}$~\cite{YB12}, in contrast to the ground energy approximation problem for general local Hamiltonians, which is $\mathsf{QMA}$-complete. 
On the other hand, the problem of approximating ground energies of commuting Hamiltonians is not known to be contained in $\mathsf{NP}$ in general.

In this section, we consider commuting Pauli Hamiltonians with $\pm 1$ coefficients, and investigate whether the state preparation problems for such Hamiltonians that we have addressed with HDQI can also be addressed using classical means. 
Fundamentally, preparing a Gibbs state for any Hamiltonian other than a diagonal one has no classical analog and thus it is nonsensical to ask whether it can be prepared efficiently classically. 
However, commuting Pauli Hamiltonians have an eigenbasis consisting entirely of stabilizer states. 
These in turn have efficient classical descriptions, e.g. the collection of the symplectic representations of their stabilizer generators, including sign (called a stabilizer tableau). 
We can thus ask whether there is an efficient classical algorithm to sample from classical descriptions of stabilizer states, such that the corresponding mixture of stabilizer states is the desired mixed state $\rho_{\calP}(H) = \mathcal{P}(H)/\mathrm{Tr}[\mathcal{P}(H)]$ that HDQI produces.

\subsection{Consequences of simultaneous Clifford diagonalization}\label{sec:diagonalizationIntro}

Given samples from such descriptions, one can efficiently estimate the expectation value of any Pauli observable (or more generally, any Clifford observable) using classical Monte Carlo. Alternatively, one can directly sample (classical distributions of states) from $\rho_{\calP}(H)$ by sampling these descriptions classically and then constructing the corresponding stabilizer states, which requires only pure Clifford circuits.

Given a set of commuting elements $\{P_1,\ldots,P_m\}$ of the $n$-qubit Pauli group, it is known that one can efficiently construct a Clifford circuit of $\mathrm{poly}(n)$ gates implementing a unitary $U$ such that, for all $j \in \{1,\ldots,m\}$, the operator $UP_jU^\dag$ is a tensor product only of $Z$ operators~\cite{gottesman2016surviving}.
(See Lemma\ref{lem:simul-diag} in the Appendix.)
Consequently, $U H U^\dag$ is a diagonal Hamiltonian whose eigenstates computational basis states, and thus can be associated with bitstrings in $\{0,1\}^n$. 
Since the coefficients in $H$ are all $\pm 1$ these energies are all integers between $-m$ and $m$.

Suppose now that we are given some non-negative function $f(E)$ and we wish to sample (classical descriptions of) states according to the mixture
\begin{equation}
    \rho_f(H) := \frac{f(H)}{\Tr[f(H)]} .
\end{equation}
We denote this task the \emph{spectral sampling problem}. 
The diagonalizing Clifford $U$ reduces the preparation of $\rho_f(H)$ to that of sampling from bitstrings $\mathbf{x} \in \{0,1\}^m$ according to the distribution
\begin{equation}
    p_f(\mathbf{x}) = \frac{f(E(\mathbf{x})) \dim(E(\mathbf{x}))}{\sum_E \dim(E)} ,
\end{equation}
where $E(\mathbf{x})$ is the energy of $\mathbf{x}$ as defined by the pure-$Z$ Hamiltonian $U H U^\dag$, and $\dim(E)$ is the degeneracy of the energy-$E$ subspace. 
More precisely, we reduce the spectral sampling problem to sampling from $p_f$ by way of the following algorithm.

\begin{enumerate}
    \item From the description of $H$, compute a description of the Clifford circuit implementing the diagonalizing unitary $U$ and, for each $i \in \{1,\ldots,m\}$ compute $\mathbf{a}^{(i)} \in \mathbb{F}_2^n$ such that \begin{align}
        U P_i U^\dag = Z_1^{a^{(i)}_1} \otimes \ldots \otimes Z_n^{a^{(i)}_n}
    \end{align}
    This diagonalization procedure runs in time $O(\poly(n, m))$~\cite{gottesman2016surviving}.
    
    \item Draw samples $\mathbf{x}$ from the distribution $p_f(\mathbf{x})$. This cannot always be done efficiently, but below we describe several cases where it can.
    
    \item For each sample drawn, construct a classical description of the stabilizer state $U \ket{\mathbf{x}}$, using the stabilizer formalism.
\end{enumerate}
In general, sampling $\mathbf{x} \sim p_f(\mathbf{x})$ is $\mathsf{NP}$-hard.
This is because $\sum_{i=1}^m v_i Z_1^{a^{(i)}_1} \otimes \ldots \otimes Z_m^{a^{(i)}_n}$ is a classical Ising Hamiltonian, whose ground state is well-known to be $\mathsf{NP}$-hard~\cite{barahona1982computational}. 
However, there are a number of techniques that can be applied to this problem which solve it in several significant special cases.

\subsection{Stabilizer Hamiltonians}

One of the simplest cases to consider is that of stabilizer Hamiltonians. Let $S = \{P_1,\ldots,P_n\}$ be an independent set of generators for the stabilizer group of some stabilizer state $\ket{\Psi_S(\mathbf{0})}$ on $n$ qubits. By independent, we mean that there is no subset of $S$ whose product is the identity. 
This is equivalent to requiring that $\mathrm{symp}(P_1),\ldots,\mathrm{symp}(P_n) \in \mathbb{F}_2^{2n}$ be linearly independent. 
Then the Hamiltonian
\begin{equation}
    H_S = - \sum_{j=1}^n P_j
\end{equation}
has $\ket{\Psi_S(0)}$ as its unique ground state. 
More generally, an eigenbasis for $H$ can be constructed from the $2^n$ stabilizer states $\ket{\Psi_S(\mathbf{s})}$ corresponding to the stabilizer groups $\langle (-1)^{s_1} P_1, \ldots, (-1)^{s_n} P_n \rangle$ for $\mathbf{s} \in \mathbb{F}_2^n$. 
The eigenvalue of the eigenstate $\ket{\Psi_S(\mathbf{s})}$ is given by $2|\mathbf{s}|-n$.

For such a Hamiltonian the spectral sampling problem for arbitrary $f$ can be always efficiently solved as follows. 
First, sample a Hamming weight $w \in \{0,1,\ldots,n\}$ according to the distribution \begin{equation}
    p_f(w) = \frac{f(2w-n) \binom{n}{w}}{\sum_{w=0}^n f(2w-n) \binom{n}{w}}.
\end{equation}
Since the distribution has efficiently computable probabilities supported on a domain of size $O(n)$, the sampling procedure can always be done efficiently. 
Next, given $w$, sample a string $\mathbf{s}$ uniformly from the set of strings of length $n$ with Hamming weight $w$. 
Finally, efficiently prepare (a classical description of) the stabilizer state $\ket{\Psi_S(\mathrm{s})}$. 
The mixture obtained from this sampling procedure over pure stabilizer state descriptions is precisely that of $f(H)/\Tr[f(H)]$.

\subsection{Nearly independent Hamiltonians}
\label{sec:perfect_classical}

We next consider the problem of generalizing the above procedure to the case where the $m$ Pauli terms comprising the Hamiltonian still all commute but are no longer independent. 
That is, some Pauli terms in the Hamiltonian can be obtain as the product of others. 
In this circumstance, unlike for the stabilizer Hamiltonians considered in the previous section, it is not the case that each of the $2^m$ possible choices of signs for the $m$ terms is realizable by a quantum state. 
Rather, suppose the maximal independent subset of the $m$ terms has size $m-k$. 
Without loss of generality, we could choose our labels so that these are $P_1,\ldots,P_{m-k}$. 
We then have dependencies of the form
\begin{equation}
    \label{eq:dependencies}
    \textrm{for $j = m-k+1,\ldots,m$:} \quad  P_j = q_j \prod_{i \in S_j} P_i \quad \mathrm{with} \quad S_j \subseteq [m-k] \quad \mathrm{and} \quad q_j \in \{\pm 1\} .
\end{equation}

In this scenario, there exist $2^{m-k}$ stabilizer states which form the eigenbasis of $H$. 
Each such state has a stabilizer group generated as $\langle s_1 P_1,\ldots, s_{m-k} P_{m-k} \rangle$, with one of the $2^{m-k}$ assignment of signs $s_1,\ldots,s_{m-k} \in \{ \pm 1\}$. 
For $i \in [m-k]$, the state defined by stabilizer group $\langle s_1 P_1,\ldots, s_{m-k} P_{m-k} \rangle$ is an eigenstate of $P_i$ with eigenvalue $s_i$. 
Hence, by Eqn.~(\ref{eq:dependencies}), this state is also an eigenstate of $P_j$---for $j > m-k$---with eigenvalue $s_j = q_j \prod_{i \in S_j} s_i$.

Since multiplication of elements of $\{\pm 1\}$ is isomorphic to addition in $\mathbb{F}_2$, we can re-express the above in linear-algebraic language. 
Let $B \in \mathbb{F}_2^{m \times 2n}$ be a binary matrix whose $i$th row is $\Symp(P_i)$.
By construction, $\rank(B) = m - k$.
Each eigenvalue $\lambda$ of $P_1,\ldots,P_m$ can be represented as a vector $\mathbf{e} \in \mathbb{F}_2^m$ given by $e_i = \frac{1}{2} (1 - s_i)$, such $\lambda = \sum_{i=1}^m (1 - 2e_i) = m - 2 |\mathbf{e}|$.
Then the dependencies in Eqn.~(\ref{eq:dependencies}) become a set of linear equations determining the last $k$ components of $\mathbf{e}$ from the first $m-k$ components of $\mathbf{e}$. 
We can thus write $\mathbf{e} = (\mathbf{y}|\mathbf{t})$, where
\begin{equation}
    \mathbf{t} = A \mathbf{y} + \mathbf{u},
\end{equation}
$A \in \mathbb{F}_2^{k \times (m-k)}$ is given by
\begin{equation}
    A_{ij} = \begin{cases}
        1 & \textrm{if $j \in S_{m-k+i}$} \\ 0 & \textrm{otherwise,}
    \end{cases}  
\end{equation}
and $\mathbf{u} \in \mathbb{F}_2^k$ is specified by $u_i = \frac{1}{2} (1 - q_{m-k+i})$.
This simple linear-algebraic structure enables efficient sampling of nearly independent Hamiltonians, as we now prove.

\begin{lemma}[State preparation by sampling solutions to linear equations]
\label{lem:commutingToEnumerationReduction}
    Given $H = \sum_{i=1}^m v_u P_i$, where $v_i \in \set{\pm 1}$ and the $P_i$ are pairwise commuting $n$-qubit Paulis, suppose that $H$ has a unique ground state.
    Let $k = \dim \Symp(H)$ be the logical dimension of the symplectic code.
    Then if for each weight $w \in \{0, \ldots, m\}$, we can in time $\poly(n)$ both (a) count the number of weight $w$ strings $(\mathbf{y}|\mathbf{t}) \in \mathbb{F}_2^m$ such that $A \mathbf{y} = \mathbf{t} + \mathbf{u}$ (where $A$ and $\mathbf{u}$ are as defined above), and (b) sample uniformly from those strings, then we can efficiently sample classical descriptions of stabilizer states according to the mixture $\rho_f(H) \propto f(H)$, for any efficiently compute function $f$.
\end{lemma}

\begin{proof}
$H$ is a Hamiltonian whose ground state is the code state of a stabilizer code.
Thus, if there is a unique ground state, the code has zero logical qubits, and so there are $n$ independent stabilizers.
That is, $m-k=n$. 
The energy of a state corresponding to string $\mathbf{e} = (\mathbf{y}|\mathbf{t})$ is $m-2 |\mathbf{e}|$. 
Hence, given $f$,  $\rho_f(H) = f(H)/\Tr[f(H)]$ is produced by first sampling the weight according to the distribution
\begin{equation}
    p_f(w) = \frac{f(m-2w) \Delta(w)}{\sum_{w=0}^m f(m-2w) \Delta(w)},
\end{equation}
where $\Delta(w)$ is the degeneracy of the energy $E=m-2w$ eigenspace, which is equal to the number of strings $\mathbf{e} \in \mathbb{F}_2^m$ of Hamming weight $w$ of the form $\mathbf{e} = (\mathbf{y}|\mathbf{t})$ satisfying $A \mathbf{y} = \mathbf{t} + \mathbf{u}$. 
Since we can count this number efficiently by assumption, we can compute $\Delta(w)$ efficiently.
Given weight $w$, we now sample uniformly among the corresponding set of $\Delta(w)$ bitstrings which are solutions $(\mathbf{y} | \mathbf{t})$ to $A \mathbf{y} = \mathbf{t} + \mathbf{u}$. 
We then consider the stabilizer group $\langle (-1)^{y_1}P_1,\ldots,(-1)^{y_{m-k}} P_{m-k} \rangle$. 
Since $m-k = n$, this group specifies a unique stabilizer state, which we can efficiently produce by a Clifford circuit in the stabilizer formalism. 
This distribution over pure states is, by construction, precisely a mixture according to $\rho_f(H)$.
\end{proof}

\THMDecodableDequant*

\begin{proof}
As a consequence of Lemma~\ref{lem:commutingToEnumerationReduction}, it suffices to show that for each weight $w$ we can count and sample from the bitstrings of weight $w$ in a given coset of a binary linear code with $k$ parity checks.
This is a variant of the subset sum problem.
We do so with Algorithm~\ref{alg:dequantDynamicProgramming}, which counts these weight $w$ strings through a dynamic programming approach.
The algorithm exploits the recursive structure of the counting problem based on whether or not the $i$th base syndrome is included.
More precisely, for parity check matrix $B^\intercal = \left[\mathbf{h}_1\cdots\mathbf{h}_m\right]$, the number of bitstrings $\mathbf{y}$ with weight $w$, syndrome $\mathbf{z}$, and whose first $i-1$ entries are zero is the sum of (a) the number of bitstrings with weight $w-1$, syndrome $\mathbf{z} - \mathbf{h}_i$, and first $i$ entries zero, and (b) the number of bitstrings with weight $w$, syndrome $\mathbf{z}$, and first $i$ entries zero.

\begin{algorithm}
\caption{Counting coset elements of a fixed weight.}\label{alg:dequantDynamicProgramming}
\begin{algorithmic}[1]
\Require Target weight $w$, parity check matrix $B^\intercal = \left[\mathbf{h}_1\cdots\mathbf{h}_m\right]\in \mathbb{F}_2^{k\times m}$, target syndrome $\mathbf{z}\in \mathbb{F}_2^k$.
\Ensure Number $N(w) = N(w, 0, \mathbf{z})$ of $\mathbf{y}\in\mathbb{F}_2^{m}$ with  weight $|\mathbf{y}|=w$ and syndrome $B^\intercal \mathbf{y}=\mathbf{z}$.
\For{each $w, i\in \{0, \ldots, m\}$ and $\mathbf{x}\in \mathbb{F}_2^k$}
\State $N(w, i, \mathbf{x}) \gets$ 0
\EndFor
\State $N(0, m, \mathbf{0}) \gets$ 1 \Comment{Corresponds to the all-$0$ string}
\For{each $r\in (0, \ldots, w)$}
    \For{each $i \in (m, m-1, \ldots, 1)$}
        \For{each $\mathbf{x}\in\mathbb{F}_2^k$}
            \State $N(r, i-1, \mathbf{x}) \gets N(r, i, \mathbf{x}) + N(r-1, i, \mathbf{x}-\mathbf{h}_i)$  \Comment{Assuming $\mathbf{h}_i\neq \mathbf{0}$, the number of possible completions of a codeword branches depending on whether we set bit $i$ to 0 or 1.}
        \EndFor
    \EndFor
\EndFor
\State \Return $N(w, 0, \mathbf{z})$\label{line:endStateTableComplete}
\end{algorithmic}
\end{algorithm}
At the end (line~\ref{line:endStateTableComplete} of Algorithm~\ref{alg:dequantDynamicProgramming}), we have computed
\begin{equation}
    N(w, i, \mathbf{z}) = \left|\left\{\mathbf{y}\in \mathbb{F}_2^m: |\mathbf{y}|=w \text{ and } \left(\bigwedge_{j=1}^{i} (y_j=0) \right) \text{ and } B^\intercal \mathbf{y}=\mathbf{z} \right\}\right|.
\end{equation}
Setting $i = 0$ gives the desired count.
We can adapt Algorithm~\ref{alg:dequantDynamicProgramming} to uniformly sample using a simple search-to-decision reduction. 
Specifically, we decide on the bits one at a time. 
When deciding whether to flip the $i^{\text{th}}$ bit, we compute the numbers
\begin{equation}
    n_{0}:= N(r, i, \mathbf{x})  \,\,\,\,\, n_1 :=  N(r-1, i, \mathbf{x}-\mathbf{h}_i),
\end{equation}
outputted by Algorithm~\ref{alg:dequantDynamicProgramming}, and flip a $n_1/(n_0+n_1)$-biased coin that decides whether the $i$th bit is set to 1. 
If we set $x_i$ to 1, then we decrement the allowed value of $w$ and update our current partial syndrome $\mathbf{z}$ by $\mathbf{h}_i$, just like in the counting algorithm.
The complexity of Algorithm~\ref{alg:dequantDynamicProgramming} is $O(2^k m^2)$.
We note, however, that once the lookup table of $N(r, i, \mathbf{z})$ values is prepared, the sampling procedure can be done in time $O(m)$ per sample, since the sampling traces a path of length $O(m)$ through the lookup table.

\end{proof}

\subsection{Beyond commuting Hamiltonians}
We conclude our discussion about competing classical algorithms by some remarks about non-commuting Hamiltonians.
The above spectral sampling approaches are only well-defined for commuting Hamiltonians.
Nonetheless, there are multiple approaches one could take when the Hamiltonian no longer commutes.
However, the techniques we outline below are (a) heuristic in nature and (b) are likely limited to approximate energy optimization rather than Gibbs state sampling, as there is no evidence that these approaches prepare a state close to a Gibbs state.
We have already highlighted one approach previously, namely a simulated annealing algorithm where the starting state is $\ket{0}$ and moves are chosen by randomly selecting a local (e.g. 1-qubit of 2-qubit) Clifford operator.
Alternatively, we could apply a stabilizer version of Prange's algorithm~\cite{prange1962use}.
In this case, however, we are limited to selecting a commuting set of Hamiltonian terms to force a $+1$ eigenstate, which means we wish to find a large independent set in the anticommutation graph (see Definition~\ref{def:commutativity_graph}).
Finally, we can use variational algorithms on parameterized families of trial states as a general method to estimate ground state energies.

\section{Non-commuting Pauli Hamiltonians}\label{sec:general}
When the Hamiltonian is a sum of commuting Pauli terms, the state preparation step of HDQI is always efficient, as shown previously. However, this no longer holds for general Pauli Hamiltonians whose terms may not commute. Note that the second step---decoding---does not depend explicitly at all on the commutativity structure of $H$. 

In this section, we introduce our most general formulation of HDQI. It is a reduction from applying a degree-$\ell$ polynomial of an arbitrary Pauli Hamiltonian $H$ to a combination of two tasks: a generalized pilot state preparation problem, and the same classical decoding problem as before.
In the case $H$ is ``nearly commuting" in a sense that we make precise below, we prove that the state preparation step can in fact be done efficiently or nearly efficiently, by proving that the requisite pilot state takes the form of a matrix product state with low bond dimension.
As an application, we then show that state sampling tasks (Gibbs state, low-energy states) for a broader class of Hamiltonians can be performed efficiently.

\subsection{Symmetric expansions of non-commutative polynomials}\label{sec:non-commuting-symmetric-expansion}
Let \begin{align} \label{eq:non-commuting_Hamiltonian}
    H = \sum_{i=1}^m v_i P_i
\end{align}
be a Hamiltonian, where $v_i \in \set{\pm 1}$ and $P_i$ is a $n$-qubit Pauli operator.
Recall that in the commuting Hamiltonian case we began by expanding the univariate polynomial $\calP(H)$ into elementary symmetric polynomials. The fact that terms in $H$ can now either commute or anticommute significantly complicates this expansion. In this section, we derive a general symmetric polynomial expansion for univariate polynomials whose input is a sum of formal variables $\sum_{i=1}^m z_i$ which either commute or anticommute. To do so, it will be convenient to keep track of the commutation structure by way of a graph.

\begin{definition}[Anticommutation graph] \label{def:commutativity_graph}
    Let $z_1, \dots, z_m$ be formal variables such that every pair of distinct variables either commutes or anticommutes. Define an undirected graph $G = (V, E)$ with node set $V$ and edge set $E$, such that every node $w_i$ represents a formal variable $z_i$. We define $E := \set{(w_i, w_j) \,|\, z_i z_j = -z_j z_i}$. That is, connect two nodes $w_i, w_j$ if their corresponding variables anticommute.
\end{definition}
If a set of variables $z_1, \dots, z_m$ has anticommutation graph $G$ with adjacency matrix $A \in \mathbb{F}_2^{m \times m}$, then \begin{align}
    z_i z_j = (1 - 2A_{ij}) z_j z_i .
\end{align}
For a given collection of formal variables $(z_i)$, there are many ways we could order the variables in a product. Once we rearrange the variables back in increasing order, the anticommutation implies that the rearranged version may have picked up a sign. It will be useful to formally state this sign as a function of the permutation $\pi \in S^m$, where $S^m$ is the symmetric group of $m$ elements.

\begin{definition}[Anticommutation sign] \label{def:sign_function}
    Let $z_1, \dots, z_m$ be formal variables with anticommutation graph $G$, given in Definition~\ref{def:commutativity_graph}. 
    Further suppose that for all $i$, $z_i^2 = 1$. 
    The \textit{sign function} of a permutation $\pi \in S^m$ is characterized by the anti-commutation graph $G$ and is defined as \begin{align}
        \sgn_{G}(\pi) := (z_m z_{m-1} \cdots z_1) (z_{\pi(1)} \cdots z_{\pi(m)}) \in \set{\pm 1} .
    \end{align}
\end{definition}
Definition~\ref{def:sign_function} is equivalent to the following implicit definition of the sign function.
\begin{align}
    z_{\pi(1)} \cdots z_{\pi(m)} = \sgn_G(\pi) \cdot z_1 \cdots z_m .
\end{align}
Importantly, the sign function depends only on the underlying anticommutation graph $G$ and not on the formal variables $(z_i)$ themselves. 
For this reason, we label the sign function with $G$ rather than $(z_i)$. 

Given a set of formal variables $z_1, \dots, z_m$, we may wish to consider the sign function for only a subset of terms $z_{i_1}, \dots, z_{i_p}$ (where we enforce a conventional order $i_1 < i_2 < \cdots < i_p$). 
Furthermore, we may wish to include multiple copies of each variable, in which case permutations which swap two copies of the same variable are equivalent to the identity permutation. 
We may specify such a subset by using a counting vector $\bmu \in \mathbb{Z}_{\geq 0}^{m}$, where $\mu_i$ counts the number of times $z_i$ is included. For example, if we wish to compute a sign function for a permutation $\pi \in S^4$ for $z_1, z_2, z_2, z_3$, then $\bmu = (1, 2, 1, 0, \dots, 0)$. 
Given $\bmu$, we may define a new graph $\widetilde{G}(\bmu)$ with $|\bmu| = \sum_{i=1}^m \mu_i$ nodes labeled $w_1^{(1)}, \dots, w_1^{(\mu_1)}, w_2^{(1)}, \dots, w_m^{(\mu_m)}$, such that $w_i^{(j)}$ is associated with the $j$th copy of $z_i$. We connect $w_i^{(j)}$ with $w_{i'}^{(j')}$ if $z_i$ and $z_{i'}$ anticommute, or equivalently if $w_i$ and $w_{i'}$ are connected in the original graph $G$. In this case we define \begin{align}
    \sgn_{G, \bmu}(\pi) =  \sgn_{\widetilde{G}(\bmu)}(\pi) ,
    \label{eq:Gtilde_sign_function}
\end{align}
where $\pi \in S^p$ and $p = |\bmu|$. If $\bmu = (1, 1, \dots, 1)$, then $\widetilde{G}(\bmu) = G$. Note that while these definitions are all formal, the most convenient way to work with the sign function is through the simple representation \begin{align}
    \sgn_{G, \bmu}(\pi) = (-1)^{\operatorname{inv}(G, \bmu, \pi)} ,
    \label{eq:sign_function_as_inversions}
\end{align}
where $\operatorname{inv}(G, \bmu, \pi)$ is the number of \emph{inversions} needed to sort the variables back into place.
That is, the number of anticommuting swaps needed to sort the variables into ascending order.
With the sign function in place, we next formulate the average value of the sign function, which depends only on $G$. 
We call this average the \textit{antisymmetry character} of the graph $G$.
\begin{definition}[Antisymmetry character of a graph] \label{def:alpha_function}
    Let $z_1, \dots, z_m$ be formal variables with anticommutation graph $G$ as in Definition~\ref{def:commutativity_graph}.
    Then, the \emph{antisymmetry character} of $G$ is defined as \begin{align}
        \alpha_G = \underset{{\pi \sim S^m}}{\mathbf{E}} [\sgn_G(\pi)] = \frac{1}{m!} \sum_{\pi \in S^m} \sgn_G(\pi) ,
    \end{align}
    where $S^m$ is the symmetric group on $m$ elements and $\sgn_G(\pi)$ is the sign function given in Definition~\ref{def:sign_function}. 
    If we wish to restrict the formal variables to a subset with multiplicity, we again specify the restriction by way of a counting vector $\bmu \in \mathbb{Z}_{\geq 0}^{m}$. 
    Then the antisymmetry character restricted to such a subset of variables is given by \begin{align}
        \alpha_G(\bmu) = \frac{\mu_1! \cdots \mu_m!}{ |\bmu|!} \sum_{\pi \in S(\bmu)} \sgn_{G, \bmu}(\pi) .
    \end{align}
    where $|\bmu| = \sum_{i=1}^m \mu_i$ and $S(\bmu)$ is the set of all strings over the character set $\set{z_1, \dots, z_m}$ which contain $z_i$ $\mu_i$ times. Here, $\pi \in S(\bmu)$ is a slight abuse of notation intended to denote the fact that any sequence in $S(\bmu)$ can be equivalently thought of as a permutation $\pi$ mapping the canonical ordering $z_1^{\mu_1} \cdots z_m^{\mu_m}$ to the sequence associated with $\pi$. The normalization arises from the fact that $|S(\bmu)| = |\bmu|!/\mu_1! \cdots \mu_m!$. Equivalently, \begin{align} \label{eq:quotiented_alpha_function}
        \alpha_G(\bmu) = \underset{\pi \sim S(\bmu)}{\mathbf{E}}[\sgn_{G, \bmu}(\pi)] .
    \end{align}
    Depending on context, it may be convenient to refer directly to the actual sequences $\tau$ in $S(\bmu)$ instead of the corresponding permutations $\pi$. We do so in the notation \begin{align}
        \alpha_G(\bmu) = \frac{\mu_1! \cdots \mu_m!}{ |\bmu|!} \sum_{\tau \in S(\bmu)} \sgn_{G, \bmu}(\tau) ,
    \end{align}
    where $\sgn_{G, \bmu}(\tau)$ is the sign incurred by sorting the sequence $\tau$ back to canonical order.
\end{definition}
The antisymmetry character is simply the average sign incurred by sorting a set of variables with anticommutation graph $G$ over all initial orderings of the variables. 
In the case where we specify the variables at play with $\bmu$, the character is the average sign incurred by the variables specified by $\bmu$, across all possible inequivalent initial orderings.
As we next show, the antisymmetry character plays a key role in the generalized symmetric polynomial expansion.

We recall some notation which we first introduced in Theorem~\ref{thm:commuting_symmetric_polynomial_expansion}. For $1 \leq k \leq m$ and $\bmu \in \Z_{\geq 0}^m$ such that $|\bmu| = k$, let \begin{align}
    \binom{k}{\bmu} = \frac{k!}{\mu_1! \mu_2! \cdots \mu_m!} \label{eq:multinomial_coefficient}
\end{align} 
be the multinomial coefficient counting the number of ways we can distribute $k$ balls into $m$ boxes such that that $i$th box gets $\mu_i$ balls.
We use $\bmu \vDash_m k$ to denote vectors $\bmu$ which are \emph{weak compositions} of $k$; that is, $|\bmu| = k$.
Further, we denote the vector $\bmu \pmod{2} \in \F_2^{m}$ as the element-wise mod-2 parity of a vector $\bmu \in \Z^{m}_{\geq 0}$.
\begin{definition}[$\beta$-function] \label{def:beta_function}
    Let $z_1, \dots, z_m$ be a set of formal variables with anti-commutation graph $G$, as in Definition~\ref{def:commutativity_graph}. Then, for $1 \leq k \leq m$, we define the $\beta$-function of order $k$ as \begin{align}
        \beta_G^{(k)}(\mathbf{y}) = \sum_{\substack{\bmu : \bmu \vDash_m k\\\bmu \text{ (mod 2)} = \mathbf{y}}} \binom{k}{\bmu} \cdot \alpha_G(\bmu), \quad \text{ for } \vec y \in \F_2^m.
    \end{align}
    Here, $\alpha_G(\bmu)$ is given by Definition~\ref{def:alpha_function} and $\binom{k}{\bmu}$ is the multinomial coefficient from Eqn.~(\ref{eq:multinomial_coefficient}).
    Note that we have $\beta_G^{(k)}(\mathbf{y}) = 0$ whenever $|\mathbf{y}| > k$, since the sum becomes vacuous. 
\end{definition}

\begin{theorem}[Non-abelian univariate polynomial symmetric expansion]\label{thm:generalized_symmetric_expansion}
    Let $\mathcal{P}(x)$ be a degree-$\ell$ polynomial of a single formal variable $x$, and let $x = \sum_{i=1}^m z_i$ such that $z_i^2 = 1$ for all $i$ and the $z_i$ have an anticommutation graph given by $G$, as in Definition~\ref{def:commutativity_graph}. Then there exists constants $c_k \in \mathbb{R}$ for $k \in \set{0, \dots, \ell}$ such that 
    \begin{align}
        \mathcal{P}\left( \sum_{i=1}^m z_i \right) & = \sum_{k=0}^\ell c_k \, \,\sum_{\mathbf{y} \in \F_2^m}
       \beta_G^{(k)}(\mathbf{y}) \cdot
        z_{1}^{y_1} \cdots z_{m}^{y_m} ,
    \end{align}
    where $\beta_G^{(k)}(\mathbf{y})$ is the $\beta$-function of order $k$, as in Definition~\ref{def:beta_function}.
\end{theorem}
\begin{proof}
    Write $\calP(x) = \sum_{k=0}^\ell c_k x^k$. Then by definition,
    \begin{align}
        \mathcal{P}\left( \sum_{i=1}^m z_i \right) & = \sum_{k=0}^\ell c_k \left( \sum_{i=1}^m z_i \right)^k = \sum_{k=0}^\ell c_k \sum_{1 \leq a_1, \dots, a_k \leq m} z_{a_1} \cdots z_{a_k}.
    \end{align}
    We simplify by collecting like terms. Specifically, we claim that
    \begin{align} \label{eq:generalized_expansion_partial}
    \sum_{1 \leq a_1, \dots, a_k \leq m} z_{a_1} \cdots z_{a_k} & =  \sum_{\bmu : \bmu \vDash_m k}  \binom{k}{\bmu} \cdot \alpha_G(\bmu) \cdot z_{1}^{\mu_1} \cdots z_{m}^{\mu_m} ,
\end{align}
To verify this equality, note that on the left-hand side of Eqn.~(\ref{eq:generalized_expansion_partial}) every term in the sum is a $k$-th power, and we sum over every choice of indices, including repeated indices. Therefore, we first categorize each term by a counting vector $\bmu$ such that $\mu_i$ counts the number of times $z_i$ appears in the term. 
Once we have fixed the count $\bmu$, the only remaining degree of freedom is the choice of how to order the terms.
Each such sequence is an element of $S(\bmu)$, which is the set of all strings comprised of characters in $\set{z_1, \dots, z_m}$ such that $z_i$ is present $\mu_i$ times.
We sum over all such sequences, and there are $\binom{k}{\bmu}$ of them. 
Thus, \begin{align}
    \sum_{1 \leq a_1, \dots, a_k \leq m} z_{a_1} \cdots z_{a_k} & = \sum_{\bmu : \bmu \vDash_m k} \sum_{\pi \in S(\bmu)} [z_1^{\mu_1} \cdots z_m^{\mu_m}]_\pi ,
\end{align}
where $[z_1^{\mu_1} \cdots z_m^{\mu_m}]_\pi$ an ordered sequence in $S(\bmu)$ which can be interpreted as a permutation $\pi$ of the canonical ordering $z_1^{\mu_1} \cdots z_m^{\mu_m}$. We can then sort every term back into increasing order, which incurs a sign function. Thus, by Eqn.~(\ref{eq:quotiented_alpha_function}), \begin{align}
    \sum_{1 \leq a_1, \dots, a_k \leq m} z_{a_1} \cdots z_{a_k} & = \sum_{\bmu : \bmu \vDash_m k} \sum_{\pi \in S(\bmu)} \sgn_{G, \bmu}(\pi) z_1^{\mu_1} \cdots z_m^{\mu_m} \\
    & = \sum_{\bmu : \bmu \vDash_m k} \frac{k!}{\mu_1! \cdots \mu_m!} \cdot \frac{\mu_1! \cdots \mu_m!}{k!} \sum_{\pi \in S(\bmu)} \sgn_{G, \bmu}(\pi) z_1^{\mu_1} \cdots z_m^{\mu_m} \\
    & =  \sum_{\bmu : \bmu \vDash_m k} \underset{\pi \sim S(\bmu)}{\mathbf{E}} \left[\sgn_{G, \bmu}(\pi)\right] z_1^{\mu_1} \cdots z_m^{\mu_m} \\
    & = \sum_{\bmu : \bmu \vDash_m k} \binom{k}{\bmu} \cdot \alpha_G(\bmu) \cdot z_1^{\mu_1} \cdots z_m^{\mu_m} .
\end{align}
Now, because $z_i^2 = 1$, we can further group terms by noting that only the parities of the $\mu_i$ matter. 
Thus, \begin{align}
    \sum_{\bmu : \bmu \vDash_m k} \binom{k}{\bmu} \cdot \alpha_G(\bmu) \cdot z_1^{\mu_1} \cdots z_m^{\mu_m} & = 
    \sum_{\mathbf{y} \in \F_2^m} \left( \sum_{\substack{\bmu : \bmu \vDash_m k\\\bmu \text{ (mod 2)} = \mathbf{y}}} \binom{k}{\bmu} \cdot \alpha_G(\bmu) \right) z_1^{y_1} \cdots z_m^{y_m} \label{eq:non-commutative_expansion_divergence_point} \\
    & = \sum_{\mathbf{y} \in \F_2^m} \beta_{G}^{(k)}(\mathbf{y}) \cdot z_1^{y_1} \cdots z_m^{y_m} ,
\end{align}
where we used Definition~\ref{def:beta_function}. The above equation, in combination with Eqn.~(\ref{eq:generalized_expansion_partial}), completes the proof.
\end{proof}
It is illustrative to compare the above derivation with that in the proof of Theorem~\ref{thm:commuting_symmetric_polynomial_expansion}. In particular, up to the point of Eqn.~(\ref{eq:commutative_expansion_divergence_point}) in the proof of Theorem~\ref{thm:commuting_symmetric_polynomial_expansion} and Eqn.~(\ref{eq:non-commutative_expansion_divergence_point}) in the proof of Theorem~\ref{thm:generalized_symmetric_expansion}, the expansion looks essentially the same, except that in the generalized case we have an extra factor of $\alpha_G(\bmu)$ which accounts for the interference of all the sign functions obtained from the anticommutation structure. In the proof of Theorem~\ref{thm:commuting_symmetric_polynomial_expansion}, the lack of the antisymmetry character enabled us to exploit the fact that the multinomial coefficient is symmetric under re-ordering of the entries of $\bmu$, and hence we can organize the sum over $\bmu$ over only the weight of $\mathbf{y}$, which simplifies the expansion in such a way that every term with the same degree has the same coefficient. The presence of the antisymmetric character destroys the symmetry, giving rise to a much more complicated sum. 

More precisely, the extra complication introduced by the anticommutation arises from the fact that higher-degree terms can interfere with lower-degree terms. In the commuting case, a term of the form $z_1 z_2 z_1$ is exactly equivalent to $z_2$, because $z_1^2 = 1$. Therefore, any time there is a repetition, we need only use combinatorics to keep track of the number of repetitions. However, in the non-commuting case, the same term may instead be $-z_2$ if $z_1 z_2 = -z_2 z_1$. Therefore, this degree-3 term may actually cancel out a degree-1 term. The $\beta$-function encodes precisely how these cancellations occur for any anticommutation graph $G$.

\subsection{Main reduction: polynomials of Hamiltonians to pilot state preparation and classical decoding}\label{sec:main-non-commuting}
\begin{theorem}[Main reduction, general case] \label{thm:main_reduction_general_case}
    Let $H = \sum_{i=1}^m v_i P_i$ be a Hamiltonian on $n$ qubits, where $v_i \in \set{\pm 1}$ and the $P_i$ are $n$-qubit Paulis. Assume that $m = \poly(n)$. 
    Suppose that $\calD_H^{(\ell)}$ is a weight-$\ell$ decoding oracle for $H$ (as in Definition~\ref{def:decoding_oracle}). 
    Let $\calP(x)$ be any univariate polynomial of degree $\ell \leq m$. 
    There is a quantum algorithm which on input a pilot state
    \begin{align}
        \ket{\calR^\ell(H)} & = \frac{1}{\calN} \sum_{k=0}^{\ell} c_k \sum_{\mathbf{y} \in \mathbb{F}_2^m} \beta_G^{(k)}(\mathbf{y}) \ket{\mathbf{y}} ,
    \end{align}
    where $\beta_G^{(k)}$ is given in Definition~\ref{def:beta_function} and $\calN$ is a normalization constant,  prepares the state 
    \begin{align}
        \rho_{\calP}(H) = \frac{\calP^2(H)}{\Tr[\calP^2(H)]}
    \end{align}
    using a single call to $\calD_H^{(\ell)}$.
    The algorithm, described in Algorithm~\ref{alg:general_HDQI_algorithm}, runs in time $\poly(n)$.
\end{theorem}
\begin{proof}
    This reduction proceeds almost identically to that of Theorem~\ref{thm:commuting_reduction_to_decoding}. 
    The key difference is that in Theorem~\ref{thm:commuting_reduction_to_decoding}, the pilot state is always efficient to prepare, so we prepare it ourselves and keep the reduction only to decoding. 
    In the general case, the pilot state need not be efficient to prepare, so we reduce to both the problems of pilot state preparation \emph{and} decoding. 
    Once the pilot state has been prepared, however, the reduction to decoding is analogous.
    
    To begin, tensor on a maximally entangled state $\ket{\Phi^n}$ on $n$ pairs of qubits, to obtain 
    \begin{align}
        \ket{\psi_1}_{ABC} = \ket{\calR^\ell(H)}_A \otimes \ket{\Phi^n}_{BC} = \frac{1}{\calN} \sum_{k=0}^{\ell} c_k \sum_{\mathbf{y} \in \mathbb{F}_2^m} \beta_G^{( k)}(\mathbf{y}) \ket{\mathbf{y}}_A \otimes \ket{\Phi^n}_{BC} .
    \end{align}
    Next, we apply $\bigotimes_{i=1}^m Z_i^{(1-v_i)/2}$ to register $A$, which applies $Z$ to the $i$th qubit if $v_i = -1$: 
    \begin{align}
        \ket{\psi_2}_{ABC} = \frac{1}{\calN} \sum_{k=0}^{\ell} c_k \sum_{\mathbf{y} \in \mathbb{F}_2^m} \beta_G^{( k)}(\mathbf{y}) \left( \prod_{i \in \supp(\mathbf{y})} v_i \right) \ket{\mathbf{y}}_A \otimes \ket{\Phi^n}_{BC} .
    \end{align}
    We now apply controlled Paulis $\prod_{i=1}^m (C^{(i)}P_i)_{A \to B}$, where $(C^{(i)}P_i)_{A \to B}$ applies the $n$-qubit Pauli $P_i$ to register $B$ if the $i$th qubit of register $A$ is $1$. 
    Recall that $P_{\mathbf{y}} = \prod_{i \in \supp(\mathbf{y})} P_i$. 
    In this notation, we now have 
    \begin{align}
        \ket{\psi_3}_{ABC} = \frac{1}{\calN} \sum_{k=0}^{\ell} c_k \sum_{\mathbf{y} \in \mathbb{F}_2^m} \beta_G^{( k)}(\mathbf{y}) \left( \prod_{i \in \supp(\mathbf{y})} v_i \right) \ket{\mathbf{y}}_A \otimes (P_{\mathbf{y}} \otimes I)_{BC} \ket{\Phi^n}_{BC} .
    \end{align}
    By applying the coherent Bell measurement technique of Theorem~\ref{thm:app-Bell_measurement}, we create 
    \begin{align}
        \ket{\psi_4}_{ABC} = \frac{1}{\calN} \sum_{k=0}^{\ell} c_k \sum_{\mathbf{y} \in \mathbb{F}_2^m} \beta_G^{( k)}(\mathbf{y}) \left( \prod_{i \in \supp(\mathbf{y})} v_i \right) \ket{\mathbf{y}}_A \otimes \ket{\sym(P_\mathbf{y})}_{BC} ,
    \end{align}
    where $\sym(P_{\mathbf{y}})$ is the symplectic representation of $P_{\mathbf{y}}$ as described in Definition~\ref{def:symplectic_vectors}. 
    Using the decoding oracle $\calD^{(\ell)}_H$ on registers $A$ and $BC$, we can uncompute $\ket{\mathbf{y}}_A$ and discard register $A$. 
    \begin{align}
        \ket{\psi_5}_{BC} = \frac{1}{\calN} \sum_{k=0}^{\ell} c_k \sum_{\mathbf{y} \in \mathbb{F}_2^m} \beta_G^{( k)}(\mathbf{y}) \left( \prod_{i \in \supp(\mathbf{y})} v_i \right) \otimes \ket{\sym(P_\mathbf{y})}_{BC} .
    \end{align}
    We then uncompute the coherent Bell measurement:
    \begin{align}
        \ket{\psi_6}_{BC} & = \frac{1}{\calN} \sum_{k=0}^{\ell} c_k \sum_{\mathbf{y} \in \mathbb{F}_2^m} \beta_G^{(k)}(\mathbf{y}) \left( \prod_{i \in \supp(\mathbf{y})} v_i \right) (P_{\mathbf{y}} \otimes I)_{BC} \ket{\Phi^n}_{BC} \\
        & = \frac{1}{\calN} \sum_{k=0}^{\ell} c_k \sum_{\mathbf{y} \in \mathbb{F}_2^m} \beta_G^{( k)}(\mathbf{y}) (z_1^{y_1} z_2^{y_2} \cdots z_m^{y_m} \otimes I)_{BC} \ket{\Phi^n}_{BC} ,
    \end{align}
    where $z_i = v_i P_i$. By comparing to the generalized polynomial expansion from Theorem~\ref{thm:generalized_symmetric_expansion}, we find that \begin{align}
        \ket{\psi_6} = \frac{1}{\calN} (\calP(H) \otimes I) \ket{\Phi^n} ,
    \end{align}
    at which point tracing out the second register of $\ket{\Phi^n}$ completes the proof.
    Each step described can be implemented efficiently.
\end{proof}

\begin{algorithm}[H]        
  \caption{Hamiltonian DQI for general Hamiltonians}
  \label{alg:general_HDQI_algorithm}
  \begin{algorithmic}[1]    
  \Require{$H = \sum_{i=1}^m v_i P_i = \sum_i \lambda_i \ketbra{\lambda_i}{\lambda_i}$ (signed Pauli Hamiltonian on $n$ qubits), $|\Phi^n\rangle$ (maximally entangled state of $2n$ qubits), $\ket{\calR^{\ell}(H)}$ (pilot state), $\calD^{(\ell)}_H$ (decoding oracle), $\calP$ (polynomial of degree $\ell$).}
    \Ensure{$\rho_{\calP}(H) \propto \calP^2(H)$ ($n$-qubit mixed state).}
      \State Begin with the state $\ket{\calR^{\ell}(H)}_A \otimes \ket{\Phi^n}_{BC}$.
      \State Apply $\bigotimes_{i=1}^m Z_i^{\frac{1-v_i}{2}}$ on the $A$ register. 
      Apply the $\prod_{i=1}^m (C^{(i)} P_i)_{A\rightarrow B}$, where each operation is a Pauli $P_i$ on register $B$, controlled by the $i$th qubit of register $A$.
      \State Apply the coherent Bell measurement to registers $BC$, mapping the Bell basis to its symplectic representation as specified in Theorem~\ref{thm:app-Bell_measurement}.
      \State Uncompute register $A$ by applying $\calD_{H}^{(\ell)}$ to registers $A$ and $BC$. 
      Undo Step 3 by applying the inverse coherent Bell measurement to registers $BC$.
      \State The reduced density matrix on $B$ now equals $\rho_{\calP}(H)$. Output register $B$.
  \end{algorithmic}
\end{algorithm}

We emphasize the parallel between Theorem~\ref{thm:commuting_reduction_to_decoding} and Theorem~\ref{thm:main_reduction_general_case} (correspondingly, between Algorithm~\ref{alg:amplitude_transform} and Algorithm~\ref{alg:general_HDQI_algorithm}).
In both cases, we expand the polynomial $\calP(H)$ into a sum of symmetric terms. We then use a combination of a pilot state representation the expanded polynomial and a decoder which carefully transfers the amplitudes from the pilot state onto the maximally entangled state. 
Theorem~\ref{thm:commuting_reduction_to_decoding} is a special case of Theorem~\ref{thm:main_reduction_general_case} in which the pilot state preparation is always efficient, as proven in Lemma~\ref{lemma:commuting_efficient_resource_state}, because the antisymmetry characters are trivial. 
Note that this generalized reduction is also robust to failures in the decoding oracle.

\begin{theorem}[Robustness of general reduction] \label{thm:robustness_general}
Let $H = \sum_{i=1}^m v_i P_i$ be a Pauli Hamiltonian, where $v_i \in \set{\pm 1}$ and the $P_i$ are $n$-qubit Paulis. 
Let $\calP(x)$ be an arbitrary polynomial of degree at most $\ell$.
Suppose that the Hamiltonian symplectic code $B^\intercal = \Symp(H)$ has distance at least $2 \ell + 1$,
and that there is a decoding oracle $\calD^{(\ell, \epsilon)}_H$ which for each $j$ from $1$ to $\ell$ satisfies $\calD^{(\ell, \epsilon)}_H \ket{\mathbf{y}} \ket{B^\intercal \mathbf{y}} = \ket{0} \ket{B^\intercal \mathbf{y}}$ for a $1 - \epsilon$ fraction of all errors $\mathbf{y}$ with weight $j$.
On the remaining $\epsilon$ fraction, $\calD^{(\ell, \epsilon)}_H$ may behave arbitrarily so long as unitarity is maintained.
Then the state $\rho_{\calP, \epsilon}(H)$ produced by the method of Theorem~\ref{thm:main_reduction_general_case} which uses $\calD^{(\ell, \epsilon)}_H$ in place of a perfect decoding oracle $\calD^{(\ell)}_H$, satisfies \begin{align}
    \calF(\rho_{\calP, \epsilon}(H), \rho_{\calP}(H)) \geq 1 - \epsilon ,
\end{align}
where $\calF$ is the fidelity, and \begin{align}
    \frac{1}{2} \lVert \rho_{\calP, \epsilon}(H) - \rho_{\calP}(H) \rVert_1 \leq \sqrt{2 \epsilon} ,
\end{align}
where $\frac{1}{2} \lVert \cdot \rVert_1$ is the trace distance.
\end{theorem}
The proof follows analogously to that of Theorem~\ref{thm:robustness_commuting}.

A key question for the generalized reduction is, therefore, whether there are any classes of Hamiltonians beyond commuting Hamiltonians for which the pilot state preparation task is efficient.

\subsection{Computing the \texorpdfstring{$\alpha$}\ \ and \texorpdfstring{$\beta$}\ \ functions}
Before we discuss techniques for efficient (and almost efficient) pilot state preparation, we develop further structural results about the antisymmetry character and the $\beta$-function which will be useful for state preparation algorithms.

The antisymmetry character has a recursive structure which is useful for computation.
\begin{lemma}[Recursive structure of $\alpha$] \label{lemma:alpha_recursion}
Let $G = (V, E)$ be an anticommutation graph as in Definition~\ref{def:commutativity_graph}, where $|V| = m$. Label the nodes $1, \dots, m$ and associate formal variable $z_i$ with node $i$ such that $z_i$ and $z_j$ anticommute if $(i, j) \in E$ and commute otherwise.
Let $\bmu \in \Z_{\geq0}^m$. Define \begin{align}
    \label{eq:Sigma_function}
    \Sigma_G(\bmu) = \frac{|\bmu|!}{\mu_1! \cdots \mu_m!} \alpha_G(\bmu) = \sum_{\tau \in S(\bmu)} \sgn_G(\tau) 
\end{align}
to be the un-normalized version of $\alpha_G(\bmu)$ from Definition~\ref{def:alpha_function} (and $S(\bmu)$ are the sequences where $z_i$ appears $\mu_i$ times). Then \begin{align}
    \Sigma_G(\bmu) & = \sum_{j : \mu_j > 0}  (-1)^{\sum_{i : i > j} A_{ij} \mu_i} \Sigma_G(\bmu - \mathbf{e}_j) ,
\end{align}
where $\mathbf{e}_j$ is the $j$th standard basis vector in $\F_2^m$.
\end{lemma}
\begin{proof}
We work with the representation of the sign function given in Eqn.~(\ref{eq:sign_function_as_inversions}), namely as the number of anticommuting swaps made while sorting a given string.
This representation can be interpreted more explicitly by executing the sorting via the bubble sort algorithm: find an adjacent pair out of order, and swap them.
Define the local inversion counting function $\operatorname{inv}_\tau(i, j)$, for $i < j$, to be the number of pairs $z_i, z_j$ in a string $\tau$ where $z_i$ appears after $z_j$.
Then \begin{align}
    \sgn_G(\tau) = \prod_{i < j : (i, j) \in E} (-1)^{\operatorname{inv}_\tau(i, j)} = \prod_{i < j} (-1)^{A_{ij} \operatorname{inv}_\tau(i, j)} ,
\end{align}
where $A$ is the adjacency matrix of $G$. 
Now, suppose we wish to compute $\Sigma_G(\bmu)$, which is a sum over all length $p = |\bmu|$ sequences $\tau$ of characters $\set{z_1, \dots, z_m}$ such that $z_i$ appears $\mu_i$ times. 
We may group the sequences $\tau$ by the last element of their sequence. Let $\calT_j = \set{\tau \in S(\bmu)\ \,|\, \tau_p = z_j}$. Then \begin{align}
    \label{eq:Sigma_last_symbol}
    \Sigma_G(\bmu) = \sum_{j : \mu_j > 0} \; \sum_{\tau \in \calT_j} \sgn_G(\tau) .
\end{align}
For $\tau \in \calT_j$, \begin{align}
    \sgn_G(\tau) = \sgn_G(\tau_{:p}) \prod_{i : i > j, (i, j) \in E} (-1)^{\mu_i} = \sgn_G(\tau_{:p}) (-1)^{\sum_{i : i > j} A_{ij} \mu_i}
\end{align} 
where $\tau_{:p}$ denotes the sequence $\tau$ with the final character omitted. This equation follows because we can compute the sign function by first sorting all but the last character, incurring a sign of $\sgn_G(\tau_{:p})$. Then, we sort the last symbol $z_j$ in place, where it will pick up a sign each time it passes through any symbol $z_i$ such that $i > j$ and $(i, j) \in E$ (i.e. $A_{ij} = 1$).
Combining this decomposition with Eqn.~(\ref{eq:Sigma_last_symbol}), 
\begin{align}
    \Sigma_G(\bmu) & = \sum_{j : \mu_j > 0} \; \sum_{\tau \in \calT_j} \sgn_G(\tau_{:p}) (-1)^{\sum_{i : i > j} A_{ij} \mu_i} \\
    & = \sum_{j : \mu_j > 0}  (-1)^{\sum_{i : i > j} A_{ij} \mu_i} \left( \sum_{\tau \in \calT_j} \sgn_G(\tau_{:p}) \right) \\
    & = \sum_{j : \mu_j > 0}  (-1)^{\sum_{i : i > j} A_{ij} \mu_i} \Sigma_G(\bmu - \mathbf{e}_j) ,
\end{align}
where the last line follows because we sum over all sequences $\tau \in S(\bmu)$ whose last character is $z_j$, of the sign ignoring the last character. This is exactly equivalent to summing over all signs of sequences from $S(\bmu - \mathbf{e}_j)$.
\end{proof}

This recursive structure enables a dynamic programming algorithm which computes the antisymmetry characters in exponential time. 
\begin{theorem}[Computing $\alpha$ via dynamic programming] \label{thm:alpha_dynamic_programming}
Let $G = (V, E)$ be an anticommutation graph as in Definition~\ref{def:commutativity_graph}, where $|V| = m$. 
Label the nodes $1, \dots, m$ and associate formal variable $z_i$ with node $i$ such that $z_i$ and $z_j$ anticommute if $(i, j) \in E$ and commute otherwise.
Let $\bmu \in \Z_{\geq0}^m$ such that $\bmu \vDash_m k$. 
Then there exists a deterministic classical algorithm which computes $\alpha_G(\bmu)$ in time 
\begin{align}
    O(k m N(\bmu) + m\cdot \poly(k)) = \begin{cases}
        O(k m 2^k) & \text{ if } k \leq m, \\
        O(k m (\lceil k/m \rceil + 1)^m ) & \text{ if } k > m .
    \end{cases}
\end{align}
where $N(\bmu) = \prod_{i=1}^m (\mu_i+1)$. 
The runtime assumes that elementary arithmetic and classical memory access are $O(1)$ time operations. 
\end{theorem}
\begin{proof}
It is enough to compute $\Sigma_G(\bmu)$ from Eqn.~(\ref{eq:Sigma_function}), which differs from $\alpha$ only by a normalization $\mu_1! \cdots \mu_m! / |\bmu|!$ which can be computed in $O(m \cdot \poly(k))$ time and space. 
Ultimately, this additive factor is subsumed by the runtime computed below and can thus be asymptotically ignored.
By Lemma~\ref{lemma:alpha_recursion}, $\Sigma_G$ has a recursive structure as \begin{align}
    \Sigma_G(\bmu) = \sum_{j : \mu_j > 0}  (-1)^{\sum_{i : i > j} A_{ij} \mu_i} \Sigma_G(\bmu - \mathbf{e}_j) .
\end{align}
This structure motivates a simple dynamic programming algorithm. 
Initially, $\Sigma_G(0) = 1$. 
Then, we can compute $\Sigma_G(\mathbf{e}_i)$ for all $i$ such that $\bmu -\mathbf{e}_i \in \Z_{\geq0}^m$. 
Then, we compute $\Sigma_G(\mathbf{e}_i + \mathbf{e}_j)$ for all $i, j$ such that $\bmu - \mathbf{e}_i - \mathbf{e}_j \in \Z_{\geq0}^m$. 
We repeat this process until we have computed $\Sigma_G(\bmu)$. 
At each step, we store our computation in a lookup table so that the next computation can reference the table without re-computation. 
Hence, calculating each next $\Sigma_G$ takes only $O(km)$ computations, since there are $O(m)$ terms in the exponentiated sum and $O(k)$ terms in the outer sum because at most $k$ entries of $\bmu$ are nonzero.
Since $\bmu$ satisfies $\bmu \vDash_m k$, each $\mu_i$ is $O(k)$.
Finally, the number of $\Sigma_G$ we must calculate is given by $N(\bmu) = \prod_{i=1}^m (\bmu_i + 1)$, since we calculate a $\Sigma_G(\bmu')$ for each $\bmu'$ such that $\bmu - \bmu' \in \Z_{\geq0}^m$, and we can select $\bmu'$ entry-wise by choosing between $0$ and $\bmu_i$ to subtract off the $i$th entry. 
The total runtime is thus $O(k m N(\bmu))$.

To complete the proof, we show that $N(\bmu) \leq 2^k$ if $k \leq m$ and $N(\bmu) \leq (\lceil k/m \rceil + 1)^m$ otherwise. 
This bound follows from the fact that a product of terms, subject to a constraint that the sum of terms is fixed, is maximized when the terms are as balanced as possible. 
If we relax $\bmu$ to be real-valued, we can prove this fact by Lagrangian optimization, by maximizing \begin{align}
    \calL(\bmu, \lambda) = \sum_{i=1}^m \log(1 + \mu_i) - \lambda \left( \sum_{i=1}^m \mu_i - k \right) .
\end{align}
We have taken the $\log$ of $N(\bmu)$ to simplify the maximization, since $\log f(x)$ and $f(x)$ have the same optima. Solving for $\nabla_{\bmu, \lambda}$, we find that $\mu_i = k/m$, the maximally balanced condition. 
In our case, $\mu_i \in \Z_{\geq0}$, so the most balanced we can be is to set is to distribute the integers as evenly as possible. 
If $k \leq m$, then we give $k$ spots $1$ and the remaining $m-k$ spots $0$, so $N(\bmu) \leq 2^k$. 
If $k > m$, then we iteratively add 1 to every spot until we run out of allocations, in which case every entry is $\geq 1$ and $\leq \lceil k/m \rceil$, so $N(\bmu) \leq (\lceil k/m \rceil + 1)^m$
\end{proof}
If we are interested only in computing $\alpha_G$ (with no multiplicity), then $N(\bmu) = 2^k$ and the algorithm above runs in $O(km 2^k)$ time. 
Typically, we take $k \leq n$ and $m = O(n)$, in which case we get a $O(n^2 2^{n})$ time algorithm.

Given an algorithm to calculate the antisymmetry character, we can also readily compute the $\beta$-function.
\begin{theorem}[Computing the $\beta$-function] \label{thm:computing_beta}
Let $\mathbf{y} \in \F_2^m$ such that $|\mathbf{y}| = w$, and let $ \beta_{G}^{(k)}(\mathbf{y})$ be as in Definition~\ref{def:beta_function}.
Assume that elementary arithmetic operations and classical memory access take $O(1)$ time.
There is a deterministic classical algorithm which computes $ \beta_{G}^{(k)}(\mathbf{y})$ in time $T$, where
    \begin{align} \label{eq:beta_compute_runtime}
    T = \begin{cases}
        O\left( \binom{\frac{k-w}{2} + m - 1}{m - 1} \cdot k m 2^k \right) & \text{ if } k \leq m, \\
        O\left( \binom{\frac{k-w}{2} + m - 1}{m - 1} \cdot k m (\lceil k/m \rceil + 1)^m \right) & \text{ if } k > m .
    \end{cases}
\end{align}
\end{theorem}
\begin{proof}
From Definition~\ref{def:beta_function},
\begin{align}
     \beta_{G}^{(k)}(\mathbf{y}) & = \sum_{\substack{\bmu : \bmu \vDash_m k\\\bmu \text{ (mod 2)} = \mathbf{y}}} \binom{k}{\bmu} \, \alpha_G(\bmu) = \sum_{\substack{\bmu : \bmu \vDash_m k\\\bmu \text{ (mod 2)} = \mathbf{y}}} \Sigma_G(\bmu) ,
\end{align}
where $\Sigma_G(\bmu)$ is defined in Eqn.~(\ref{eq:Sigma_function}).
Recall that vectors $\bmu$ for which $\bmu \vDash_m k$ are known as $m$-\emph{weak compositions} of $k$. 
There are exactly $\binom{k+m-1}{m-1}$ $m$-weak compositions of $k$.
We wish, however, to count the number of $m$-weak compositions of $k$, $\bmu$, whose parity is $\mathbf{y}$. 
Let $w = |\mathbf{y}|$ be the weight of $\mathbf{y}$.
If $w > k$, then there are zero such weak compositions and $ \beta_{G}^{(k)}(\mathbf{y}) = 0$. 
If $w \leq k$, then a $m$-weak composition of $k$, $\bmu$, with parity $\mathbf{y}$ can equivalently be expressed as an assignment of $\boldsymbol{\nu} \in \Z_{\geq0}^{m}$ such that $\bmu = \mathbf{y} + 2\boldsymbol{\nu}$. In other words, \begin{align}
    k = \sum_{i=1}^m \mu_i = \sum_{i=1}^m y_i + 2 \sum_{i=1}^m \nu_i = w + 2 \sum_{i=1}^m \nu_i .
\end{align}
Therefore, $\bmu$ is a $m$-weak composition of $k$ with parity $\mathbf{y}$ if and only if $\boldsymbol{\nu}$ is a $m$-weak composition of $\frac{k-w}{2}$.
If $\frac{k-w}{2} \notin \Z_{\geq0}$, then there are no such compositions.
Otherwise, there are $\binom{\frac{k-w}{2} + m - 1}{m - 1}$ such compositions.
Combined with the piecewise runtime analysis in Theorem~\ref{thm:alpha_dynamic_programming}, we find a total runtime for calculating the $\beta$-function to be bounded by Eqn.~(\ref{eq:beta_compute_runtime}).
\end{proof}

\subsection{Efficient pilot state preparation for sparse anticommutation graphs}
In this section, we prove that if the anticommutation graph (from Definition~\ref{def:commutativity_graph}) of the Hamiltonian $H$ is sufficiently sparse in a precise sense, then there is an efficient algorithm to prepare the pilot state required for Theorem~\ref{thm:main_reduction_general_case}. This algorithm relies on the factorization structure of the antisymmetry character, which we first prove.

\begin{lemma}[$\alpha$ factorizes into connected components] \label{lemma:alpha_factorization}
Let $z_1, \dots, z_m$ denote a set of $m$ formal variables with anticommutation graph $G = (V, E)$ as in Definition~\ref{def:commutativity_graph}.
Define $C_1, \dots, C_r$ to be the connected components of $G$.
That is, $C_t = (V_t, E_t)$, $V = \bigsqcup_{t=1}^r V_t$, $E = \bigsqcup_{t=1}^r E_t$.
Recall that we denote $\bmu \vDash_m k$ if $\bmu \in \Z_{\geq 0}^m$ and $\sum_{i=1}^m \mu_i = k$. Then the antisymmetry character, given in Definition~\ref{def:alpha_function}, factorizes as
\begin{align}
    \alpha_G(\bmu) = \prod_{t=1}^r \alpha_{C_t}(\bmu\rvert_{C_t}),
\end{align}
where $\bmu\rvert_{C_t}$ is the restriction of $\bmu$ onto the indices corresponding to the vertices in $C_t$.
\end{lemma}
\begin{proof}
By Definition~\ref{def:alpha_function},
\begin{align}
    \alpha_{G}(\bmu) & = \frac{\mu_1! \cdots \mu_m!}{|\bmu|!} \sum_{\pi \in S(\bmu)} \sgn_{G, \bmu}(\pi) ,
\end{align}
where $\sgn_{G, \bmu}(\pi)$ is the sign function of Definition~\ref{def:sign_function}.
Our technique is similar to that in the proof of Lemma~\ref{lemma:alpha_recursion}, but we give a self-contained argument.
We wish to re-express the sign function in a manifestly factorizable manner.
To do so, we will apply Eqn.~(\ref{eq:Gtilde_sign_function}).
In other words, define $\widetilde{G} = (\widetilde{V}, \widetilde{E})$ a new graph where $|\widetilde{V}| = |\bmu|$, constructed by associating the first $\mu_1$ nodes with $z_1$, the next $\mu_2$ nodes with $z_2$, and so on. 
Label the nodes in $\widetilde{V}$ as $a_1, \dots, a_p$.
We define the function $z(a_i)$ to map $a_i$ to the corresponding $z_j$. 
For example, $z(a_1) = z(a_2) = \cdots = z(a_{\mu_1}) = z_1$.
Connect $a_i$ to $a_j$ in $\widetilde{E}$ if $z(a_i)$ and $z(a_j)$ are connected in $E$.

A direct way to compute the sign function is to perform bubble sort on the starting sequence of variables. 
Specifically, there are $p = |\bmu|$ symbols, given by $\mu_i$ copies of $z_i$ for each $i$. 
We define symbols $a_1, \dots, a_p$ which are understood to represent the $z_i$ in such a way that $a_1, \dots, a_{\mu_1}$ each represent $z_1$, $a_{\mu_1 + 1}, \dots, a_{\mu_1 + \mu_2}$ represent $z_2$, and so on.
Initially then, our sequence is $a_{\pi(1)} a_{\pi(2)} \cdots a_{\pi(p)}$.
Using bubble sort, we initialize a running sign variable $\sigma = 1$, then repeatedly swap adjacent variables which are not in order until no such pair exists. 
If the swap occurs between $a_i$ and $a_j$ such that $(a_i, a_j) \in \widetilde{E}$, flip the sign of $\sigma$.
When bubble sort terminates, $\sigma = \sgn_{\widetilde{G}}(\pi) = \sgn_{G}(\pi)$ where the second equality follows from Eqn.~(\ref{eq:Gtilde_sign_function}).

When we execute bubble sort, the total number of adjacent swaps required to map $(a_{\pi(1)}, \dots, a_{\pi(p)})$ back to the initial sorted array $(a_1, \dots, a_p)$ is equal to the number of index pairs $(i, j)$ such that $i<j$ and $\pi(i)>\pi(j)$.
Among these swaps, the only ones that flip the sign are further those for which $(a_i, a_j) \in \widetilde{E}$.
As a consequence,
\begin{align}
    \alpha_G(\bmu) &= \frac{\mu_1! \cdots \mu_m!}{p!} \sum_{\pi \in S(\bmu)} (-1)^{\sum_{i < j : (a_i, a_j) \in \widetilde{E}} \mathbb{1}[\pi(i) > \pi(j)]} \\
    & = \frac{\mu_1! \cdots \mu_m!}{p!} \sum_{\pi \in S(\bmu)} \; \prod_{i < j : (a_i, a_j) \in \widetilde{E}} (-1)^{\mathbb{1}[\pi(i) > \pi(j)]}
\end{align}
Recall that $E = \bigsqcup_{t=1}^r E_t$, where $E_t$ is the edge set of the $t$th connected component of $G$. Using the same procedure which we applied to create $\widetilde{G}$ from $G$, we can create $\widetilde{C}_t = (\widetilde{V}_t, \widetilde{E}_t)$ from $C_t$. Note that $\widetilde{E} = \bigsqcup_{t=1}^r \widetilde{E}_t$, because $\widetilde{C}_t$ is constructed by making copies of each node in $C_t$ according to the multiplicities given by $\bmu$, and then duplicating their edge connections. Therefore, \begin{align}
    \alpha_G(\bmu) & = \frac{\mu_1! \cdots \mu_m!}{p!} \sum_{\pi \in S(\bmu)} \prod_{k=1}^r \;\;\prod_{i < j : (a_i, a_j) \in \widetilde{E}_k} (-1)^{\mathbb{1}[\pi(i) > \pi(j)]} \\
    & = \underset{\pi \sim S(\bmu)}{\mathbf{E}}\left[ \prod_{t=1}^r \underbrace{\left( \prod_{i < j : (a_i, a_j) \in \widetilde{E}_t} (-1)^{\mathbb{1}[\pi(i) > \pi(j)]} \right)}_{X_t(\pi)} \right] ,
\end{align}
where we have defined a new set of random variables $X_1(\pi), \dots, X_r(\pi)$.
For each $\pi \in S(\bmu)$, $X_t(\pi)$ depends only on the action of $\pi$ within $\widetilde{V}_t$. 
Since each such vertex set is disjoint, $X_t(\pi)$ and $X_{t'}(\pi)$ must be independent for $t \neq t'$. 
As a result, \begin{align}
    \alpha_G(\bmu) & = \prod_{t=1}^r \underset{\pi \sim S(\bmu)}{\mathbf{E}}\left[ \prod_{(i, j) \in \widetilde{E}_t} (-1)^{\mathbb{1}[\pi(i) > \pi(j)]} \right] \\
    & =\prod_{t=1}^r \underset{\pi \sim S(\bmu\rvert_{C_t})}{\mathbf{E}}\left[ \prod_{(i, j) \in \widetilde{E}_t} (-1)^{\mathbb{1}[\pi(i) > \pi(j)]} \right] \\
    & = \prod_{t=1}^r \underset{\pi \sim S(\bmu\rvert_{C_t})}{\mathbf{E}}\left[ \sgn_{\widetilde{C}_t}(\pi) \right] \\
    & = \prod_{t=1}^r \underset{\pi \sim S(\bmu\rvert_{C_t})}{\mathbf{E}}\left[ \sgn_{C_t, \bmu\rvert_{C_t}}(\pi) \right] \\
    & = \prod_{t=1}^r \alpha_{C_t}(\bmu\rvert_{C_t}) .
\end{align}
In the second line, we restricted the permutation to the only relevant indices (with an implicit understanding that the permutation acts on the index set in the relevant connected component). In the third line, we used the fact that the sign function can be expressed in the bubble sort form of the second line. 
In the fourth line, we applied Eqn.~(\ref{eq:Gtilde_sign_function}). The last line applies Definition~\ref{def:alpha_function} to complete the proof.
\end{proof}

The factorization of $\alpha$ implies that the $\beta$-function takes on a form reminiscent of a tensor contraction structure.
\begin{lemma}[$\beta$-function reduced form] \label{lemma:beta_function_reduced_form}
    Let $z_1, \dots, z_m$ denote a set of $m$ formal variables with anticommutation graph $G = (V, E)$ as in Definition~\ref{def:commutativity_graph}.
    Define $C_1, \dots, C_r$ to be the connected components of $G$, where $C_t = (V_t, E_t)$, $V = \bigsqcup_{t=1}^r V_t$, and $E = \bigsqcup_{t=1}^r E_t$.
    Let $m_t = |V_t|$ be the number of nodes in $C_t$, so that $\sum_{t=1}^r m_t = m$.
    Recall that we denote $\bkappa \vDash_r k$ if $\bkappa \in \Z_{\geq 0}^r$ and $\sum_{i=1}^r \kappa_i = k$. Then for $\mathbf{y} \in \F_2^m$, \begin{align}
         \beta_{G}^{(k)}(\mathbf{y}) = \sum_{\bkappa : \bkappa \vDash_r k} \binom{k}{\bkappa} \prod_{t=1}^r \beta_{C_t}^{(\kappa_t)}(\mathbf{y}\rvert_{C_t}) ,
    \end{align}
    where $\mathbf{y}\rvert_{C_t}$ denotes the restriction of $\mathbf{y}$ to the indices whose corresponding nodes are in $C_t$.
\end{lemma}

\begin{proof}
For brevity, let $\bmu^{(t)}=\bmu\rvert_{C_t}$ and let $\kappa_t = |\bmu^{(t)}|$ denote the size of the restriction of $\bmu$ onto $C_t$ so that $\sum_{t=1}^r \kappa_t = |\bmu|= k$.
We observe that for any $\boldsymbol{\kappa} \vDash_r k$, the multinomial coefficient admits a decomposition over subsets as
\begin{align}
    \binom{k}{\bmu} = \binom{k}{\bkappa} \cdot \prod_{t=1}^r \binom{\kappa_t}{\bmu^{(t)}}.
\end{align}
Applying this observation and Lemma~\ref{lemma:alpha_factorization},
\begin{align}
    \beta_{G}^{(k)}(\mathbf{y}) & =  \sum_{\substack{\bmu : \bmu \vDash_m k \\
    \bmu \text{ (mod 2)} = \mathbf{y}}} \binom{k}{\bmu} \, \alpha_G(\bmu) \\
    & = \sum_{\substack{\bmu : \bmu \vDash_m k \\
    \bmu \text{ (mod 2)} = \mathbf{y}}} \binom{k}{\bmu} \prod_{t=1}^r \alpha_{C_t}(\bmu^{(t)}) \\
    & = \sum_{\substack{\bmu : \bmu \vDash_m k \\
    \bmu \text{ (mod 2)} = \mathbf{y}}} \left[ \binom{k}{\bkappa} \cdot \prod_{t=1}^r \binom{\kappa_t}{\bmu^{(t)}} \right] \prod_{t=1}^r \alpha_{C_t}(\bmu^{(t)}) .
\end{align}
Next, choosing $\bmu$ such that $\bmu \vDash_m k$ and $\bmu \text{ (mod 2)} = \mathbf{y}$ is equivalent to choosing how many of each node of $G$ to include in a given term. 
An equivalent way to make this choice is to first choose how many terms is included in each connected component, i.e. choosing $\bkappa \in \Z_{\geq0}^{r}$ such that $\bkappa \vDash_r k$, and then for each connected component $C_t$ choosing $\bmu^{(t)}$ how many of each node to include in the total allocated number $\kappa_t$, such that the count's parity is $\mathbf{y}^{(t)}$. 
Therefore, \begin{align}
    \beta_{G}^{(k)}(\mathbf{y}) & =  \sum_{\bkappa : \bkappa \vDash_r k} \sum_{\substack{\bmu^{(t)} : \bmu^{(t)} \vDash_{m_t} \kappa_t \\
    \bmu^{(t)} \text{ (mod 2)} = \mathbf{y}^{(t)}}} \left[ \binom{k}{\bkappa} \cdot \prod_{t=1}^r \binom{\kappa_t}{\bmu^{(t)}} \right] \prod_{t=1}^r \alpha_{C_t}(\bmu^{(t)}) \\
    & = \sum_{\bkappa : \bkappa \vDash_r k} \binom{k}{\bkappa} \sum_{\substack{\bmu^{(t)} : \bmu^{(t)} \vDash_{m_t} \kappa_t \\
    \bmu^{(t)} \text{ (mod 2)} = \mathbf{y}^{(t)}}} \prod_{t=1}^r \binom{\kappa_t}{\bmu^{(t)}} \alpha_{C_t}(\bmu^{(t)}) \\
    & = \sum_{\bkappa : \bkappa \vDash_r k} \binom{k}{\bkappa} \prod_{t=1}^r \beta_{C_t}^{(\kappa_t)}(\mathbf{y}\rvert_{C_t}) .
\end{align}
\end{proof}

The tensorial structure of the reduced form in Lemma~\ref{lemma:beta_function_reduced_form} hints that the pilot state may also have a tensor-type form. 
To elucidate this connection more precisely---and thereby construct an algorithm to prepare the general pilot state $\ket{\calR^\ell(H)}$---we first review the formulation of matrix product states (MPS).
\begin{definition}[Matrix product state with open boundary conditions] \label{def:matrix_product_state}
    Let $\ket{\psi} \in (\C^2)^{\otimes n}$ be a $n$-qubit quantum state. For a vector $\mathbf{y} \in \F_2^n$, partition the vector into $\mathbf{y} = (\mathbf{y}^{(1)}, \dots, \mathbf{y}^{(r)})$. Suppose that there exists matrices $A^{(1)}[\mathbf{y}^{(1)}], \dots, A^{(r)}[\mathbf{y}^{(r)}] \in \C^{D \times D}$, and vectors $\mathbf{v}_L, \mathbf{v}_R \in \F_2^{D}$, such that \begin{align}
        \ket{\psi} = \sum_{\mathbf{y} \in \F_2^m} \mathbf{v}_L^\intercal A^{(1)}[\mathbf{y}^{(1)}] \cdots A^{(r)}[\mathbf{y}^{(r)}] \mathbf{v}_R \ket{\mathbf{y}} .
    \end{align}
    Then we say that $\ket{\psi}$ is a $q$-ary \emph{matrix product state}, with open boundary conditions, of \emph{bond dimension} $D$, where $q = \max_{t \in [1, r]} 2^{\dim \mathbf{y}^{(t)}}$.
\end{definition}
Much of the development of MPS theory has been motivated by understanding when a quantum state admits a concise classical description. 
In general, a MPS with $\poly(n)$ bond dimension can be represented in $\poly(n)$ classical memory~\cite{mps1,mps2}.
We next show, however, that MPS theory is also useful for quantum algorithms.

\begin{lemma}[General HDQI pilot state is a MPS] \label{lemma:resource_state_is_MPS}
    Let $H$ be a Pauli Hamiltonian and let $G = (V, E)$ be the corresponding anticommutation graph as in Definition~\ref{def:commutativity_graph}. 
    Suppose that $G = \bigsqcup_{t=1}^r C_t$, where $C_t = (V_t, E_t)$ are the connected components of $G$. 
    Let $m_t = |V_t|$ be the number of nodes in the $t$-th connected component.
    Then the general pilot state of degree $\ell$ for $H$, 
    \begin{align}
        \label{eq:general_resource_state}
        \ket{\calR^\ell(H)} & = \frac{1}{\calN} \sum_{k=0}^{\ell} c_k \sum_{\mathbf{y} \in \mathbb{F}_2^m}  \beta_{G}^{(k)}(\mathbf{y}) \ket{\mathbf{y}} ,
    \end{align}
    is a $q$-ary matrix product state with open boundary conditions of bond dimension $D = \ell + 1$, where \begin{align}
        q = \underset{1 \leq t \leq r}{\max} 2^{m_t} .
    \end{align}
\end{lemma}
\begin{proof}
First, let 
\begin{align} \label{eq:MPS_reduction_vLvR}
    \vec v_L = (1,0,\dots,0), \quad \vec v_R = (c_0,c_1, \dots, c_\ell) .
\end{align}
Next, we construct the matrices $A^{(t)}[\mathbf{y}^{(t)}]$ element-wise, and only write down certain entries of the matrices. 
The value of the remaining entries will not be relevant---they are forced to vanish because they will at some point be multiplied by 0---but by convention we define any entry not explicitly defined to be 0. 
Define indices $K_0, K_1, \dots, K_{r-1}$, $K_r = k$ (where $k$ is also just an arbitrary index) which all range from $0$ to $\ell$, and let $\kappa_t = K_t - K_{t-1}$ for $t$ from $1$ to $r$. 
We construct \begin{align} \label{eq:MPS_reduction_matrix_def}
    (A^{(t)}[\mathbf{y}^{(t)}])_{K_{t-1}, K_t} & = \begin{cases}
        \binom{K_{t}}{K_{t-1}} \cdot \beta^{(\kappa_t)}_{C_t}(\mathbf{y}^{(t)}) & \text{ if } \kappa_t \geq 0, \\
        0 & \text{ otherwise.}
    \end{cases} 
\end{align}
Then \begin{align}
    \mathbf{v}_L^\intercal A^{(1)}[\mathbf{y}^{(1)}] \cdots A^{(r)}[\mathbf{y}^{(r)}] \mathbf{v}_R & = \sum_{K_0, \dots, K_r} (v_L^\intercal)_{K_0} (A^{(1)}[\mathbf{y}^{(1)}])_{K_{1} K_2} \cdots (A^{(r)}[\mathbf{y}^{(r)}])_{K_{r-1} K_r} (\mathbf{v}_R)_{K_r} \\
    & = \sum_{K_1, \dots, K_{r-1}, k} (A^{(1)}[\mathbf{y}^{(1)}])_{K_{1} K_2} \cdots (A^{(r)}[\mathbf{y}^{(r)}])_{K_{r-1} k} c_{k} \\
    & = \sum_{k, K_1, \dots, K_{r-1}} c_{k} \binom{K_1}{0} \beta_{C_1}^{(\kappa_1)}(\mathbf{y}^{(1)}) \cdots \binom{K_r}{K_{r-2}} \beta_{C_r}^{(\kappa_r)}(\mathbf{y}^{(r)}) \\
    & = \sum_{k, K_1, \dots, K_{r-1}} c_{k} \binom{k}{\bkappa} \left[\beta_{C_1}^{(\kappa_1)}(\mathbf{y}^{(1)}) \cdots \beta_{C_r}^{(\kappa_r)}(\mathbf{y}^{(r)}) \right] \label{eq:MPS_reduction_calculation_last_line}
\end{align}
where the last line follows the telescoping structure of binomial coefficients \begin{align}
    \prod_{t=1}^{r} \binom{K_t}{K_{t-1}} & = \prod_{t=1}^{r} \frac{K_t!}{K_{t-1}! \kappa_t!} = \frac{K_r!}{\kappa_1! \cdots \kappa_r!} = \binom{k}{\bkappa} .
\end{align}
The sum in Eqn.~(\ref{eq:MPS_reduction_calculation_last_line}) goes over $K_1, \dots, K_{r-1}$ all from $0$ to $\ell$, but unless $K_t \geq K_{t-1}$ for all $t$, the summand is zero by Eqn.~(\ref{eq:MPS_reduction_matrix_def}). 
Thus, we can equivalently sum over the differences $\kappa_1, \dots, \kappa_{r-1}$, and maintain the final sum $K_r = k$ from $0$ to $\ell$. 
Furthermore, because all the $K_t$'s must be in increasing order for a nonzero summand and $K_r = k$, we need only keep terms for which $\sum_{t=1}^{r} \kappa_t = \sum_{t=1}^r K_{t} - K_{t-1} = K_r = k$. 
That is, we can sum over $\bkappa$ such that $\bkappa \vDash_r k$. 
These observations yield
\begin{align}
    \mathbf{v}_L^\intercal A^{(1)}[\mathbf{y}^{(1)}] \cdots A^{(r)}[\mathbf{y}^{(r)}] \mathbf{v}_R & = \sum_{k=0}^\ell \; \sum_{\bkappa : \bkappa \vDash_r k} c_{k} \binom{k}{\bkappa} \prod_{t=1}^r \beta_{C_t}^{(\kappa_t)}(\mathbf{y}^{(t)}) \\
    & = \sum_{k=0}^\ell c_k \left[ \sum_{\bkappa : \bkappa \vDash_r k} c_{k} \binom{k}{\bkappa} \prod_{t=1}^r \beta_{C_t}^{(\kappa_t)}(\mathbf{y}^{(t)}) \right] \\
    & = \sum_{k=0}^\ell c_k  \beta_{G}^{(k)}(\mathbf{y})
\end{align}
by Lemma~\ref{lemma:beta_function_reduced_form}. 
Therefore, \begin{align}
    \ket{\calR^\ell(H)} & = \frac{1}{\calN} \sum_{k=0}^{\ell} c_k \sum_{\mathbf{y} \in \mathbb{F}_2^m}  \beta_{G}^{(k)}(\mathbf{y}) \ket{\mathbf{y}} \\
    & = \frac{1}{\calN} \sum_{\mathbf{y} \in \F_2^m} \mathbf{v}_L^\intercal A^{(1)}[\mathbf{y}^{(1)}] \cdots A^{(r)}[\mathbf{y}^{(r)}] \mathbf{v}_R \ket{\mathbf{y}} ,
\end{align}
and we can absorb the normalization into $\mathbf{v}_L$ by defining $\mathbf{v}_L = (1/\calN, 0, \dots, 0)$ instead of $\mathbf{v}_L = (1, 0, \dots, 0)$ as we initially did Eqn.~(\ref{eq:MPS_reduction_vLvR}).

We conclude with remarks about the dimensions of the MPS. Eqn.~(\ref{eq:MPS_reduction_matrix_def}) defines the matrix from index $0$ up to at most $K_r = k \leq \ell$, so the matrices are $D \times D$ where $D = \ell + 1$. Also, by definition, $\dim \mathbf{y}^{(t)} = m_t$, so $q$ is as claimed.
\end{proof}

\begin{theorem}[MPS pilot state preparation algorithm] \label{thm:resource_state_algprithm_MPS}
    Let $H = \sum_{i=1}^m v_i P_i$, where $v_i \in \set{\pm 1}$ and the $P_i$ are $n$-qubit Pauli operators, and let $G = (V, E)$ be the corresponding anticommutation graph as in Definition~\ref{def:commutativity_graph}. 
    Assume that $m = \poly(n)$.
    Suppose that $G = \bigsqcup_{t=1}^r C_t$, where $C_t = (V_t, E_t)$ are the connected components of $G$. 
    Let $m_t = |V_t|$ be the size of the $t$-th connected component.
    If $\max_{t} m_t = \calM$, then we can prepare the general pilot state for a polynomial $\calP$ of degree $\ell \leq n$,
    \begin{align}
        \ket{\calR^\ell(H)} & = \frac{1}{\calN} \sum_{k=0}^{\ell} c_k \sum_{\mathbf{y} \in \mathbb{F}_2^m}  \beta_{G}^{(k)}(\mathbf{y}) \ket{\mathbf{y}} ,
    \end{align}
    using a classical pre-processing step which takes time $T_{\text{compute}} = n^{O(\calM)}$ and a subsequent quantum algorithm which takes time $T_{\text{prepare}} = O(\poly(n) \cdot 2^\calM)$.
\end{theorem}

\begin{proof}
By Lemma~\ref{lemma:resource_state_is_MPS}, we may express the pilot state as a $q$-ary MPS
\begin{align}
     \ket{\calR^\ell(H)} & = \frac{1}{\calN} \sum_{\mathbf{y} \in \F_2^m} \mathbf{v}_L^\intercal A^{(1)}[\mathbf{y}^{(1)}] \cdots A^{(r)}[\mathbf{y}^{(r)}] \mathbf{v}_R \ket{\mathbf{y}}
\end{align}
of bond dimension $D = \ell + 1$ and $q = \max_{1 \leq t \leq r} 2^{m_t}$, where the matrices $A^{(t)}[\mathbf{y}^{(t)}]$ depend on the $\beta$-function, where each nonzero entry is of the form \begin{align}
    g_t(K_t, K_{t-1}) = \binom{K_{t}}{K_{t-1}} \cdot \beta^{(\kappa_t)}_{C_t}(\mathbf{y}^{(t)}) 
\end{align}
by Eqn.~(\ref{eq:MPS_reduction_matrix_def}). The $\beta$-function in this term is given by \begin{align}
    \beta^{(\kappa_t)}_{C_t}(\mathbf{y}^{(t)}) = \sum_{\substack{\bmu^{(t)} : \bmu^{(t)} \vDash_{m_t} \kappa_t\\\bmu^{(t)} \text{ (mod 2)} = \mathbf{y}^{(t)}}} \binom{\kappa_t}{\bmu^{(t)}} \alpha_{C_t}(\bmu^{(t)}) .
\end{align}
We first classically compute $g_t(K_t, K_{t-1})$. The normalization factor is a binomial coefficient where $K_{t} = O(k) = O(\ell) = O(n)$ for all $t$, so it can be computed in $\poly(n)$ time. 
For any graph $G'$ with $m$ nodes, the runtime of the $\beta_{G'}^{(k)}(\mathbf{y})$ function computation algorithm in Theorem~\ref{thm:computing_beta} increases with $k$ and $m$, and decreases with $w = |\mathbf{y}|$. 
Therefore, the worst case is when $k$ and $m$ are both at their largest. 
In our case, $m_t \leq \calM$, and $\kappa_t \leq k \leq n$, so the hardest computation is for a function of the form $\beta_{C}^{(n)}(0)$, where $C$ is some arbitrary graph with $\calM$ nodes.
Computing this function takes time \begin{align}
    O\left( \binom{\frac{n}{2} + \calM - 1}{\calM - 1} \cdot n \calM n^\calM \right) = O\left( \left(\frac{n}{2}^{\calM}\right) \cdot n \calM \cdot n^\calM \right) = n^{O(\calM)} .
\end{align}
There are $r = O(m)$ such matrices, each having $O(q) = O(2^\calM)$ choices of $\mathbf{y}^{(t)}$, and each matrix has $D^2 = O(\ell^2) = O(n^2)$ entries. 
Hence, the runtime to compute the entire set of matrices is $O(mn^2 2^{\calM} n^{O(\calM)}) = n^{O(\calM)}$.
We must also compute the vectors $\mathbf{v}_L$ and $\mathbf{v}_R$. 
By expanding $\calP(x) = \sum_{k=0}^{\ell} c_k x^k$, we showed in Lemma~\ref{lemma:resource_state_is_MPS} that $\mathbf{v}_R = (c_0, \dots, c_\ell)$. 
Finally, $\mathbf{v}_L = (1/\calN, 0, 0, \dots, 0)$, so the final task is to compute the normalization $\calN$. Note that \begin{align}
    \calN^2 & = \sum_{\mathbf{y} \in \F_2^m} \left( \mathbf{v}_L^\intercal A^{(1)}[\mathbf{y}^{(1)}] \cdots A^{(r)}[\mathbf{y}^{(r)}] \mathbf{v}_R \right)^2 \\
    & = \sum_{\mathbf{y} \in \F_2^m} \mathbf{v}_R^\intercal A^{(r)}[\mathbf{y}^{(r)}]^\intercal \cdots A^{(1)}[\mathbf{y}^{(1)}]^\intercal \mathbf{v}_L \mathbf{v}_L^\intercal A^{(1)}[\mathbf{y}^{(1)}] \cdots A^{(r)}[\mathbf{y}^{(r)}] \mathbf{v}_R .
\end{align}
Summing over $\mathbf{y}$ is equivalent to summing over all of $\mathbf{y}^{(1)}, \dots, \mathbf{y}^{(r)}$. 
We may distribute the first sum over $\mathbf{y}^{(1)}$ to calculate $A_1 = \sum_{\mathbf{y}^{(1)}} A^{(1)}[\mathbf{y}^{(1)}]^\intercal \mathbf{v}_L \mathbf{v}_L^\intercal A^{(1)}[\mathbf{y}^{(1)}]$ in $\poly(D) \cdot 2^{m_1}$ time. 
Once that is calculated, we can distribute the sum over $\mathbf{y}^{(2)}$ to calculate $A_2 = \sum_{\mathbf{y}^{(2)}} A^{(2)}[\mathbf{y}^{(2)}]^\intercal A_1 A^{(2)}[\mathbf{y}^{(2)}]$ in $\poly(D) \cdot 2^{m_2}$ time. 
We repeat until we have calculated $A_r$, at which point the remaining quantity is $\mathbf{v}_R^\intercal A_r \mathbf{v}_R$, which we compute in $\poly(D)$ time. Thus, the runtime for the normalization is $O(m \cdot \poly(D) \cdot 2^{\calM}) = O(\poly(n) \cdot 2^{\calM})$. 
In sum, the runtime to compute the classical representation of the MPS is \begin{align}
    T_{\text{compute}} = n^{O(\calM)} + O(\poly(n) \cdot 2^{\calM}) = n^{O(\calM)} .
\end{align}
Once we have computed the matrices and vectors representing the MPS, we can use the fact that a $q$-ary MPS (with $r = O(n)$ matrices) with open boundary conditions of bond dimension $D = \poly(n)$ can be prepared by a quantum algorithm running in time $T_{\text{prepare}} = O(\poly(n) \cdot q)$~\cite{prep_mps_1,prep_mps_2,prep_mps_3}.
\end{proof}

As a result of the general pilot state preparation algorithm in Theorem~\ref{thm:resource_state_algprithm_MPS}, the pilot state can always be efficiently prepared so long as the connected components of the anticommutation graph are sufficiently small.
\begin{corollary}[Efficient and nearly efficient pilot state preparation for sparse anticommutation graphs] \label{cor:resource_state_sparse_graph}
    Let $H$ and $G$ be as in Theorem~\ref{thm:resource_state_algprithm_MPS}. 
    Assume that $m = \poly(n)$.
    Suppose that $G$ decomposes into connected components $C_t$ having vertex sets $V_t$, and let $\max_{t} |V_t| = \calM$. 
    The following guarantees hold.
    \begin{enumerate}
        \item[(1) ] If $\calM = O(1)$, then we can prepare the pilot state $\ket{\calR^{\ell}(H)}$ from Eqn.~(\ref{eq:general_resource_state}) in quantum time $\poly(n)$.
        \item[(2) ] If $\calM = O(\log n)$, then we can prepare the state using a combination of a classical pre-processing step running in $n^{O(\log n)}$ time, and a quantum algorithm running in $\poly(n)$ time.
        \item[(3) ] If $\calM = O(\polylog(n))$, then we can prepare the state in quantum time $O(2^{\polylog(n)})$.
    \end{enumerate}
\end{corollary}
This result shows that for a sufficiently sparse anticommutation graph, we get either an efficient or nearly efficient---in the sense of a quasipolynomial time algorithm of runtime $2^{\polylog(n)}$---algorithm to prepare the pilot state.

\section{Applications of HDQI to non-commuting Hamiltonians}
\label{sec:applications_non-commuting}

We next give concrete non-commuting Hamiltonians  which are optimizable by the general formulation of HDQI.
In each case, the approach is to begin with a completely commuting Hamiltonian (i.e. the anticommutation graph is completely disconnected), and then to introduce random defects in such a way that the graph does not gain large connected components.

\subsection{Connected component sizes for sparse local random graphs}
We begin by developing a general formalism for analyzing the size of connected components of random graphs.
This technique is inspired by the proof strategies devised by Erdős and Rényi for uniformly random graphs---$n$ nodes with every edge included independently with probably $p$~\cite{erdds1959random}.
We here generalize parts of their proof techniques to give guarantees for random \emph{local} graphs; that is, random graphs with a guaranteed upper bound on the degree.
To enforce locality, we follows a two-level sampling approach.
First, a random graph with bounded degree is chosen; we call this graph the \emph{template graph}.
Then, we choose randomly which edges in the template graph to keep or discard, producing the final random graph.

\begin{definition}[Template distributions] \label{def:template_distributions}
    Let $\calG_m$ be a distribution over graphs with $m$ nodes.
    We say that $\widetilde{\calG}_m$ is a $(K, \calD_m)$ \emph{template distribution} for $\calG_m$ if almost surely a sample $\widetilde{G} \sim \calG_m$ has maximum degree $K$, and there exists a sampling process of the following form, whose output distribution is $\calG_m$.
    First, sample $\widetilde{G} \sim \widetilde{\calG}_m$.
    Next, choose according to some distribution $\calD_m$ over $[m]$ a set of nodes to color red; all other nodes are colored black.
    Finally, construct $G$ by deleting all edges in $\widetilde{G}$ which connect nodes of the same color.
\end{definition}

Template distributions enable us to create sparse local random graphs of interest by first choosing a random local connectivity structure---which nodes in the graph that one \emph{could} connect such that the maximum degree is bounded by a desired number---and then deciding by some random process exactly which of those edges to actually include in the final graph.
The simplest way to partition nodes into the red and black sets is to start with all nodes black and then, independently for each node, color it red with probability $p$.
This is the only type of node-coloring distribution $\calD_m$ we consider, and thus in this case we refer to $\widetilde{\calG}_m$ as a $(K, p)$ template distribution.
We next show that if $p$ is a sufficiently small constant depending on $K$, then graphs drawn from $\calG_m$ have small connected components with high probability.

\begin{lemma}[Connected component size of random graphs with templates] \label{lemma:connected_components_size_template}
    Let $(\calG_m)_m$ be a family of distributions over graphs with $m$ nodes, which has a corresponding $(K, p)$ template distribution $\widetilde{\calG}_m$.
    Then if $p < 1/K^2$, the largest connected component of a random graph $G \sim \calG_m$ has size $O(\log m)$ with probability $1 - \frac{1}{\Omega(\poly(m))}$.
\end{lemma}
\begin{proof}
We first sketch our proof technique.
We probe this random graph by fixing a starting node $v$ and running a breadth-first search (BFS) starting from that node.
That is, we search for connected neighbors of $v$, add them to a queue, then take the next node off the queue; we continue until the queue is empty.
The queue empties precisely when we have explored the entire connected component contained in $v$.
We will show that with high probability, we explore less than $\sim \log n$ nodes, which implies that that $v$ is contained in a component of size $\gg \log n$ is very small.
From there, we union bound over all $m$ possible choices of $v$.
Within each fixed node exploration, we show that the distribution of the number of connected components is dominated stochastically by a simple binomial distribution, and apply Chernoff concentration to bound the size of the exploration.

Sample $G \sim \calG_m$ with vertices $V$ and edges $E$.
Fix a node $v \in V$. 
We run a BFS with $v$ as the root.
More precisely, let $\calQ$ be a queue which initially only contains $v$.
At step $t$, we dequeue one element $u$ from $\calQ$, mark $u$ as \texttt{explored}, then add all un-\texttt{explored} neighbors of $u$ to $\calQ$.
This process continues until the queue is empty.
We will show that the queue empties quickly with high probability.

Let $X_t$ be the number of newly enqueued vertices.
Then at time $t$, the size of the queue is precisely $1 - t + \sum_{j=1}^t X_j$ (with the convention that $t$ starts at 1).
We wish to bound the probability that the size $S_v$ of the connected component containing $v$ is above some number $s$.
Since we dequeue exactly one node every step, $S_v$ is exactly the number of steps the BFS runs before the queue empties and the algorithm terminates.
Hence, if $S_v > t$, then the queue must not be empty at time $t$.
That is, $1 - t + \sum_{j=1}^t X_j \geq 1$, or equivalently $\sum_{j=1}^t X_j \geq t$.
It is therefore sufficient to show that $\Pr[\sum_{j=1}^t X_j \geq t]$ decays exponentially with $t$.

We employ two techniques to achieve this bound.
First, we observe that the coloring process can be done dynamically, in the sense that there is no need to decide the color of a node $u$ until the first time we encounter $u$ as a neighbor to the current node on the BFS.
The second is stochastic domination; that is, if we wish to upper bound the probability of the event $A = \set{\sum_{j=1}^t X_j \geq t}$, we are always free to instead use a different event which has a greater probability.
In our case, the distribution of $\sum_{j=1}^t X_j$ is complicated, but we can dominate it with a simpler random variable.

Assume that the fixed starting node $v$ is black; we remove this assumption later.
Since $\calG_m$ has a $(K, p)$ template distribution $\deg(v) \leq K$ almost surely, so there are at most $K$ nodes to which $v$ \emph{could} be connected.
Let these nodes be $w_1, \dots, w_{K'}$, where $K' \leq K$.
In reality, any $w_i$ will only actually be connected to $v$ if the $w_i$ is colored red, which occurs with probability $p$.
Amongst $w_1, \dots, w_{K'}$, we now generate all the colors.
The number of red nodes, and thus $X_1$, is distributed as $\Bin(K', p)$, which is stochastically dominated by $\Bin(K, p)$.
Now, all of the nodes just added to the queue are red, and they are allowed to connect to only black nodes.
We could repeat this sampling process, but nodes are black with probability $1 - p$ where $p$ is small, so there is a high probability that most of the potential connections become genuine connections.
We bound over this complication by using stochastic domination over the worst case in which the red node connects to all its potential neighbors, of which there are at most $K$.
Thus, instead of adding $\Bin(K, p)$ nodes to the queue, we add at most $K \cdot \Bin(K, p)$ nodes, and can assume these new nodes are always black.
In future steps, i.e. for some $t \geq 2$, some of the nodes we encounter as we examine neighbor sets may have already been colored.
If such a node is red, it will be connected to our node at hand $u$.
However, the fact that it is red and colored previously means it must have been connected to a node considered in a previous step of the BFS, and therefore already accounted for in our counting of the connected component size.
Therefore, the connected component can only get larger if we assume that at each step every neighbor has not yet been colored, so $K \cdot \Bin(K, p)$ stochastically dominates our random variable.
Consequently, this sequence of stochastic dominations imply that \begin{align}
    \Pr\left[\sum_{j=1}^t X_j \geq t \right] \leq \Pr \left[K \cdot \sum_{j=1}^t Y_j \geq t \right] = \Pr\left[Z_t \geq t/K\right] ,
\end{align}
where $Y_j \stackrel{\text{iid}}{\sim} \Bin(K, p)$ and $Z_t \sim \Bin(tK, p)$.
The last inequality uses the fact that a sum of $t$ independent $\Bin(K, p)$ random variables is equal in distribution to a single $\Bin(tK, p)$ random variable.
Now that we have bounded our probability of interest by a simple Binomial probability, we apply the Chernoff bound, namely that for $X \sim \Bin(n, q)$ and $\delta > 0$, \begin{align}
    \Pr[X \geq (1 + \delta) nq] \leq \exp{- \frac{\delta^2}{2 + \delta} nq} .
\end{align}
In our case, $\mathbf{E}[Z_t] = tKp$.
Since $p < 1/K^2$ is a constant, we may write $p = \frac{1}{K^2 (1 + \delta)}$ for some $\delta > 0$.
Thus, $t/K = tK (1+\delta) p$, and
\begin{align}
    \Pr[Z_t \geq t/K] & = \Pr[Z_t \geq (1+\delta) tKp] \leq \exp{- \left(\frac{\delta^2 K p}{2+\delta}\right) t} = \exp{- c_{K, \delta} t} ,
\end{align}
where $c_{K, \delta} = \delta^2 K p/ (2+\delta) = \delta^2 / K(1+\delta)(2+\delta)$.
The condition $Z_t \geq t/K$ is implied by $S_v > t$, i.e. the size of the connected component containing $v$ being larger than $t$.
Consequently, $\Pr[S_v > t] \leq \exp{- c_{K, \delta} t}$.
Now, let $S$ be the size of the largest connected component in $G$. Then by a union bound, \begin{align}
    \Pr[S > t] & = \Pr[\exists v \in G \,:\, S_v > t] \\
    & \leq m \cdot \Pr[S_v > t] \\
    & \leq m \cdot \exp{- c_{K, \delta} t} \\
    & = \exp{- (c_{K, \delta} t - \ln m)} .
\end{align}
By choosing $t = \bar{c}_{K, \delta, D} \ln m = O(\log m)$, where $\bar{c}_{K, \delta, D} = (D+1)/c_{K, \delta}$, the right-hand side is bounded by $1/n^D$. Since $D$ can be made any positive constant, the probability that the largest connected component has size $O(\log m)$ can be made $1 - 1/\poly(m)$ as claimed.
\end{proof}

\subsection{Semiclassical spin glasses}
We now apply the formalism of Lemma~\ref{lemma:connected_components_size_template} to some concrete Hamiltonians.
We first give a model which is not a defected version of a Hamiltonian discussed previously in this work; rather, it is a random \emph{classical} spin glass with random quantum defects.
Starting from a certain type of classical spin glass, which has a Hamiltonian consisting only of $Z$-type terms, we turn each term into a $X$-type term independently with probability $p$.
When $p = 1/2$, this Hamiltonian takes the form of a quantum spin glass.
For a general constant $p \in (0, 1/2)$, the model can be interpreted as a ``semiclassical spin glass" or a ``classical spin glass with quantum defects".
Formally, let $B^\intercal_0 \in \F_2^{n \times m}$ be a random $(a, b)$-sparse binary matrix.
That is, $B^\intercal_0$ is a uniformly random sample from the set of $n \times m$ binary matrices for which every column has weight $a$ and every row has weight $b$.
Let $\mathbf{h}_i \in \F_2^n$ be the $i$th column of $B^\intercal_0$.
We define a Hamiltonian $H(p)$ as \begin{align}
    H(p) = \sum_{i=1}^m v_i P(\mathbf{h}_i) ,
\end{align}
where independently for each $i$, $P(\mathbf{h}_i) = Z^{\mathbf{h}_i}$ with probability $1-p$ and $P(\mathbf{h}_i) = X^{\mathbf{h}_i}$ with probability $p$.
The $v_i \in \set{\pm 1}$ can be drawn from any distribution desired.
We denote this model the $p$-semiclassical spin glass with sparsity $(a, b)$.

\THMspinglass*


\begin{proof}
We show two properties, that $\Symp(H(p))$ can be efficiently decoded up to errors of weight $d(a, b, c)$ and that the anticommutation graph of $H(p)$ has no large connected components---both with high probability.
Let $B_H^\intercal \in \F_2^{2n \times m}$ be the parity check matrix of $\Symp(H(0))$.
To the former, when $p = 0$, \begin{align}
    B_H^\intercal = \begin{bmatrix}
        0 \\ B^\intercal_0
    \end{bmatrix} ,
\end{align}
where $B^\intercal_0 \in \F_2^{n \times m}$ is a random parity check matrix with column-weight $a$ and row-weight $b$.
By the results of Appendix~\ref{sec:app-expander_graphs}, with probability $1 - n^{-\Omega(1)}$ this code is efficiently decodable for all errors up to weight $dm$, for some threshold constant $d(a, b, c) \in (0, 1)$.
The Pauli terms which are randomly chosen to become $X$-type terms flip to the top of the symplectic code.
That is, if $P(\mathbf{h}_i)$ becomes $X$-type, then the $i$th column of $B_H^\intercal$ is given by $\begin{bmatrix}
    \mathbf{h}_i \\ 0
\end{bmatrix}$ as opposed to $\begin{bmatrix}
    0 \\ \mathbf{h}_i
\end{bmatrix}$.
The efficient decodability of $\Symp(H(0))$ arises from the fact that its Tanner graph has, with high probability, a certain property known as high expansion which implies efficient decodability (see Appendix~\ref{sec:app-expander_graphs}).
In the Tanner graph representation, the action of moving some columns from the $Z$ to the $X$ part of the symplectic representation corresponds to making a copy of all the check nodes, and then moving the selected data nodes (each of which corresponds to a column) to connect identically with this second copy instead of the first.
Such an action can only increase the expansion property, and thus $\Symp(H(p))$ remains efficiently decodable.

Next we show that the anticommutation graph of $H(p)$ has connected components only of size $O(\log n)$.
This completes the proof by way of Theorem~\ref{thm:main_reduction_general_case} and Corollary~\ref{cor:resource_state_sparse_graph}.
Initially, with $H(0)$, the anticommutation graph $G_0$ is completely disconnected.
Then, independently for each $Z$-type Pauli, we decide with probability $p$ whether or not to flip it into a $X$-type Pauli.
This is equivalent to deciding for each node whether or not to color it red (with probability $p$) or black ($1-p$).
At that point, we connect nodes which anticommute, forming the anticommutation graph $G$ of $H(p)$.
Note that every node corresponds to a column in $B_H^\intercal$ which has weight $a$.
Each row has weight $b$, so any given column has overlapping support with at most $a(b-1)$ other columns.
Thus, the degree of every node in $G$ is at most $K = a(b-1)$.
The distribution over anticommutation graphs $G$ is therefore precise that of a random graph with a $(K, p)$ template, so by Lemma~\ref{lemma:connected_components_size_template} the connected components are of size $O(\log n)$ with probability $1 - n^{-\Omega(1)}$ as desired.
\end{proof}

\begin{corollary}[Gibbs sampling semiclassical spin glasses at constant temperature]
    Let $H(p)$ be a $p$-semiclassical spin glass as in Theorem~\ref{thm:semiclassical_spin_glass_HDQI}. 
    If $p < 1/K^2$, there exists a constant threshold inverse temperature $\beta^* > 0$ such that for every $\beta < \beta^*$, there is an algorithm running in $n^{O(\log n)}$ classical pre-processing time and $\poly(n)$ quantum time which prepares the Gibbs state with inverse temperature $\beta$ up to trace distance $\exp{-\Omega(n)}$.
\end{corollary}
\begin{proof}
    Follows immediately from Theorems~\ref{thm:gibbs_explicit} and \ref{thm:semiclassical_spin_glass_HDQI}.
\end{proof}

\begin{remark}
    Note that with constant probability (which can be made arbitrarily close to $1$ by increasing the maximum allowed size of the connected components), the maximum connected component size is $O(1)$.
    Hence, if we replace the probability bound of $1 - 1/n^{\Omega(1)}$ in Theorem~\ref{thm:semiclassical_spin_glass_HDQI} with $1 - \epsilon$ for an arbitrarily small choice of $\epsilon$, we instead get an efficient quantum algorithm by Case 1 of Corollary~\ref{cor:resource_state_sparse_graph}.
\end{remark}

\begin{remark}
    The above strategy does not immediately generalize to a defected version of the random $k$-local commuting Hamiltonian of Section~\ref{sec:commuting-applications}.
    However, a slightly more restrictive version of the model---namely one where every term is $k$-local and also every qubit is supported in at most $d$ terms---would immediately generalize to the defected version, and the above theorem would hold just as well for this model.
\end{remark}

\subsection{Commuting Hamiltonians with random non-commuting defects}

The template graph formalism developed above applies to a wide range of physically relevant Hamiltonians.
In particular, if we start with \emph{any} commuting Hamiltonian which is amenable to HDQI (i.e. $\Symp(H)$ is an efficiently decodable code up to error weight $\Omega(n)$) such that \begin{enumerate}
    \item[(a) ] the Hamiltonian has bounded locality interaction, i.e. every Pauli term has overlapping support with at most $O(1)$ other terms, and
    \item[(b) ] the defects introduced do not destroy the efficient decodability of $\Symp(H)$ up to linear error weight.
\end{enumerate}
The former condition is usually satisfied by physically relevant models, e.g. every Pauli has $O(1)$ weight and has overlapping support with only $O(1)$ other Paulis.
The latter condition depends more sensitively on the exact structure of the decoder applied to the particular model, but is likely generally satisfied as well.
To illustrate the generality of the template graph formalism, we revisit the model of local nearly independent Hamiltonians from Section~\ref{sec:commuting-applications} and show that even with defects HDQI can prepare Gibbs states up to some constant temperature.
Here, by local, we mean that every Pauli term has overlapping support with at most $O(1)$ other terms.
While in the commuting case we were able to Gibbs sample at arbitrary temperature, this required a different state preparation procedure.
While it is plausible that a generalized version of this different procedure will hold also for non-commuting Hamiltonians, we leave this for future work.
\begin{theorem}[Gibbs sampling of defected high-distance nearly independent Hamiltonians] \label{thm:gibbs_defected_perfectly_decodable}
    Let $H_0 = \sum_{i=1}^m v_i P_i$ be a commuting Pauli Hamiltonian on $n$ qubits with $v_i \in \set{\pm 1}$ and $m = O(n)$, such that $\Symp(H)$ encodes $k = O(1)$ logical bits and every relation (set of Paulis that multiply to $\mathbb{1}$) involves $\Omega(m)$ terms.
    Define $H$ by applying a random Pauli to each term independently with probability $p$, where the random Pauli is allowed to be drawn from any distribution as long as it does not introduce new independent relations among the Hamiltonian terms.
    Suppose there exists a constant $K$ such that for every $i \neq j$, $|\supp(P_i) \cap \supp(P_j) | \leq K$.
    Then there exists a threshold inverse temperature $\beta^*$ such that for any $\beta \leq \beta^*$ and any $p < 1/K^2$, HDQI prepares the Gibbs state of $H$ up to $\exp{-\Omega(n)}$ trace distance.
    With probability $1 - 1/n^{\Omega(1)}$, the algorithm runs in $n^{O(\log n)}$ classical pre-processing time and $\poly(n)$ quantum time.
\end{theorem}

The proof is a substantially simpler version of the argument given for the semiclassical spin glass, so we briefly sketch it here.
As before, the assumption that the Pauli terms have bounded interaction and that $p$ is below threshold implies that connected components of the anticommutation graph are $O(\log n)$ in size with high probability.
As for the decoder, since the distance is linear and there are only $k = O(1)$ logical bits, for a sufficiently small constant $\alpha$, any weight-$\alpha n$ error on $\Symp(H_0)$ can be decoded efficiently by direct Gaussian elimination. 
This fact holds true as well for $\Symp(H)$, since defects preserve independence relations by assumption.
Small connected components and efficient decoding are the only requirements for HDQI to succeed, so the theorem is proven.
Like before, should one wish for a completely efficient algorithm, one can do so by replacing the probability $1 -1/n^{\Omega(1)}$ with $1 - \epsilon$, where $\epsilon$ can be chosen to be an arbitrarily small constant.

As a concrete example, consider the 2D toric code wherein every vertex, which typically hosts a 4-local $X$ operator, has a probability $p$ of having a random Pauli (supported on that vertex) applied to it.
This may cause it to anticommute with at most 8 other operators, namely the 4 plaquette operators and 4 vertex operators it touches.
By Theorem~\ref{thm:gibbs_defected_perfectly_decodable}, so long as $p < 1/64$, it is still possible via HDQI to prepare the Gibbs state up to a constant inverse temperature of this defected toric code Hamiltonian, with high probability.

\section{Beyond Pauli Hamiltonians}
\label{sec:beyond_paulis}
While we have exclusively focused on Pauli Hamiltonians so far, the framework of HDQI is not limited to this setting.
This section extends HDQI from the Pauli group, acting on qubits, to the Weyl group (also known as the discrete Heisenberg group or generalized Pauli group), acting on qudits. 
The corresponding decoding problem is a syndrome decoding problem over prime fields. 
For simplicity, we give a detailed analysis only on commuting $n$-qudit Weyl Hamiltonians.
We, however, expect that our results for non-commuting Hamiltonians, i.e. Theorem~\ref{thm:main_reduction_general_case} and Corollary~\ref{cor:resource_state_sparse_graph}, hold just as well in the $\F_p$ generalization, and at the end of this section we sketch the generalized proof.
Our result is summarized in the following theorem. 
\begin{theorem}[Main reduction over $\F_p$, commuting case] \label{thm:commuting-p}
Let $p$ be prime and let $H=\sum_{i=1}^m v_i W(\mathbf{u}_i)$ be a commuting $n$-qudit Hamiltonian where $W(\mathbf{u})$ is an $n$-qudit Weyl operator, $v_i \in \set{e^{2 \pi i t / p} \,|\, 0 \leq t \leq p-1}$ is a phase, and $m=\mathrm{poly}(n)$.  
Let $\mathcal D_H^{(\ell)}$ be a decoding oracle (as in Definition~\ref{def:decoding_oracle}) for the classical linear $p$-ary symplectic code $\mathrm{Symp}_p(H)$, and let $\calP$ be any univariate  polynomial with $\deg(\calP) \leq \ell$.  
There is a quantum algorithm, running in $\poly(n)$ time, which prepares
\begin{equation}
\rho_{\mathcal P}^{(p)}(H) := \frac{\mathcal P(H)\,\mathcal P(H)^\dagger}{\mathrm{Tr}[\mathcal P(H)\,\mathcal P(H)^\dagger]}
\end{equation}
using a single call to the decoding oracle $\mathcal D_H^{(\ell)}$.
\end{theorem}

\subsection{Generalized Pauli operators and Bell states}
Fix a prime $p$ and set $\omega := e^{2\pi i/p}$. On one qudit $\mathbb C^{p}$ define for each $j \in \F_p$
\begin{equation}
X\ket{j} = \ket{j+1 \text{ (mod $p$)}}, \qquad Z\ket{j} = \omega^{j}\ket{j} .
\end{equation}
Note that both operators are traceless, just as their analogs over qubits are.
Let $X_k$ be a $n$-qudit operator with $X$ the $k$th qudit and $I$ on all others; define $Z_k$ similarly.
For $\mathbf{u} = (\mathbf a \,|\, \mathbf b) \in \mathbb F_p^{2n}$ with $\mathbf a, \mathbf b \in \mathbb F_p^n$, the $n$-qudit Weyl (generalized Pauli) operator is defined as
\begin{equation}
\label{def:Weyl_operator}
W(\mathbf{u}) := \prod_{k=1}^{n} Z_k^{a_k} X_k^{b_k}.
\end{equation}
The Weyl operators obey the the \emph{Weyl commutation relations}:
\begin{equation}
\label{eq:Weyl_commutation_relations}
W(\mathbf{u}) W(\mathbf v) = \omega^{\langle \mathbf u , \mathbf v \rangle} W(\mathbf v) W(\mathbf u), \qquad
\langle \mathbf u , \mathbf v \rangle := \mathbf a \cdot \mathbf b' - \mathbf b \cdot \mathbf a' \in \mathbb F_p,
\end{equation}
where $\mathbf v = (\mathbf a' \,|\, \mathbf b')$ and $\cdot$ is the standard dot product over $\mathbb F_p$. 
In other words, $\langle \cdot, \cdot \rangle$ is the natural generalization of the symplectic inner product to a general prime field.
On the Weyl group, we define the \emph{symplectic map}
\begin{equation}
\sym \,:\, \omega^z W(\mathbf u) \longmapsto \mathbf u \in \F_p^{2n} \quad \left(\forall z \in \F_p\right) ,
\end{equation}
which is the analog of the symplectic map of Paulis on an arbitrary prime field.
As in the Pauli case, $\sym$ is a homomorphism from the generalized Pauli group to $\F_p^{2n}$, since
\begin{equation}
\sym(\omega^z W(\mathbf u) \cdot \omega^{y} W(\mathbf v)) = \sym(W(\mathbf u)) + \sym(W(\mathbf v)) = \mathbf u + \mathbf v,
\end{equation}
where addition is performed mod $p$. 
The kernel of $\mathrm{symp}$ is precisely the center $\langle \omega \Id \rangle$ of the Weyl group.
To define HDQI, we require an analogous notion of the maximally entangled state over $\F_p$.
We define this $2n$-qudit state as
\begin{equation}
\ket{\Phi_p^n} = \frac{1}{\sqrt{p^{n}}} \sum_{\mathbf x \in\mathbb F_p^n}\ket{\mathbf x} \otimes \ket{\mathbf x} .
\end{equation}
In Appendix~\ref{app:bell_basis}, we show that, as in the case of the maximally entangled state over qubits, application of two distinct Weyl operators on $\ket{\Phi_p^n}$ yields two orthonormal states.
Hence, the set $\set{(W(\mathbf u) \otimes I) \ket{\Phi_p^n}}$ for all $\mathbf{u}$ forms an orthonormal basis for $(\C^p)^{\otimes 2n}$.
We moreover derive the $\F_p$ analog of the coherent Bell measurement, which unitarily maps $(W(\mathbf u) \otimes I) \ket{\Phi_p^n}$ to $\ket{\sym(W(\mathbf{u}))} = \mathbf{u}$.
These facts serve as the foundation for the lifting of HDQI to higher prime fields.

\subsection{Symplectic codes over \texorpdfstring{$\mathbb F_p$}{Fp}}
Given a commuting $n$-qudit Weyl Hamiltonian
\begin{equation}
H = \sum_{i=1}^{m} v_i W(\mathbf{u}_i) ,
\end{equation}
where $\mathbf{u}_i \in \F_p^{2n}$, $v_i \in \set{e^{2 \pi j / p} \,:\, 0 \leq j \leq p-1}$, and $[W(\mathbf{u}_i), W(\mathbf{u}_j)] = 0$ for all $i, j \in [m]$,
we define the symplectic code $\Symp(H)$ of $H$ as by the parity check matrix
\begin{equation}
B^\intercal := 
\begin{bmatrix} 
| & | & & | \\ 
\mathbf{u}_1 & \mathbf{u}_2 & \cdots & \mathbf{u}_m \\ 
| & | & & | 
\end{bmatrix} \in \mathbb F_p^{2n \times m}.
\end{equation}
That is, the code $\Symp(H)$ is given by
\begin{equation}
\mathrm{Symp}_p(H) := \{ \mathbf y \in \F_p^m : B^\intercal \mathbf y = 0 \}.
\end{equation}
We write $\mathrm{wt}(y)$ for the Hamming weight (number of nonzero components) of $y$.
Just as in the binary setting, for $\ell\ge 1$, we define a weight-$\ell$ decoding oracle for $H$ as a unitary that implements
\begin{equation}
\mathcal D_H^{(\ell)}: \ket{\mathbf y}\ket{B^\intercal \mathbf y} \mapsto \ket{0} \ket{B^\intercal \mathbf y}
\end{equation}
for all $\mathbf y\in\mathbb F_p^m$ with $\mathrm{wt}(y) \in [1, \ell]$.
\begin{remark}
  For $p>2, W(\mathbf u)$ is unitary but not Hermitian, and moreover $H$ need not be Hermitian since the phases $v_i$ are not necessarily real.
  Thus, to call $H$ a Hamiltonian is a slight abuse of notation.
  Nevertheless, the output state $\rho_{\mathcal{P}}^{(p)}(H) \propto \mathcal{P}(H) \mathcal{P}(H)^{\dagger}$ of HDQI is positive for any polynomial $\mathcal{P}$, so the output state is well-defined.
\end{remark}

\subsection{\texorpdfstring{$p$}-ary symmetric expansion and multinomial Dicke states}
With the $p$-ary symplectic codes defined, we next formulate a $p$-ary symmetric polynomial expansion.
For any $\mathbf{y} \in \F_p^m$, let $\mathbf{a} = (a_0, \dots, a_{p-1}) \in \Z_{\geq 0}^{p}$ be the \emph{type} of $\mathbf{y}$.
That is, $a_t = |\set{i \,:\, y_i = t}|$, and we denote $\mathsf{Type}(\mathbf{y}) = \mathbf{a}$.
We denote the set of types for all $\mathbf{y} \in \F_p^m$ as $\calT_p(m)$.

\begin{theorem}[$p$-ary symmetric polynomial expansion] \label{thm:p-ary_symmetric_expansion}
    Fix a prime $p$.
    Let $z_1, \dots, z_m$ be formal variables such that $z_i^p = 1$ for all $i \in [m]$.
    Let $\calP(x)$ be a univariate polynomial of degree $\ell$.
    Then there exists coefficients $w(\mathbf{a})$, collectively computable in time $\poly(m)$, for each $\mathbf{a} \in \calT_p(m)$, which do not depend on $\set{z_i}$, such that \begin{align}
    \mathcal P\left(\sum_{i=1}^m z_i\right) 
    & = \sum_{\mathbf{a} \in \calT_p(m)} w(\mathbf{a}) \sum_{\substack{\mathbf{y} \in \F_p^m \\ \mathsf{Type}(\mathbf{y}) = \mathbf{a}}} z_1^{y_1} \cdots z_m^{y_m} .
    \end{align}
    $w(\mathbf{a}) = 0$ for all $\mathbf{a}$ such that $a_0 < m - \ell$.
\end{theorem}

\begin{proof}
Let $\mathcal P(x) = \sum_{k=0}^{\ell} c_k x^k$.
Since the $z_i$ commute and $z_i^p = \mathbb{1}$, we can expand in an analogous manner to Theorem~\ref{thm:commuting_symmetric_polynomial_expansion}:
\begin{equation}
\mathcal P\left(\sum_{i=1}^m z_i\right)
= \sum_{k=0}^{\ell} c_k \sum_{\mathbf y\in\mathbb F_p^m} M_k(\mathbf y) z_1^{y_1} \cdots z_m^{y_m} ,
\end{equation}
where $M_k(\mathbf y) = \sum_{\bmu \vDash_m k : \bmu \text{ (mod $p$)} = \mathbf{y}} \binom{k}{\bmu}$.
Since $\binom{k}{\bmu}$ is symmetric under permutations of indices of $\bmu$, $M_k(\mathbf{y})$ depends only on the type of $\mathbf{y}$, so we instead write it as $M_k(\mathbf{a})$.
Hence, the generalized expansion is given by \begin{align}
    \mathcal P\left(\sum_{i=1}^m z_i\right) & = \sum_{k=0}^{\ell} c_k \sum_{\mathbf{a} \in \calT_p(m)} M_k(\mathbf{a}) \sum_{\substack{\mathbf{y} \in \F_p^m \\ \mathsf{Type}(\mathbf{y}) = \mathbf{a}}} z_1^{y_1} \cdots z_m^{y_m} \\
    & = \sum_{\mathbf{a} \in \calT_p(m)} \left(\sum_{k=0}^{\ell} c_k M_k(\mathbf{a}) \right) \sum_{\substack{\mathbf{y} \in \F_p^m \\ \mathsf{Type}(\mathbf{y}) = \mathbf{a}}} z_1^{y_1} \cdots z_m^{y_m} .
\end{align}
We define $w(\mathbf{a})$ to be the parenthesized expression.
Using analogous techniques as those in Appendix~\ref{app:combinatorics}---replacing binomials with multinomials---we may compute each $w(\mathbf{a})$ in time $\poly(m)$.
The number of types $|\calT_p(m)|$ is the number non-negative integer vectors in $p$ dimensions whose $L^1$ norm is $m$.
This is precisely the number of $p$-weak compositions of $m$, so there are \begin{align}
    \binom{m+p-1}{p-1} = O(m^{p-1}) = \poly(m)
\end{align}
terms total.
Note that $M_k(\mathbf{y}) = 0$ if $|\mathbf{y}| > k$, so $w(\mathbf{a}) = 0$ when the number of zero elements $a_0$ is less than $m-\ell$, since for such $\mathbf{a}$, $M_k(\mathbf{a}) = 0$ for all $k \leq \ell$.
\end{proof}

As in the qubit case, the $p$-ary symmetric expansion is quantumly implemented in part by a qudit generalization of Dicke states.

\begin{definition}[Multinomial Dicke states]
For $\mathbf{a} = (a_0, \ldots, a_{p-1}) \in \calT_p(m)$ a type, the qudit Dicke state with type $\mathbf{a}$ is defined as
\begin{equation}
\ket{D_\mathbf{a}^m} := \frac{1}{\sqrt{\binom{m}{\mathbf{a}}}} \sum_{\substack{\mathbf y \in \F_p^m \\ \mathsf{Type}(\mathbf y) = \mathbf a}}\ket{\mathbf y}.
\end{equation}
\end{definition}

\subsection{Commuting HDQI algorithm for the Weyl group}
We now derive qudit Hamiltonian Decoded Quantum Interferometry on Weyl Hamiltonians of the form
\begin{equation}
    H = \sum_{i=1}^{m} v_i W(\mathbf u_i) ,
\end{equation}
where the Weyl operators satisfy $[W(\mathbf{u}_i), W(\mathbf{u}_j)] = 0 \, \forall i, j$ and $v_i \in \set{e^{2 \pi j / p} \,:\, 0 \leq j \leq p-1}$.

\begin{lemma}[Efficient pilot state preparation for commuting qudit Hamiltonians]
\label{lemma:p-ary_commuting_efficient_resource_state}
    Fix a prime $p$ and for each $\mathbf{a} \in \calT_p(m)$ let $w(\mathbf{a}) \in \R$ be given coefficient.
    There exists a quantum algorithm running in time $\poly(m)$ which prepares the $m$-qudit state \begin{align}
        \ket{\calR^\ell_{p, \text{comm}}(H)} \propto \sum_{\mathbf{a} \in \calT_p(m)} w(\mathbf{a}) \sum_{\substack{\mathbf{y} \in \F_p^m \\ \mathsf{Type}(\mathbf{y}) = \mathbf{a}}} \ket{\mathbf{y}} .
    \end{align}
\end{lemma}
\begin{proof}
First, prepare the state $\propto \sum_{\mathbf{a} \in \calT_p(m)} w(\mathbf{a}) \ket{\mathbf{a}} \otimes \ket{0^m}$, using the same manner as in that in the proof of Lemma~\ref{lemma:commuting_efficient_resource_state} but for qudits.
In essence, since there are only $\poly(m)$ nonzero amplitudes, this step can be done efficiently.
Next, apply a unitary transformation $U \,:\, \ket{\mathbf{a}} \otimes \ket{0^m} \mapsto \ket{\mathbf{a}} \otimes \ket{D^m_{\mathbf{a}}}$.
This is the qudit generalization of the Dicke state preparation step which we used in Lemma~\ref{lemma:commuting_efficient_resource_state}.
For a general $p$, it is known how to efficiently implement $U$ in polynomial depth~\cite{nepomechie2023qudit}.
\end{proof}

We now construct HDQI over $\F_p$ to prove Theorem~\ref{thm:commuting-p}.

\begin{proof}[Proof of Theorem~\ref{thm:commuting-p}]
We begin with the tripartite state
\begin{align}
    \ket{\psi_1}_{ABC} = \ket{\calR^{\ell}_{p, \text{comm}}(H)}_{A} \otimes \ket{\Phi_p^n}_{BC} ,
\end{align}
where the coefficients used to prepare the pilot state are precisely those in the symmetric expansion of $\calP$, as given in Theorem~\ref{thm:p-ary_symmetric_expansion}.
Next, let $\widetilde{v}_i := \frac{p}{2 \pi} \operatorname{arg}(v_i)$, i.e. $v_i = e^{2 \pi i \widetilde{v}_i / p}$.
We then apply $\bigotimes_{i=1}^m Z^{\widetilde{v}_i}$ to get the state \begin{align}
    \ket{\psi_2}_{ABC} = \frac{1}{\calN}  \sum_{\mathbf{a} \in \calT_p(m)} w(\mathbf{a}) \sum_{\substack{\mathbf{y} \in \F_p^m \\ \mathsf{Type}(\mathbf{y}) = a}} \left( \prod_{i=1}^m v_i^{y_i} \right) \ket{\mathbf{y}}_A \otimes \ket{\Phi_p^n}_{BC} .
\end{align}
Here $\calN$ is an appropriate normalization.
Next, let $(C^{(i)} W(\mathbf{u}_i))_{A \to B}$ be the operation of applying the Weyl operator $W(\mathbf{u}_i)$ to register $B$, controlled on the $i$th qubit of register $A$.
Apply $\prod_{i=1}^m (C^{(i)} W(\mathbf{u}_i))_{A \to B}$, yielding 
\begin{align}
    \ket{\psi_3}_{ABC} & = \frac{1}{\calN} \sum_{\mathbf{a} \in \calT_p(m)} w(\mathbf{a}) \sum_{\substack{\mathbf{y} \in \F_p^m \\ \mathsf{Type}(\mathbf{y}) = a}}  \ket{\mathbf{y}}_A \otimes \left( \prod_{i=1}^m [v_i W(\mathbf{u}_i)]^{y_i} \otimes I\right) \ket{\Phi_p^n}_{BC} \\
    & = \frac{1}{\calN} \sum_{\mathbf{a} \in \calT_p(m)} w(\mathbf{a}) \sum_{\substack{\mathbf{y} \in \F_p^m \\ \mathsf{Type}(\mathbf{y}) = a}}  \ket{\mathbf{y}}_A \otimes \left( \prod_{i=1}^m v_i^{y_i} \right) (W(\mathbf{u}(\mathbf y)) \otimes I)\ket{\Phi_p^n}_{BC} 
\end{align}
where for $\mathbf y \in \F_p^m$, we define $\mathbf{u}(\mathbf{y}) = \sum_{i=1}^m y_i \mathbf{u}_i$ (and all arithmetic done mod $p$ as usual).
We next utilize the coherent Weyl-Bell basis measurement, which unitarily maps $(W(\mathbf{u}) \otimes I) \ket{\Phi_p^n} \longrightarrow \ket{\mathbf{u}}$ as described in Theorem~\ref{thm:coherent_weyl_bell_measurement} in Appendix~\ref{app:bell_basis}.
This gives the state
\begin{align}
    \ket{\psi_4}_{ABC} = \frac{1}{\calN} \sum_{\mathbf{a} \in \calT_p(m)} w(\mathbf{a}) \sum_{\substack{\mathbf{y} \in \F_p^m \\ \mathsf{Type}(\mathbf{y}) = a}}  \ket{\mathbf{y}}_A \otimes \left( \prod_{i=1}^m v_i^{y_i} \right) \ket{\mathbf{u}(\mathbf{y})}_{BC} .
\end{align}
By Theorem~\ref{thm:p-ary_symmetric_expansion}, the only elements of the superposition which have nonzero amplitude are those for which $|\mathbf{y}| \leq \ell$. 
Hence, we may apply our decoding oracle $\calD_H^{(\ell)}$ to registers $A$ and $BC$ to uncompute register $A$ and discard it.
We are left with the state
\begin{align}
    \ket{\psi_5}_{BC} = \frac{1}{\calN} \sum_{\mathbf{a} \in \calT_p(m)} w(\mathbf{a}) \sum_{\substack{\mathbf{y} \in \F_p^m \\ \mathsf{Type}(\mathbf{y}) = a}} \left( \prod_{i=1}^m v_i^{y_i} \right) \ket{\mathbf{u}(\mathbf{y})}_{BC} .
\end{align}
Let $z_i = v_i W(\mathbf{u}_i)$.
We uncompute the coherent Weyl-Bell measurement, yielding
\begin{align}
    \ket{\psi_6}_{BC} & = \frac{1}{\calN} \sum_{\mathbf{a} \in \calT_p(m)} w(\mathbf{a}) \sum_{\substack{\mathbf{y} \in \F_p^m \\ \mathsf{Type}(\mathbf{y}) = a}} \left( \prod_{i=1}^m v_i^{y_i} \right) (W(\mathbf{u}(\mathbf y)) \otimes I)\ket{\Phi_p^n}_{BC} \\
    & = \frac{1}{\calN} \sum_{\mathbf{a} \in \calT_p(m)} w(\mathbf{a}) \sum_{\substack{\mathbf{y} \in \F_p^m \\ \mathsf{Type}(\mathbf{y}) = a}} z_1^{y_1} \cdots z_m^{y_m} \ket{\Phi_p^n}_{BC} \\
    & = \frac{1}{\calN} \calP(H) \ket{\Phi_p^n}_{BC} .
\end{align}
By discarding register $C$, we produce precisely $\rho_{\calP}^{(p)}(H)$.
\end{proof}

\begin{algorithm}      
  \caption{Hamiltonian DQI over $\F_p$ for Commuting Hamiltonians}
  \label{alg:HDQ)_qudit}
  \begin{algorithmic}[1]    
  \Require{$H = \sum_{i=1}^m v_i W(\mathbf{u}_i)$ (commuting Weyl Hamiltonian on $n$ qudits with $d = p$ for prime $p$), $|\Phi_p^n\rangle$ (maximally entangled state of $2n$ qudits), $\calD^{(\ell)}_H$ (decoding oracle), $\calP$ (polynomial of degree $\ell$).}
    \Ensure{$\rho_{\calP}^{(p)}(H) \propto \calP(H) \calP(H)^\dag$ ($n$-qudit mixed state).}
      \State Decompose $\calP(\sum_{i=1}^m z_i)$ as a sum of reduced elementary symmetric polynomials in the variables $z_i = v_i W(\mathbf{u}_i)$ by applying Theorem~\ref{thm:p-ary_symmetric_expansion}, obtaining (for $\calT_p(m)$ the set of types over $\F_p^m$)
      \begin{equation}
          \calP\left( \sum_{i=1}^m z_i \right) = \sum_{\mathbf{a} \in \calT_p(m)} w(\mathbf{a}) \sum_{\substack{\mathbf{y} \in \F_p^m \\ \mathsf{Type}(\mathbf{y}) = \mathbf{a}}} z_1^{y_1} \cdots z_m^{y_m} .
      \end{equation}
      \State Prepare the $m$-qubit pilot state in register $A$, using Lemma~\ref{lemma:p-ary_commuting_efficient_resource_state}.
      \begin{equation}
         \ket{\calR^\ell_{p, \text{comm}}(H)} \propto \sum_{\mathbf{a} \in \calT_p(m)} w(\mathbf{a}) \sum_{\substack{\mathbf{y} \in \F_p^m \\ \mathsf{Type}(\mathbf{y}) = \mathbf{a}}} \ket{\mathbf{y}} .
      \end{equation}
      Append this state to $\ket{\Phi_p^n}_{BC}$.
      \State Apply $\bigotimes_{i=1}^m Z^{\widetilde{v}_i}$ on the $A$ register, where $\widetilde{v}_i = \frac{2 \pi}{p} \operatorname{arg}(v_i)$. 
      Apply $\prod_{i=1}^m (C^{(i)} W(\mathbf{u}_i))_{A\rightarrow B}$, where each operation is a Weyl operator $W(\mathbf{u}_i)$ on register $B$, controlled by the $i$th qubit of register $A$.
      \State Apply the coherent Weyl-Bell measurement to registers $BC$, mapping the Weyl-Bell basis to its symplectic representation as specified in Theorem~\ref{thm:coherent_weyl_bell_measurement}.
      \State Uncompute register $A$ by applying $\calD_{H}^{(\ell)}$ to registers $A$ and $BC$. 
      Undo Step 4 by applying the inverse coherent Bell measurement to registers $BC$.
      \State The reduced density matrix on $B$ now equals $\rho_{\calP}^{(p)}(H)$. Output register $B$.
  \end{algorithmic}
\end{algorithm}

\begin{remark}
When $p=2$ the expansion collapses to the elementary-symmetric basis, the pilot reduces to a weighted sum of binary Dicke states, and the target becomes $\rho_{\mathcal P}(H)=\mathcal P(H)^2/\mathrm{Tr}[\mathcal P(H)^2]$, reproducing the previous description of HDQI.
\end{remark}

\begin{remark}
    The analogous robustness guarantee from Theorem~\ref{thm:robustness_commuting} holds in this more general $\F_p$ case, by a similar proof.
\end{remark}

\subsection{Generalization to non-commuting Hamiltonians}
We conclude with remarks on the non-commuting case over $\F_p$.
Much of the algorithm and proof in this case remain similar to that of Section~\ref{sec:general}.
The main issue is that the Weyl commutation relations are more general than either commuting or anticommuting, as swapping the order of multiplication can incur a phase of the form $e^{2 \pi i j / p}$.
Therefore, the sign function must be modified to account for this increased generality.
That is, a permutation $\pi$ no longer can be mapped only to a sign, but instead to a general phase.
Nonetheless, the antisymmetry character $\alpha_G$ still maintains a factorization structure over connected components of the anticommutation graph (where now we connect nodes if their corresponding Weyl operators do not commute), and a recursive computation structure.
For this reason, the computations we give in the qubit case are still generalizable to the qudit case (with mod 2 computations, e.g. in the $\beta$-function definition, lifting to mod $p$), albeit with more cumbersome notation by way of a phase function in place of a sign function.

\section*{Acknowledgments}
\addcontentsline{toc}{section}{Acknowledgments}  
We are grateful to Tommy Schuster for suggesting the use of the maximally entangled state in the HDQI reduction. 
We also thank Anurag Anshu, Ryan Babbush, Dave Bacon, Anthony Chen, Aram Harrow, Seth Lloyd, Mary Wootters, and Henry Yuen for helpful conversations and feedback.

A.S. is supported by a Google PhD Fellowship, the Simons Foundation (MP-SIP-00001553, AWH), and NSF grant PHY-2325080. 
J.Z.L. is funded in part by a National Defense Science and Engineering Graduate (NDSEG) fellowship. J.Z.L. and A.P. are supported in part by the U.S. Department of Energy,
Office of Science, National Quantum Information Science Research Centers, Co-design
Center for Quantum Advantage (C2QA) under contract number DE-SC0012704.
Y.Q. is supported by a collaboration between the US DOE and other Agencies. 
This material is based upon work supported by the U.S. Department of Energy, Office of Science, National Quantum Information Science Research Centers, Quantum Systems Accelerator. 
This work was done in part while some of the authors were visiting the Simons Institute for the Theory of Computing and the Challenge Institute for Quantum Computation at UC Berkeley.

\bibliographystyle{alpha}
\bibliography{references}

\newcommand{\etalchar}[1]{$^{#1}$}
\begin{thebibliography}{CPGSV21}

\bibitem[AE11]{AE11}
Dorit Aharonov and Lior Eldar.
\newblock On the complexity of commuting local {H}amiltonians, and tight conditions for topological order in such systems.
\newblock In {\em 2011 IEEE 52nd Annual Symposium on Foundations of Computer Science}, pages 334--343. IEEE, 2011.
\newblock arXiv:1102.0770.

\bibitem[AE15]{AE15}
Dorit Aharonov and Lior Eldar.
\newblock The commuting local {H}amiltonian problem on locally expanding graphs is approximable in {NP}.
\newblock {\em Quantum Information Processing}, 14(1):83--101, 2015.
\newblock arXiv:1311.7378.

\bibitem[AGM20]{anshu2020beyond}
Anurag Anshu, David Gosset, and Karen Morenz.
\newblock Beyond product state approximations for a quantum analogue of max cut.
\newblock {\em arXiv preprint arXiv:2003.14394}, 2020.

\bibitem[AKV18]{AKV18}
Dorit Aharonov, Oded Kenneth, and Itamar Vigdorovich.
\newblock On the complexity of two dimensional commuting local {H}amiltonians.
\newblock In {\em 13th Conference on the Theory of Quantum Computation, Communication and Cryptography}, 2018.
\newblock arXiv:1803.02213.

\bibitem[AL99]{Abrams_1999}
Daniel~S. Abrams and Seth Lloyd.
\newblock Quantum algorithm providing exponential speed increase for finding eigenvalues and eigenvectors.
\newblock {\em Physical Review Letters}, 83(24):5162–5165, December 1999.

\bibitem[AS10]{AltlandSimons2010}
Alexander Altland and Ben~D. Simons.
\newblock {\em Condensed Matter Field Theory}.
\newblock Cambridge University Press, 2 edition, 2010.

\bibitem[ATS03]{aharonov2003adiabaticquantumstategeneration}
Dorit Aharonov and Amnon Ta-Shma.
\newblock Adiabatic quantum state generation and statistical zero knowledge, 2003.

\bibitem[Bar82]{barahona1982computational}
Francisco Barahona.
\newblock On the computational complexity of ising spin glass models.
\newblock {\em Journal of Physics A: Mathematical and General}, 15(10):3241, 1982.

\bibitem[BCL24]{BCL24}
Thiago Bergamaschi, Chi-Fang Chen, and Yunchao Liu.
\newblock Quantum computational advantage with constant-temperature {G}ibbs sampling.
\newblock In {\em 2024 IEEE 65th Annual Symposium on Foundations of Computer Science (FOCS)}, pages 1063--1085. IEEE, 2024.
\newblock arXiv:2404.14639.

\bibitem[BE22]{bartschi2022short}
Andreas B{\"a}rtschi and Stephan Eidenbenz.
\newblock Short-depth circuits for {D}icke state preparation.
\newblock In {\em 2022 IEEE International Conference on Quantum Computing and Engineering (QCE)}, pages 87--96. IEEE, 2022.

\bibitem[BH14]{bravyi2014complexityquantumisingmodel}
Sergey Bravyi and Matthew Hastings.
\newblock On complexity of the quantum ising model, 2014.

\bibitem[BHMT00]{brassard2000quantum}
Gilles Brassard, Peter Hoyer, Michele Mosca, and Alain Tapp.
\newblock Quantum amplitude amplification and estimation.
\newblock {\em arXiv preprint quant-ph/0005055}, 2000.

\bibitem[BMD07]{bombin2007optimal}
H{\'e}ctor Bomb{\'\i}n and Miguel~A. Martin-Delgado.
\newblock Optimal resources for topological two-dimensional stabilizer codes: Comparative study.
\newblock {\em Physical Review A}, 76(1):012305, 2007.

\bibitem[Bom14]{bombin2014structure}
H{\'e}ctor Bomb{\'\i}n.
\newblock Structure of {2D} topological stabilizer codes.
\newblock {\em Communications in Mathematical Physics}, 327:387--432, 2014.

\bibitem[Boo12]{bookatz2012qma}
Adam~D Bookatz.
\newblock Qma-complete problems.
\newblock {\em arXiv preprint arXiv:1212.6312}, 2012.

\bibitem[BV05]{BV05}
Sergey Bravyi and Mikhail Vyalyi.
\newblock Commutative version of the local {H}amiltonian problem and common eigenspace problem.
\newblock {\em Quantum Information and Computation}, 5(3):187--215, 2005.
\newblock arXiv:quant-ph/0308021.

\bibitem[CAB{\etalchar{+}}21]{Cerezo_2021}
M.~Cerezo, Andrew Arrasmith, Ryan Babbush, Simon~C. Benjamin, Suguru Endo, Keisuke Fujii, Jarrod~R. McClean, Kosuke Mitarai, Xiao Yuan, Lukasz Cincio, and Patrick~J. Coles.
\newblock Variational quantum algorithms.
\newblock {\em Nature Reviews Physics}, 3(9):625–644, August 2021.

\bibitem[Cal85]{Callen1985}
Herbert~B. Callen.
\newblock {\em Thermodynamics and an Introduction to Thermostatistics}.
\newblock Wiley, 2 edition, 1985.

\bibitem[Cha87]{Chandler1987}
David Chandler.
\newblock {\em Introduction to Modern Statistical Mechanics}.
\newblock Oxford University Press, 1987.

\bibitem[CKFB23]{GilyenThermal23a}
Chi-Fang Chen, Michael~J. Kastoryano, and Andr{\'a}s~Gily{\'e}n Fernando~Brand{\~a}o.
\newblock Quantum thermal state preparation.
\newblock {\em arXiv preprint}, 2023.

\bibitem[CKG23]{chen2023efficient}
Chi-Fang Chen, Michael~J Kastoryano, and Andr{\'a}s Gily{\'e}n.
\newblock An efficient and exact noncommutative quantum gibbs sampler.
\newblock {\em arXiv preprint arXiv:2311.09207}, 2023.

\bibitem[CM16]{cubitt2016complexityclassificationlocalhamiltonian}
Toby Cubitt and Ashley Montanaro.
\newblock Complexity classification of local hamiltonian problems, 2016.

\bibitem[CPGSV21]{mps2}
J.~Ignacio Cirac, David P\'erez-Garc\'{\i}a, Norbert Schuch, and Frank Verstraete.
\newblock Matrix product states and projected entangled pair states: Concepts, symmetries, theorems.
\newblock {\em Rev. Mod. Phys.}, 93:045003, Dec 2021.

\bibitem[Cub23]{cubitt2023dissipativegroundstatepreparation}
Toby~S. Cubitt.
\newblock Dissipative ground state preparation and the dissipative quantum eigensolver, 2023.

\bibitem[DLL{\etalchar{+}}24a]{DLL24}
Zhiyan Ding, Zeph Landau, Bowen Li, Lin Lin, and Ruizhe Zhang.
\newblock Polynomial-time preparation of low-temperature {G}ibbs states for {2D} toric code.
\newblock {\em arXiv:2410.01206}, 2024.

\bibitem[DLL24b]{LinLinThermal24}
Zhiyan Ding, Bowen Li, and Lin Lin.
\newblock Efficient quantum {G}ibbs samplers with {Kubo–Martin–Schwinger} detailed balance condition.
\newblock {\em arXiv preprint}, 2024.

\bibitem[{\relax DLMF}]{NIST:DLMF}
{\it NIST Digital Library of Mathematical Functions}.
\newblock \url{https://dlmf.nist.gov/}, Release 1.2.4 of 2025-03-15.
\newblock F.~W.~J. Olver, A.~B. {Olde Daalhuis}, D.~W. Lozier, B.~I. Schneider, R.~F. Boisvert, C.~W. Clark, B.~R. Miller, B.~V. Saunders, H.~S. Cohl, and M.~A. McClain, eds.

\bibitem[ER60]{erdds1959random}
Paul Erd{\"o}s and Alfr{\'e}d R{\'e}nyi.
\newblock On the evolution of random graphs.
\newblock {\em Publ. Math. Inst. Hungar. Acad. Sci}, 5:17--61, 1960.

\bibitem[Fey82]{feynman1982simulating}
Richard~P Feynman.
\newblock Simulating physics with computers.
\newblock In {\em Feynman and computation}, pages 133--153. cRc Press, 1982.

\bibitem[FGG{\etalchar{+}}01]{farhi2001adiabatic}
Edward Farhi, Jeffrey Goldstone, Sam Gutmann, et~al.
\newblock A quantum adiabatic evolution algorithm applied to random instances of an {NP}-complete problem.
\newblock {\em Science}, 292(5516):472--475, 2001.

\bibitem[FVDG02]{fuchs2002cryptographic}
Christopher~A Fuchs and Jeroen Van De~Graaf.
\newblock Cryptographic distinguishability measures for quantum-mechanical states.
\newblock {\em IEEE Transactions on Information Theory}, 45(4):1216--1227, 2002.

\bibitem[Gal62]{gallager1962ldpc}
Robert~G. Gallager.
\newblock Low-density parity-check codes.
\newblock {\em IRE Transactions on Information Theory}, 8(1):21--28, 1962.

\bibitem[Gal63]{gallager1963low}
Robert Gallager.
\newblock Low-density parity-check codes.
\newblock {\em IRE Transactions on information theory}, 8(1):21--28, 1963.

\bibitem[GHZ22]{prep_mps_3}
Chenhua Geng, Hong-Ye Hu, and Yijian Zou.
\newblock Differentiable programming of isometric tensor networks.
\newblock {\em Machine Learning: Science and Technology}, 3(1):015020, 2022.

\bibitem[Got16]{gottesman2016surviving}
Daniel Gottesman.
\newblock Surviving as a quantum computer in a classical world.
\newblock {\em Textbook manuscript preprint}, 2016.

\bibitem[GP19]{gharibian2019almost}
Sevag Gharibian and Ojas Parekh.
\newblock Almost optimal classical approximation algorithms for a quantum generalization of max-cut.
\newblock {\em arXiv preprint arXiv:1909.08846}, 2019.

\bibitem[Haa11]{haah2011local}
Jeongwan Haah.
\newblock Local stabilizer codes in three dimensions without string logical operators.
\newblock {\em Physical Review A}, 83(4):042330, 2011.

\bibitem[Haa13a]{Haah2013Modules}
Jeongwan Haah.
\newblock Commuting {P}auli {H}amiltonians as maps between free modules.
\newblock {\em Communications in Mathematical Physics}, 324:351--399, 2013.

\bibitem[Haa13b]{haah2013lattice}
Jeongwan Haah.
\newblock {\em Lattice quantum codes and exotic topological phases of matter}.
\newblock California Institute of Technology, 2013.

\bibitem[HJ24]{HJ24}
Yeongwoo Hwang and Jiaqing Jiang.
\newblock Gibbs state preparation for commuting {H}amiltonian: Mapping to classical {G}ibbs sampling.
\newblock {\em arXiv:2410.04909}, 2024.

\bibitem[HJL{\etalchar{+}}24]{Hanada_2024}
Masanori Hanada, Antal Jevicki, Xianlong Liu, Enrico Rinaldi, and Masaki Tezuka.
\newblock A model of randomly-coupled pauli spins.
\newblock {\em Journal of High Energy Physics}, 2024(5), May 2024.

\bibitem[IJ23]{IJ23}
Sandy Irani and Jiaqing Jiang.
\newblock Commuting local {H}amiltonian problem on {2D} beyond qubits.
\newblock {\em arXiv:2309.04910}, 2023.

\bibitem[JSW{\etalchar{+}}24]{JSW24}
Stephen~P. Jordan, Noah Shutty, Mary Wootters, Adam Zalcman, Alexander Schmidhuber, Robbie King, Sergei~V. Isakov, and Ryan Babbush.
\newblock Optimization by decoded quantum interferometry.
\newblock {\em arXiv:2408.08292}, 2024.

\bibitem[KB16]{KB16}
Michael~J. Kastoryano and Fernando~GSL Brandao.
\newblock Quantum {G}ibbs samplers: the commuting case.
\newblock {\em Communications in Mathematical Physics}, 344(3):915--957, 2016.

\bibitem[Kit95a]{kitaev1995quantum}
A~Yu Kitaev.
\newblock Quantum measurements and the {A}belian stabilizer problem.
\newblock {\em arXiv preprint quant-ph/9511026}, 1995.

\bibitem[Kit95b]{kitaev1995quantummeasurementsabelianstabilizer}
Alexei~Yu Kitaev.
\newblock Quantum measurements and the {Abelian} stabilizer problem.
\newblock {\em Electronic Colloquium on Computational Complexity}, 3(3), 1995.

\bibitem[Kit03]{kitaev2003anyons}
A.~Yu. Kitaev.
\newblock Fault-tolerant quantum computation by anyons.
\newblock {\em Annals of Physics}, 303(1):2--30, 2003.

\bibitem[KKR06]{KKR06}
Julia Kempe, Alexei Kitaev, and Oded Regev.
\newblock The complexity of the local {H}amiltonian problem.
\newblock {\em Siam journal on computing}, 35(5):1070--1097, 2006.
\newblock arXiv:quant-ph/0406180.

\bibitem[KN98]{kadowaki1998quantum}
Tadashi Kadowaki and Hidetoshi Nishimori.
\newblock Quantum annealing in the transverse {I}sing model.
\newblock {\em Physical Review E}, 58(5):5355--5363, 1998.

\bibitem[KSV02]{kitaev2002local}
A.~Yu. Kitaev, A.~Shen, and M.~N. Vyalyi.
\newblock Classical and quantum computation.
\newblock {\em Graduate Studies in Mathematics}, 47, 2002.
\newblock Includes proof of QMA-completeness of Local Hamiltonian.

\bibitem[LB14]{LandauBinder2014}
David~P. Landau and Kurt Binder.
\newblock {\em A Guide to {M}onte {C}arlo Simulations in Statistical Physics}.
\newblock Cambridge University Press, 4 edition, 2014.

\bibitem[LKS24]{low2024trading}
Guang~Hao Low, Vadym Kliuchnikov, and Luke Schaeffer.
\newblock Trading {T} gates for dirty qubits in state preparation and unitary synthesis.
\newblock {\em Quantum}, 8:1375, 2024.

\bibitem[LL80]{LandauLifshitz1980}
L.~D. Landau and E.~M. Lifshitz.
\newblock {\em Statistical Physics, Part 1}.
\newblock Course of Theoretical Physics, Vol. 5. Pergamon Press, 3 edition, 1980.

\bibitem[Llo96]{lloyd1996universal}
Seth Lloyd.
\newblock Universal quantum simulators.
\newblock {\em Science}, 273(5278):1073--1078, 1996.

\bibitem[LMSS02]{luby2002efficient}
Michael~G Luby, Michael Mitzenmacher, Mohammad~Amin Shokrollahi, and Daniel~A Spielman.
\newblock Efficient erasure correcting codes.
\newblock {\em IEEE Transactions on Information Theory}, 47(2):569--584, 2002.

\bibitem[LT20]{Lin_2020}
Lin Lin and Yu~Tong.
\newblock Near-optimal ground state preparation.
\newblock {\em Quantum}, 4:372, December 2020.

\bibitem[Mal03]{maldacena2003eternal}
Juan Maldacena.
\newblock Eternal black holes in anti-de {S}itter.
\newblock {\em Journal of High Energy Physics}, 2003(04):021, 2003.

\bibitem[Mer65]{Mermin1965}
N.~David Mermin.
\newblock Thermal properties of the inhomogeneous electron gas.
\newblock {\em Physical Review}, 137(5A):A1441--A1443, 1965.

\bibitem[Mon17]{montanaro2017learningstabilizerstatesbell}
Ashley Montanaro.
\newblock Learning stabilizer states by {B}ell sampling, 2017.

\bibitem[MTD{\etalchar{+}}23]{prep_mps_2}
Ar~A Melnikov, Alena~A Termanova, Sergey~V Dolgov, Florian Neukart, and MR~Perelshtein.
\newblock Quantum state preparation using tensor networks.
\newblock {\em Quantum Science and Technology}, 8(3):035027, 2023.

\bibitem[NR23]{nepomechie2023qudit}
Rafael~I Nepomechie and David Raveh.
\newblock Qudit dicke state preparation.
\newblock {\em arXiv preprint arXiv:2301.04989}, 2023.

\bibitem[Oru14]{mps1}
Roman Orus.
\newblock A practical introduction to tensor networks: Matrix product states and projected entangled pair states.
\newblock {\em Annals of Physics}, 349:117--158, 2014.

\bibitem[P{\etalchar{+}}14]{peruzzo2014variational}
Alberto Peruzzo et~al.
\newblock A variational eigenvalue solver on a photonic quantum processor.
\newblock {\em Nature Communications}, 5:4213, 2014.

\bibitem[PB11]{PathriaBeale2011}
R.~K. Pathria and Paul~D. Beale.
\newblock {\em Statistical Mechanics}.
\newblock Elsevier, 3 edition, 2011.

\bibitem[Pra62]{prange1962use}
Eugene Prange.
\newblock The use of information sets in decoding cyclic codes.
\newblock {\em IRE Transactions on Information Theory}, 8(5):5--9, 1962.

\bibitem[PVSSC25]{PSS25}
Pablo P{\'a}ez-Velasco, Niclas Schilling, Samuel~O. Scalet, and Frank Verstraete~{\'A}ngela Capel.
\newblock Efficient and simple {G}ibbs state preparation of the {2D} toric code via duality to classical {I}sing chains.
\newblock {\em arXiv:2508.00126}, 2025.

\bibitem[PW09]{Poulin_2009}
David Poulin and Pawel Wocjan.
\newblock Preparing ground states of quantum many-body systems on a quantum computer.
\newblock {\em Physical Review Letters}, 102(13), April 2009.

\bibitem[Reg05]{10.1145/1060590.1060603}
Oded Regev.
\newblock On lattices, learning with errors, random linear codes, and cryptography.
\newblock In {\em Proceedings of the Thirty-Seventh Annual ACM Symposium on Theory of Computing}, STOC '05, page 84–93, New York, NY, USA, 2005. Association for Computing Machinery.

\bibitem[RU08]{richardson2008modern}
Tom Richardson and Ruediger Urbanke.
\newblock {\em Modern coding theory}.
\newblock Cambridge university press, 2008.

\bibitem[RW24]{rajakumar2024gibbs}
Joel Rajakumar and James~D Watson.
\newblock Gibbs sampling gives quantum advantage at constant temperatures with {O(1)-local Hamiltonians}.
\newblock {\em arXiv preprint arXiv:2408.01516}, 2024.

\bibitem[Sac11]{Sachdev2011}
Subir Sachdev.
\newblock {\em Quantum Phase Transitions}.
\newblock Cambridge University Press, 2 edition, 2011.

\bibitem[Sch11]{Sch11}
Norbert Schuch.
\newblock Complexity of commuting {H}amiltonians on a square lattice of qubits.
\newblock {\em Quantum Information \& Computation}, 11(11-12):901--912, 2011.
\newblock arXiv:1105.2843.

\bibitem[Ske46]{skellam}
J.~G. Skellam.
\newblock The frequency distribution of the difference between two poisson variates belonging to different populations.
\newblock {\em Journal of the Royal Statistical Society. Series A (General)}, 109(3):296--296, 1946.

\bibitem[SLS{\etalchar{+}}]{Schmidhuber_HDQI}
Alexander Schmidhuber, Jonathan~Z. Lu, Noah Shutty, Stephen Jordan, Alexander Poremba, and Yihui Quek.
\newblock {HDQI: https://github.com/jz-lu/hdqi}.

\bibitem[Spe20]{speagle2020conceptualintroductionmarkovchain}
Joshua~S. Speagle.
\newblock A conceptual introduction to {M}arkov chain {M}onte {C}arlo methods, 2020.

\bibitem[SSTX09]{cryptoeprint:2009/285}
Damien Stehlé, Ron Steinfeld, Keisuke Tanaka, and Keita Xagawa.
\newblock Efficient public key encryption based on ideal lattices.
\newblock Cryptology {ePrint} Archive, Paper 2009/285, 2009.

\bibitem[SSV{\etalchar{+}}05]{prep_mps_1}
Christian Sch{\"o}n, Enrique Solano, Frank Verstraete, J~Ignacio Cirac, and Michael~M Wolf.
\newblock Sequential generation of entangled multiqubit states.
\newblock {\em Physical review letters}, 95(11):110503, 2005.

\bibitem[TOV{\etalchar{+}}11]{Temme_2011}
K.~Temme, T.~J. Osborne, K.~G. Vollbrecht, D.~Poulin, and F.~Verstraete.
\newblock Quantum {M}etropolis sampling.
\newblock {\em Nature}, 471(7336):87–90, March 2011.

\bibitem[Tuc10]{Tuckerman2010}
Mark~E. Tuckerman.
\newblock {\em Statistical Mechanics: Theory and Molecular Simulation}.
\newblock Oxford University Press, 2010.

\bibitem[UMT82]{umezawa1982thermo}
Hiroomi Umezawa, Hiroshi Matsumoto, and Masashi Tachiki.
\newblock Thermo field dynamics and condensed states.
\newblock {\em online}, 1982.

\bibitem[YB12]{YB12}
Jajiang Yan and Dave Bacon.
\newblock The k-local {P}auli commuting {H}amiltonians problem is in {P}.
\newblock {\em arXiv:1203.3906}, 2012.

\bibitem[YZ24]{10.1145/3658665}
Takashi Yamakawa and Mark Zhandry.
\newblock Verifiable quantum advantage without structure.
\newblock {\em J. ACM}, 71(3), June 2024.

\bibitem[ZJL{\etalchar{+}}20]{zhu2020generation}
Daiwei Zhu, Sonika Johri, Norbert~M Linke, KA~Landsman, C~Huerta~Alderete, Nhunh~H Nguyen, AY~Matsuura, TH~Hsieh, and Christopher Monroe.
\newblock Generation of thermofield double states and critical ground states with a quantum computer.
\newblock {\em Proceedings of the National Academy of Sciences}, 117(41):25402--25406, 2020.

\end{thebibliography}


\appendix      
\addcontentsline{toc}{section}{Appendix}  

\section{Preliminaries}
\label{sec:app_prelim}
We here review some facts from quantum information theory and classical coding theory which are useful in the main text.

\subsection{Bell basis} \label{app:bell_basis}
Let $\ket{\Phi^n} = \frac{1}{2^{n/2}} \sum_{\mathbf{x} \in \F_2^n} \ket{\mathbf{x}} \otimes \ket{\mathbf{x}}$ be the maximally entangled state on $n$ pairs of qubits. Equivalently, $\ket{\Phi^n}$ can be interpreted as $n$ copies of a EPR pair $\frac{1}{\sqrt{2}} (\ket{00} + \ket{11})$. A useful fact about a Hilbert space over $2n$ qubits is that there exists an orthonormal basis of states which are all of the form $(P \otimes I) \ket{\Phi^n}$, where $P$ is a $n$-qubit Pauli operator.
\begin{lemma}[Bell basis is orthonormal] \label{lemma:bell_basis}
    Let $n \geq 1$ be an integer. The set of states given by \begin{align} \label{eq:app-Bell_basis}
        \ket{P} := (P \otimes I) \ket{\Phi^n} ,
    \end{align}
    for every (sign-free) $n$-qubit Pauli operator $P \in \set{I, X, Y, Z}^{\otimes n}$,
    forms an orthonormal basis of $(\C^2)^{\otimes 2n}$.
\end{lemma}
\begin{proof}
    We show that the set of states in Eqn.~(\ref{eq:app-Bell_basis}) are normalized, orthogonal, and span the space.
    Normalization is visible from the definition. We first show that if $P \neq I$, $\braket{I}{P} = 0$. 
    We call this basis the \emph{Bell basis}. \begin{align}
        \braket{I}{P} & = \bra{\Phi^n} (P \otimes I) \ket{\Phi^n} = \frac{1}{2^n} \left( \sum_{\mathbf{y} \in \F_2^n} \bra{\mathbf{y}} \otimes \bra{\mathbf{y}} \right) \left( \sum_{\mathbf{x} \in \F_2^n} P \ket{\mathbf{x}} \otimes \ket{\mathbf{x}} \right) \\
        & = \frac{1}{2^n} \sum_{\mathbf{x}, \mathbf{y} \in \F_2^n} (\bra{\mathbf{y}} P \ket{\mathbf{x}}) (\braket{\mathbf{y}}{\mathbf{x}}) = \frac{1}{2^n} \sum_{\mathbf{x} \in \F_2^n} \bra{\mathbf{x}} P \ket{\mathbf{x}} = \frac{1}{2^n} \Tr[P] = 0 
    \end{align}
    because any non-identity Pauli is traceless. 
    Finally, for any $P \neq P'$, \begin{align}
        \braket{P}{P'} & = \bra{\Phi^n} (P^\dagger \otimes I) (P' \otimes I) \ket{\Phi^n} \\
        & = \bra{\Phi^n}(P P' \otimes I) \ket{\Phi^n} = \braket{I}{PP'} = 0
    \end{align}
    since $PP' \neq I$. 
    Therefore, the set $\set{\ket{P}}$ is orthonormal. 
    Finally, there are $4^n = 2^{2n}$ distinct such states, which matches $\dim (\C^2)^{\otimes 2n}$, so the states span $(\C^2)^{\otimes 2n}$. 
\end{proof}
As a consequence, we observe that a measurement in the Bell basis distinguishes with probability $1$ all Bell basis states.
\begin{theorem}[Coherent Bell measurement] \label{thm:app-Bell_measurement}
    Let $\ket{\psi} = \sum_{P} c_P (P \otimes I) \ket{\Phi^n}$, where $P \in \set{I, X, Y, Z}^{\otimes n}$ is a sign-free $n$-qubit Pauli, $c_P \in \C$, and $\ket{\Phi^n}$ is the maximally entangled state on $n$ pairs of qubits. There exists a Clifford operation $V$ implementable with $\poly(n)$ gates in $\set{H, CX_{i \to j}}$ which maps \begin{align}
    \label{eq:app-Bell_measurement_result}
        \ket{\psi} \longmapsto \sum_{P} c_P \ket{\sym(P)} ,
    \end{align}
    where $\sym(P)$ is the symplectic representation of $P$ from Definition~\ref{def:symplectic_vectors}.
\end{theorem}
\begin{proof}
Apply the inverse of the circuit $U = (\prod_{i=1}^n CX_{i \to i+n}) (\prod_{i=1}^n H_i)$ which maps $\ket{0^{2n}} \to \ket{\Phi^n}$, which uniquely maps every $(P \otimes I) \ket{\Phi^n}$ to a computational basis state $\mathbf{x}(P)$ by the orthonormality property given in Lemma~\ref{lemma:bell_basis}. 
One can show by direct computation that $\mathbf{x}(P)$ differs from $\sym(P)$ by a permutation that we can explicitly write down.
By applying this permutation, we complete the construction of $V$.
\end{proof}

We next show that the above facts have natural generalizations to \emph{qudits}, where $d = p$ is a prime.
Recall in particular that the maximally entangled state on $n$ pairs of qudits is defined by
\begin{align}
    \ket{\Phi^n_p} := \frac{1}{p^{n/2}} \sum_{\mathbf{x} \in \F_2^p} \ket{\mathbf{x}} \otimes \ket{\mathbf{x}} .
\end{align}

\begin{lemma}[Weyl-Bell basis is orthonormal] \label{lemma:weyl_bell_orthonormal}
    Fix a prime $p$.
    Let $W(\mathbf{u})$ be a Weyl operator for $\mathbf{u \in \F_p^{2n}}$, as defined in Eqn.~(\ref{def:Weyl_operator}).
    Then the set of states given by \begin{align}
        \ket{W(\mathbf{u})} := (W(\mathbf{u}) \otimes I) \ket{\Phi_p^n}
    \end{align}
    for each $\mathbf{u} \in \F_p^{2n}$ forms an orthonormal basis of $(\C^p)^{\otimes 2n}$.
\end{lemma}
\begin{proof}
    We proceed by direct calculation.
    For any $\mathbf{u}, \mathbf{v} \in \F_p^{2n}$,
    \begin{align}
        \langle W(\mathbf{u}) | W(\mathbf{v}) \rangle & = \langle \Phi_p^n | (W(\mathbf u)^\dagger W(\mathbf{v}) \otimes I) | \Phi_p^n \rangle \\
        & = \frac{1}{p^n} \left( \sum_{\mathbf{x}\in \F_p^n} \bra{\mathbf{x}} \otimes \bra{\mathbf{x}} \right) \left( \sum_{\mathbf{y}\in \F_p^n} W(\mathbf u)^\dag W(\mathbf v) \ket{\mathbf{y}} \otimes \ket{\mathbf{y}} \right) \\
        & = \frac{1}{p^n} \sum_{\mathbf{x} \in \F_p^n} \langle \mathbf{x} | W(\mathbf u)^\dag W(\mathbf v) | \mathbf{x} \rangle \\
        & = \frac{1}{p^n} \Tr[W(\mathbf u)^\dag W(\mathbf v)] .
    \end{align}
    If $\mathbf{u} = \mathbf{v}$ then the above evaluates to $1$.
    Otherwise, note that $W(\mathbf{u})^\dag = W(u\mathbf{u})$, and up to a phase, $W(\mathbf{u})^\dag W(\mathbf{v}) = W(\mathbf{v} - \mathbf{u})$.
    Let $\mathbf{w} = \mathbf{v} - \mathbf{u} = (\mathbf{a} \,|\, \mathbf{b})$.
    Then \begin{align}
        \Tr[W(\mathbf{w})] = \Tr\left[ \prod_{k=1}^n Z_k^{a_k} X_k^{b_k} \right] = \prod_{k=1}^n Tr[Z_k^{a_k} X_k^{b_k}] = 0 ,
    \end{align}
    since the base Weyl operators $Z$ and $X$ are traceless.
\end{proof}

Using this generalized orthonormal basis, we next show that the $\F_p$ maximally entangled state admits an analog of the coherent Bell measurement, which we denote as the coherent Weyl-Bell measurement.

\begin{theorem}[Coherent Weyl-Bell measurement] \label{thm:coherent_weyl_bell_measurement}
    Let $\ket{\psi} = \sum_{\mathbf{u} \in \F_p^{2n}} c_{\mathbf{u}} (W(\mathbf{u}) \otimes I) \ket{\Phi_p^n}$, where $W(\mathbf{u})$ is a $n$-qudit Weyl operator over $\F_p$.
    There exists an efficiently implementable unitary which maps \begin{align}
        \ket{\psi} \longmapsto  \sum_{\mathbf{u} \in \F_p^{2n}} c_{\mathbf{u}} \ket{\mathbf{u}} .
    \end{align}
\end{theorem}
\begin{proof}
Define the following two-qudit Weyl operators
\begin{equation}
S_Z := Z \otimes Z^{-1}, \quad S_X := X \otimes X ,
\end{equation}
which commute because by Eqn.~(\ref{eq:Weyl_commutation_relations}),
\begin{equation}
S_Z S_X = (Z X) \otimes\left(Z^{-1} X\right)=(\omega X Z) \otimes\left(\omega^{-1} X Z^{-1}\right)=S_X S_Z. 
\end{equation}
The Weyl-Bell basis on two qudits can be represented as
\begin{equation}
\left|\Phi_{p, a, b} \right\rangle: = \left(I \otimes Z^b X^a\right)\left|\Phi^{1}_{p}\right\rangle = \frac{1}{\sqrt{p}} \sum_{k=0}^{p-1} \omega^{b k}|k\rangle \otimes|k+a\rangle, 
\end{equation}
indexed by $(a, b) \in \mathbb{F}_p^2$. These are joint eigenstates of $S_Z$ and $S_X$:
\begin{equation}
S_Z\left|\Phi_{p, a, b} \right\rangle = \omega^{-a}\left|\Phi_{p, a, b} \right\rangle, \quad S_X\left|\Phi_{p, a, b} \right\rangle = \omega^{-b}\left|\Phi_{p, a, b} \right\rangle .
\end{equation}
The ``generalized Bell measurement'' which simultaneously diagonalizes the abelian quotient of the Weyl group is thus given by parallel measurements of $S_Z$ and $S_X$ across all $n$ qudit-pairs. 
The corresponding coherent operation is defined in terms of qudit \texttt{SUM} (generalized controlled-\texttt{NOT}) 
\begin{equation}
\text{\texttt{SUM}}_{1 \rightarrow 2}:|x\rangle|y\rangle \mapsto|x\rangle|x+y\rangle,
\end{equation}
(where, as usual, arithmetic is mod $p$)
and the single-qudit Fourier transform \begin{equation}
F|x\rangle=\frac{1}{\sqrt{p}} \sum_{k=0}^{p-1} \omega^{k x}|k\rangle .
\end{equation}
Both are Clifford operations over $\F_p$ (i.e. they map Weyl operators to Weyl operators by conjugation). 
By construction, the two-qudit Clifford
\begin{equation}
    U_{\mathrm{WBell}}:=\left(F^{\dagger} \otimes I\right) \mathrm{SUM}_{1 \rightarrow 2}^{\dagger} 
\end{equation} satisfies for every $(a, b) \in \mathbb{F}_{p,}^2$,
\begin{equation}
U_{\mathrm{WBell}}\left|\Phi_{a, b}\right\rangle=|b\rangle \otimes|a\rangle .\end{equation}
\end{proof}

\subsection{Tanner graphs and vertex expanders} \label{sec:app-expander_graphs}
Tanner graphs represent a parity check matrix using a bipartite graph. This representation turns out to be very useful to construct code ensembles which have high distance and efficient decodability.
Note that the notation in this subsection follows standard coding theory notation rather than the notation we have established in the main text.
Namely, $H$ is a parity check matrix, not a Hamiltonian.

\begin{definition}[Tanner graph] \label{def:Tanner_graph}
    Let $H \in \mathbb{F}_2^{(1-R)n \times n}$, where $R \in (0, 1)$, be a parity check matrix. The Tanner graph $\calG(H) = (\calL, \calR, E)$ is an undirected bipartite graph with $\calL = [n]$, $|\calR| = [(1-R)n]$ and $E = \{ (i, j) \in \calL \times \calR \,:\, H_{ji} = 1 \}$. 
\end{definition}
We refer to the left nodes as \textit{data nodes} and the right nodes as \textit{check nodes} or \textit{syndrome nodes}. The Tanner graph connects a data node to a check node if the corresponding data bit is supported in the corresponding check.

\begin{definition}[Bipartite regularity]
    A bipartite graph $G = (\calL, \calR, E)$ is (weakly) $(a, b)$-regular if $\deg(i) = a$ ($\deg(i) \leq a$) for all $i \in \calL$ and $\deg(j) = b$ ($\deg(j) \leq b$) for all $j \in \calR$.
\end{definition}
A (weakly) $(a, b)$-regular bipartite graph corresponds to a LDPC code with check sparsity (bounded by) $a$ and bit sparsity (bounded by) $b$. Note that $\frac{a}{b} = 1 - R$ for a $(a, b)$-regular Tanner graph.

Let $N(v)$ be the set of neighbors of a node $v$. 
For a subset $S \subseteq \mathcal{L}$ of a bipartite graph, we denote $\Gamma(S) = \{v \in \mathcal{R} \,:\, \exists u \in S \, (v, u) \in \mathcal{E}\}$ to be the set of nodes in $\mathcal{R}$ adjacent to $S$. 
That is, $\Gamma(S)$ is the set of neighbors of $S$.
Moreover, we denote by $\Gamma_1(S)$ to be the set of neighbors of $S$ which are connected to exactly one element of $S$.
We refer to these as the neighbor set and the unique neighbor set, respectively.

\begin{lemma}[Syndromes in the graph picture] \label{lemma:syndromes_in_graph}
Let $H \in \mathbb{F}_2^{(1-R)n \times n}$ be a parity check matrix with Tanner graph $\mathcal{G} = (\mathcal{L} , \mathcal{R}, \mathcal{E})$. Let $\mathbf{e} \in \mathbb{F}_2^n$ be an error with support $S = \{ j \in [n] \,:\, e_j = 1\}$. 
The syndrome computed is $H \mathbf{e}$. For $j \in \mathcal{R}$, let $w_j = \sum_{v \in N(j)} \mathbf{1}[v \in S] \pmod{2}$. Then $\mathbf{w} = H \mathbf{e}$.
\end{lemma}
\begin{proof}
    The neighbors of a check node $j$ determine which data nodes $j$ sums together as the check. 
    The $j$th syndrome bit is 1 if and only if the parity of the data nodes connected to check node $j$ which are also supported by the error is 1.
\end{proof}

\begin{corollary} \label{corollary:unique_neighbor_syndrome}
    In the setup of Lemma~\ref{lemma:syndromes_in_graph}, if for an error $\mathbf{e}$ with support $S$ we have $|\Gamma_1(S)| > 0$, then $H \mathbf{e} \neq 0$.
\end{corollary}
\begin{proof}
    Since $|\Gamma_1(S)| > 0$, there exists some $j \in \Gamma_1(S)$. Then $w_j = 1$, since $N(j) \cap S = 1$. 
    By Lemma~\ref{lemma:syndromes_in_graph}, $\mathbf{w} = H \mathbf{e}$, so $H \mathbf{e} \neq 0$.
\end{proof}

We now state the key structural property that we need, which is related to the neighbor sets.

\begin{definition}[Vertex expansion]
    Let $\mathcal{G} = (\mathcal{L} , \mathcal{R}, \mathcal{E})$ be a weakly $(a, b)$-regular bipartite graph with $|\mathcal{L}| = n$ and $|\mathcal{R}| = (1-R)n$. 
    We say that $\mathcal{G}$ is a $(\delta, \gamma)$ (unique) expander if for all $S \subseteq \mathcal{L}$ such that $|S| \leq \delta n$, $|\Gamma(S)| \geq \gamma |S|$ ($|\Gamma_1(S)| \geq \gamma |S|$).
\end{definition}
Here $a, b, \delta, \gamma, R$ are all understood to be constants. We now show that high expansion implies unique expansion.

\begin{lemma}[High expansion implies unique expansion] \label{lemma:expansion_implies_unique_expansion}
    Let $\mathcal{G} = (\mathcal{L} , \mathcal{R}, \mathcal{E})$ be a weakly $(a, b)$-regular bipartite graph with $|\mathcal{L}| = n$ and $|\mathcal{R}| = (1-R)n$. 
    If $\mathcal{G}$ is a $(\delta, \gamma)$ expander, then $\mathcal{G}$ is a $(\delta, 2 \gamma - a)$ unique expander.
\end{lemma}
\begin{proof}
    Fix a subset of nodes $S$ such that $|S| \leq \delta n$. 
    We partition $\Gamma(S)$ into two pieces, $\Gamma_1(S)$ and $\Gamma(S) - \Gamma_1(S)$. Denote $u = |\Gamma_1(S)|$ and $r = |\Gamma(S) - \Gamma_1(S)|$ the number of unique neighbors and repeated neighbors, respectively. 
    Then $u + r = |\Gamma(S)| \geq \gamma |S|$ by assumption. 
    Next, there are at most $c |S|$ edges connected to $S$. 
    Each neighbor in $\Gamma_1(S)$ has exactly one such edge, and each neighbor in $\Gamma(S) - \Gamma_1(S)$ has at least 2. 
    Hence, $u + 2r \leq a |S|$. Equivalently, $r \leq \frac{a}{2} |S| - \frac{1}{2} u$. 
    Solving, \begin{align}
        u \geq \gamma |S| - r \geq \gamma |S| - \frac{a}{2} |S| - \frac{1}{2} u .
    \end{align}
    Rearranging proves the claim.
\end{proof}

Expansion is useful because it implies unique expansion, and unique expansion is useful because it implies linear distance codes. 
We now show the latter.
\begin{lemma}[Unique expansion implies linear distance] \label{lemma:unique_expansion_implies_distance}
    Let $H \in \mathbb{F}_2^{(1-R)n \times n}$ be a parity check matrix whose Tanner graph $\mathcal{G} = (\mathcal{L} , \mathcal{R}, \mathcal{E})$ is a weakly $(a, b)$-regular bipartite graph. 
    If, for any constants $\gamma' > 0$ and $\delta > 0$, $\mathcal{G}$ is a $(\delta, \gamma')$ unique neighbor expander, then $\mathcal{G}$ corresponds to a linear distance code with relative distance $\delta$.
\end{lemma}
\begin{proof}
    Let $\mathbf{e} \in \mathbb{F}_2^n$ be any error such that $1 \leq |\mathbf{e}| \leq \delta n$. 
    Let $S := \{j \in [n] \,:\, e_j = 1 \}$ be the support of $\mathbf{e}$, and thus $|S| \leq \delta n$. 
    Then $|\Gamma_1(S)| \geq \gamma' |S| > 0$. 
    By Corollary~\ref{corollary:unique_neighbor_syndrome}, $H \mathbf{e} \neq 0$, so $\mathbf{e}$ is not a codeword.
\end{proof}

\begin{corollary}[High expansion implies linear distance] \label{corollary:expansion_implies_linear_distance}
    Let $H \in \mathbb{F}_2^{(1-R)n \times n}$ be a parity check matrix whose Tanner graph $\mathcal{G} = (\mathcal{L} , \mathcal{R}, \mathcal{E})$ is a weakly $(a, b)$-regular bipartite graph. 
    If there exists $\delta > 0$ and $\gamma > c/2$ for which $\mathcal{G}$ is a $(\delta, \gamma)$ expander, then $\operatorname{dist}(H) \geq \delta n$.
\end{corollary}
\begin{proof}
    This is an immediate consequence of Lemma~\ref{lemma:expansion_implies_unique_expansion} and Lemma~\ref{lemma:unique_expansion_implies_distance}.
\end{proof}

Not only does high expansion imply a linear distance, but it also implies that all information-theoretically correctable errors can be efficiently corrected, using the \texttt{Flip} decoding algorithm, provided that $\gamma > \frac{3}{4} a$~\cite{richardson2008modern}.
Fortunately, with high probability this property holds for a random $(a, b)$-regular Tanner graph.
More precisely, a classical result of coding theory is that random graphs have high vertex expansion with high probability. 
In what follows we define for $x \in [0, 1]$ the binary entropy function $H_2(x) = -x \log_2 x - (1-x) \log_2(1-x)$.
\begin{theorem}[Random regular graphs are expanders] \label{thm:random_expanders}
    Let $\mathcal{G} = (\calL, \calR, \calE)$ be a uniformly random $(a-1, b-1)$-regular bipartite graph such that $|\calL| = n$ and $|\calR| = nR = \frac{a-1}{b-1} n$, where $R = \frac{a-1}{b-1}$. For any $\widetilde{\gamma}\in (0, 1 - \frac{1}{a})$ and any $\delta \in (0, \delta^*)$, where $\delta^*$ is the positive solution to the equation \begin{align}
        \frac{a-1}{a} H_2(\delta) - \frac{a}{b} H_2(\delta \widetilde{\gamma} b) - \delta \widetilde{\gamma} b H_2\left( \frac{1}{\widetilde{\gamma} b} \right) = 0 ,
    \end{align}
    we have \begin{align}
        \Pr[\calG \text{ is a } (\delta, \gamma) \text{ expander}] \geq 1 - O(n^{-\beta}) ,
    \end{align}
    where $\gamma = (a-1)\widetilde{\gamma}$ and $\beta = a (1 - \widetilde{\gamma}) - 1 \in (0, a-1)$.
\end{theorem}
\begin{proof}
    See Chapter 8 of~\cite{richardson2008modern}.
\end{proof}

As a corollary, if we generate a $(a, b)$-regular Tanner graph with $n$ data nodes and $nR$ check nodes, with probability $1 - 1/\poly(n)$ the code will be efficiently decodable up to a linear number of errors $\delta^* n$, for some constant $\delta^*$ which depends on $a, b, R$.
Such a Tanner graph has a parity check matrix which is distributed as a uniformly random matrix whose rows all have weight $b$ and columns all have weight $a$.

\section{Relation between Hamiltonian DQI and Decoded Quantum Interferometry}
\label{sec:app-cDQI_simulation}
In this appendix we formalize the claim that when
$H=\sum_{i=1}^m v_i P_i$ is a \emph{commuting} Pauli Hamiltonian, the commuting-case
reduction of Section~\ref{sec:commuting} 
can be simulated by calls to an oracle implementing Decoded Quantum Interferometry (DQI) together with efficient classical computation and Clifford operations.
Intuitively, commuting Pauli terms can be simultaneously diagonalized by a Clifford;
after this change of basis the task reduces to the purely classical/diagonal setting
addressed by DQI, and we can map the outcome back to the original basis by undoing
the Clifford rotation.%

Let $H=\sum_{i=1}^m v_i P_i$ be a commuting Pauli Hamiltonian on $n$ qubits,
with $v_i\in\{\pm 1\}$ and $P_i\in\{I,X,Y,Z\}^{\otimes n}$.
Let $B^\intercal\in\mathbb{F}_2^{2n\times m}$ denote the parity-check matrix of the 
symplectic code $\mathrm{Symp}(H)$.
Fix a univariate polynomial $\mathcal{P}$ of degree $\ell$.
Measuring the HDQI state $\rho_{\mathcal{P}}(H)=\mathcal{P}(H)^2/\mathrm{Tr}[\mathcal{P}(H)^2]$
in the eigenbasis of $H$ yields an eigenpair $(\lambda,\ket{\lambda})$ with probability
proportional to $\mathcal{P}(\lambda)^2$ (see Section~\ref{sec:commuting}).

\begin{lemma}[Simultaneous Clifford diagonalization~\cite{gottesman2016surviving}]
\label{lem:simul-diag}
If $[P_i,P_j]=0$ for all $i,j$, there exists a Clifford unitary $U$ and binary
vectors $\mathbf a_1 , \dots , \mathbf a_m \in \mathbb{F}_2^n$ such that
\begin{align}
    U P_i U^\dagger \;=\; Z^{\mathbf a_i} \;:=\; \bigotimes_{q=1}^n Z_q^{(a_i)_q}
\qquad\text{and hence}\qquad
U H U^\dagger \;=\; \sum_{i=1}^m v_i\, Z^{\mathbf a_i}.
\end{align}
Moreover, $U$ and the list $(\mathbf a_i)_{i=1}^m$ can be computed efficiently by standard encoding circuit synthesis algorithms in stabilizer coding theory.
\end{lemma}

\noindent
Define $R\in\mathbb{F}_2^{n\times m}$ by taking the $i$-th column to be $\mathbf a_i$.
For any support $\mathbf y\in\mathbb{F}_2^m$, write $P_{\mathbf y} := \prod_{i \in \supp(\mathbf y)} P_i$ and
$P_{\mathbf y}':=U P_{\mathbf y} U^\dagger=Z^{R \mathbf y}$.
In particular, the \emph{classical} (diagonal) Hamiltonian $H':=U H U^\dagger$ is specified by
the parity-check $R$ in exactly the way used by DQI in the diagonal setting.

\begin{theorem}[DQI simulates commuting HDQI at the sampling level]
\label{thm:dqi-simulates-commuting}
Let $H=\sum_{i=1}^m v_i P_i$ be a commuting Pauli Hamiltonian with parity-check matrix
$B^\intercal$ defined by $\mathrm{Symp}(H)$, and let $\mathcal{P}$ be a degree-$\ell$ polynomial.
There is a polynomial-time classical procedure that, given $(B^\intercal,\mathcal{P})$ and
a bounded-distance decoder for $\mathrm{Symp}(H)$ that corrects up to $\ell$ errors, constructs:
\begin{enumerate}
\item a Clifford $U$ and a classical parity-check $R\in\mathbb{F}_2^{n\times m}$ for the
diagonalized instance $H'=U H U^\dagger=\sum_i v_i Z^{\mathbf a_i}$ with columns $\mathbf a_i$,
\item a bounded-distance decoder for the classical code with parity check $R$ that corrects up to $\ell$ errors,
\end{enumerate}
such that a single call to a DQI oracle on input $(R,\mathcal{P})$ produces a computational-basis
string $\mathbf z\in\{0,1\}^n$ distributed as
$\Pr[\mathbf z]\propto \mathcal{P}(\lambda(\mathbf z))^2$, where $\lambda(\mathbf z)$ is the eigenvalue of $H'$ on $\ket{\mathbf z}$.
Equivalently, defining $\ket{\lambda}:=U^\dagger \ket{z}$ and noting that $U^\dagger$ is Clifford, the algorithm can output stabilizer descriptions of $\ket{\lambda}$ distributed as
$\Pr[\ket{\lambda}]\propto \mathcal{P}(\lambda)^2$.
Hence DQI, together with efficient classical pre/post-processing and Cliffords, reproduces the
eigenbasis sampling behavior of commuting HDQI.
\end{theorem}

\begin{proof}
By Lemma~\ref{lem:simul-diag}, compute a Clifford $U$ and columns $\mathbf a_i\in\mathbb{F}_2^n$
so that $H'=\sum_i v_i Z^{\mathbf a_i}$ with $\mathbf a_i$ as the $i$-th column of
$R\in\mathbb{F}_2^{n\times m}$. 
For any support $y$, we have
$P_{\mathbf y}' = U P_{\mathbf{y}} U^\dagger = Z^{R \mathbf y}$.
Thus $H'$ is diagonal in the computational basis with the standard DQI representation:
the classical (diagonal) parity structure is exactly $R$.

\emph{Decoder for $R$ via the decoder for $\mathrm{Symp}(H)$.}
A decoder for $R$ receives a syndrome $\mathbf s = R \mathbf y$ and must recover an error pattern $y$
(of weight $\le \ell$ in the regime of interest). From $\mathbf s$ form the diagonal Pauli
$P'_{\mathbf y} = Z^{\mathbf s}$. Undo the diagonalization: $P_{\mathbf y} = U^\dagger Z^{\mathbf s} U$.
Let $\mathbf b := \mathrm{Symp}(P_{\mathbf y}) \in\mathbb{F}_2^{2n}$. 
Recovering $\mathbf y$ is the standard bounded-distance decoding task for the code with parity check $B^\intercal$: solve $B^\intercal \mathbf y = \mathbf b$ using the assumed decoder for $\mathrm{Symp}(H)$. 
The output $\mathbf y$ is then returned as the decoded error for the code with check $R$.
(Equivalently, since $U$ is Clifford there exists a symplectic matrix
$S\in\mathrm{Sp}(2n,\mathbb{F}_2)$ with $\mathrm{Symp}(U P U^\dagger)=S \cdot \mathrm{Symp}(P)$;
hence $\mathrm{Symp}(U^\dagger Z^{\mathbf s} U)=S^{-1}(\mathbf s, \mathbf 0)$, and one decodes $B^\intercal \mathbf y = S^{-1}(\mathbf s, \mathbf 0)$.)

\emph{Reduction to DQI.}
DQI for the diagonal instance $H'=\sum_i v_i Z^{\mathbf a_i}$ with parity-check $R$ (and decoder as above)
prepares the usual ``$q$-sample'' whose computational-basis measurement yields $\mathbf z$ with probability
proportional to $\mathcal{P}(\lambda(\mathbf z))^2$ (cf.\ the classical/diagonal discussion in the main text).
There is an efficient classical algorithm that returns a stabilizer description of $\ket{\lambda}:=U^\dagger\ket{\mathbf z}$ given $\mathbf z$, and after applying this algorithm, we obtain stabilizer descriptions of $\ket{\lambda}$ distributed as $\Pr[\ket{\lambda}]\propto \mathcal{P}(\lambda)^2$ as claimed.
All steps are efficient: computing $U$ and $(\mathbf a_i)$ by Gaussian elimination on $B^\intercal$,
running the bounded-distance decoder for $\mathrm{Symp}(H)$, and  classically tracking
Cliffords.
\end{proof}

\begin{remark}
\label{rem:scope}
The equivalence established above is an \emph{eigenbasis-sampling} equivalence:
DQI (after diagonalization by $U$) reproduces the $\propto\mathcal{P}(\lambda)^2$ sampling
distribution over eigenstates of $H$.
Certain HDQI tasks that rely on access to the mixed state $\rho_{\mathcal{P}}(H)$ beyond such
sampling (e.g., estimation of general non-Clifford observables on Gibbs states or other
polynomial Hamiltonian states) are not directly captured by DQI in this black-box simulation.
\end{remark}

\section{No-go result on structured classical codes} \label{app:symplectic_extension}

A key observation from \cite{JSW24} was that an algebraically structured classical code maps by the DQI reduction to a particular optimization problem which could be a candidate for quantum advantage.
Since DQI is a special case of HDQI, such a strategy holds here as well, but all such Hamiltonians would be completely classical (i.e. a $Z$-type Hamiltonian).
A natural question is whether any version of this technique could be applied for HDQI.
We show in this Appendix that, at least for the commuting Hamiltonian case, this is not possible in a certain sense.
More precisely, there is no way to start with a structured classical code which is provably decodable to high distance, transform it in a way that the columns remain symplectically orthogonal (so that the corresponding Hamiltonian is commuting), and obtain a Hamiltonian which is non-trivially quantum.
Roughly speaking, this is because any valid transformation is equivalent to applying an explicitly computable Clifford operation to the classical Hamiltonian.
In general, the construction of the HDQI decoder in this setting would be to undo this Clifford, obtaining the good classical code, run the classical code's decoder, and then re-do the Clifford.
In this setting, HDQI would not do anything that DQI could not have already done.
Therefore, if we wish to search for quantum advantage by using structured classical codes, we must be able to either find a Clifford which maps it to an inherently interesting commuting Hamiltonian, or be able to treat significantly non-commuting models.

To formally support the above argument, we prove that any linear transformation of a code which preserves the symplectic inner product can be equivalently carried out by a Clifford operation which we explicitly compute.
In what follows, $\odot$ refers to the symplectic inner product.
A fact about symplectic vector spaces is that, analogous to Gram-Schmidt, there is an orthogonalization process which produces a basis of $2n$ binary vectors paired up into $n$ pairs, so that every basis vector is symplectically orthogonal exactly to all other basis vectors except for its paired partner.

\begin{theorem}[Symplectic extension] \label{thm:symplectic_extension}
    Let $B^\intercal \in \F_2^{2n \times m}$ be a binary matrix. Suppose there is a matrix $T \in \F_2^{2n \times 2n}$ which preserves that symplectic inner product over columns of $B^\intercal$. That is, for all $i, j \in [m]$, \begin{align}
        \mathbf{b}_i \odot \mathbf{b}_j = T \mathbf{b}_i \odot T \mathbf{b}_j .
    \end{align}
    Then there exists a symplectic matrix $T'$, efficiently computable given $T$, such that $T' B^\intercal = T B^\intercal$.
\end{theorem}

\begin{proof}
    Let $S := \operatorname{span}(\set{\mathbf{b}_i}_{i=1}^m)$ which has some dimension $k \leq 2n$. 
    Using symplectic Gram-Schmidt, we may construct a symplectic basis of the subspace $S \subseteq \F_2^n$. 
    This basis consists of some $r$ pairs of vectors $(\mathbf{u}_1, \mathbf{u'}_1), \dots, (\mathbf{u}_r, \mathbf{u'}_r)$ which satisfy \begin{align}
        \mathbf{u}_i \odot \mathbf{u'}_j = \delta_{ij} ,\; \mathbf{u}_i \odot \mathbf{u}_j = 0 ,\; \mathbf{u'}_i \odot \mathbf{u'}_j = 0 .
    \end{align}
    In general, $2r \neq k$, and there will be some $k - 2r$ unpaired vectors $\mathbf{w}_1, \dots, \mathbf{w}_{k-2r}$ which satisfy \begin{align}
        \mathbf{w}_i \odot \mathbf{w}_j = 0 ,\; \mathbf{w}_i \odot \mathbf{u}_j = 0 ,\; \mathbf{w}_i \odot \mathbf{u'}_j = 0 .
    \end{align}
    Note that $T$ preserves the symplectic inner product for any linear combination of $\mathbf{u}_i$, $\mathbf{u'}_j$ where $i, j \in [r]$, and $\mathbf{w}_i$ where $i \in [k-2r]$.
    Next, find the dual space $S^\perp$ write down a basis for it. 
    From that, we know the partners of each $\mathbf{w}_i$ live in $S^\perp$. 
    Applying more Gram-Schmidt, we can pair up the $\mathbf{w}_i$ as well as any remaining vectors, into a complete basis. 
    We will call this basis $\set{(\mathbf{u}_i, \mathbf{u'}_i)}_{i=1}^n$, where $\mathbf{u}_i$ for $i \in [r]$ are as before, $\mathbf{u}_i$ for $i = r+1, \dots, k-r$ are the $\mathbf{w}_1, \dots, \mathbf{w}_{k-2r}$, $\mathbf{u}'_i$ for $i \in [r]$ are as before, and $\mathbf{u}'_i$ for $i > r$ are from $S^\perp$. 
    Next, let \begin{align}
        \mathbf{v}_i = T(\mathbf{u}_i) ,\, \mathbf{v'}_i = T(\mathbf{v}_i) \text{  for  } i \in [r] .
    \end{align}
    Since $T$ preserves the symplectic inner product for any linear combination of $\mathbf{u}_i$ and $\mathbf{u'}_j$ where $i, j \in [r]$, we are in an isomorphic setting and can apply the same procedure as before to build an extended basis $\set{(\mathbf{v}_i, \mathbf{v'}_i)}_{i=1}^n$. 
    To complete the proof, we define $T'$ by taking every $\mathbf{u}_i \longrightarrow \mathbf{v}_i$ and $\mathbf{u'}_i \longrightarrow \mathbf{v'}_i$ and extending linearly. 
    By construction, the symplectic inner product is always preserved, and $T'$ agrees with $T$ on $S$ as desired.
\end{proof}

Note that Clifford operators are isomorphic to symplectic matrices, so from a symplectic matrix $T'$ we can compute efficiently the corresponding Clifford circuit.
Now, let $T$ be any transformation on $B^\intercal$ which preserves the symplectic inner product on all columns of $B^\intercal$.
If the Hamiltonian corresponding to $B^\intercal$ is related to a classical Hamiltonian by a Clifford, it remains related to a classical Hamiltonian by a Clifford after the transformation by $T$. 
As a consequence, there is no hope of constructing a HDQI instance by (a) combining one or more arbitrary good classical codes and then transforming them in such a way that the symplectic inner product is preserved, and (b) using a decoder for HDQI in which we undo the linear transformation, decode, and re-do the transformation.

\section{Polynomial approximations to Gibbs states}
\label{app:gibbs_theorem_proof}
In this section, we formally prove Theorem~\ref{thm:gibbs_explicit}.
For a given Hamiltonian $H$, let $\rho_{\beta}(H) = \frac{e^{-\beta H}}{\Tr(e^{-\beta H})}$ be the Gibbs state at temperature $\beta$ associated with $H$. 
We argue that HDQI can approximately prepare purified Gibbs states. 

\Gibbsthm*
Note that we have solved for $m$ with all constants included, as that affects the (constant) temperature at which we can prepare the Gibbs state. 
Choose $K$ such that $\norm{H} \leq K$. 
We will wish to approximate with a polynomial $\calP$, the function 
\begin{equation}
    f(u) =e^{-a u} ,
\end{equation} 
where
    \begin{align}
        u&:= \frac{x}{K} , \quad a&:= \frac{\beta K}{2}.
    \end{align}
and the variable $x \in [-K,K]$ is supported on the range of eigenvalues of $H$, and therefore $u \in [-1,1]$. The division by $2$ in the definition of $a$ comes about because the HDQI state is proportional to the {\em square} of the polynomial we apply. 

To approximate $f(u)$, we will use the Chebyshev expansion of $e^{-au}$, but truncate it at degree $m$ (to be determined), obtaining the truncated Chebyshev polynomial $p_m$:
\begin{equation}\label{eq:truncated_Chebyshev}
\calP(u) = p_m(u) := I_0(-a)+2 \sum_{k=1}^{m} I_k(-a) T_k(u) 
= I_0(a) + 2 \sum_{k=1}^{m}(-1)^k I_k(a) T_k(u),
\end{equation}
where $T_k$ are Chebyshev polynomials of the first kind and $I_k$ are the modified Bessel functions of the first kind. 
The second equality follows from the Jacobi–Anger expansion of $e^{a \cos (\theta)}$ (see for instance \cite{NIST:DLMF}, 10.35); in the third equality we have used the identity $I_k(-t) = (-1)^k I_k(t)$. 
The Chebyshev polynomials of the first kind are defined by $T_k(u):=\cos (k \arccos (u))$ and so $\forall u\in [-1,1],\, T_k(u) \in [-1,1].$ 
Theorem~\ref{thm:polyexp} shows that choosing degree $m$ to be $\propto \beta K$ suffices, and computes the exact constant needed. 

Before going into details of the polynomial approximation, we first provide a bridging lemma that lifts the error of $p_m$ as an approximation to $f(u)$, to the trace distance error between the target Gibbs state, and the state that results from applying a factor proportional $p_m(\lambda_i)$ to the corresponding energy eigenstates $\{\ket{\lambda_i}\}$ of $H$. Since $e^{-au} = p_{\infty}(u)$, the absolute error made by $p_m$ as an approximation to $e^{-au}$ is
\begin{equation}
    \sup_{u\in [-1,1]} \left\lvert 2 \sum_{k=m+1}^{\infty} (-1)^k I_k(a) T_k(u)\right\rvert 
    \leq \left\lvert 2 \sum_{k=m+1}^{\infty} I_k(a) \right\rvert,
\end{equation}
and the relative error (which we will later bound) is
\begin{equation}
    \sup_{u\in [-1,1]} \Bigg|\frac{p_m(u)-e^{-au}}{e^{-au}}\Bigg| =: \delta.
\end{equation}
Then noting that the Gibbs weights are actually the {\em square} of the target function $f(u)=e^{-au}$, we compute that
\begin{align}
    |f(u)^2 - p_m(u)^2| &\leq |f(u)-p_m(u)|\, \Big(2\min \{f(u),p_m(u)\} + |f(u)-p_m(u)|\Big) \\
    &\leq \delta f(u)(2f(u)+\delta f(u))\\
    &\leq 3\delta f(u)^2.\label{def:eta}
\end{align}

To connect this back to the problem of approximating Gibbs states, let us evaluate $f^2$ and $p_m^2$ at the actual eigenvalues of $H$, which we call $\{\lambda_i\}$. Recall that, in keeping with the definition of $u$, we will need to normalize $\lambda_i$ by dividing by $K$. In the following Lemma, in view of Definition~\ref{def:eta}, we should think of $\eta = 3\delta$ where $\delta$ is the error of the polynomial approximation.
\begin{lemma}[From error of polynomial approximation to trace distance error of final state]\label{lem:polyerrortotrdisterror}
Let $(w')_i=p_m(\lambda_i/K)^2$; similarly let $(w)_i=e^{-\beta\,\lambda_i/K} = f(\lambda_i/K)^2$, and suppose these weights satisfy 
\begin{equation}
\lvert w_i' - w_i\rvert \le \eta\, w_i
\end{equation}
for all $i$ with some $\eta<1$. Let $p_i = w_i/\sum_j w_j$ and $q_i = w_i'/\sum_j w_j'$. Then
\begin{align}
        \frac{1}{2}\left\lVert \rho_{\calP}(H) - \rho_{\beta}(H)\right\rVert_1
        &= \frac{1}{2}\left\lVert \frac{w}{\sum_j w_j} - \frac{w'}{\sum_j w_j'} \right\rVert_1 = \frac{1}{2}\sum_i \lvert p_i - q_i\rvert \\
        &\le\; \frac{\eta}{1-\eta}.
\end{align}
\end{lemma}

\begin{proof} We will prove the last inequality in the statement of the Lemma, as everything else follows by definition. 
Write $r_i := w_i'/w_i \in [1-\eta,\,1+\eta]$. Then
\begin{align}
    \sum_i w_i'=\sum_i r_i w_i \in \big[(1-\eta)\sum_i w_i,\,(1+\eta)\sum_i w_i\big].
\end{align}
Using
\begin{align}
    \lvert q_i - p_i\rvert
= \left\lvert \frac{w_i'}{\sum_j w_j'} - \frac{w_i}{\sum_j w_j}\right\rvert
\le \frac{\lvert w_i' - w_i\rvert}{\sum_j w_j'} \;+\;
\frac{w_i\,\lvert \sum_j w_j' - \sum_j w_j\rvert}{(\sum_j w_j')(\sum_j w_j)}\,,
\end{align}
summing over $i$ gives
\begin{align}
    \sum_i \lvert q_i - p_i\rvert
\;\le\; \frac{\sum_i \lvert w_i' - w_i\rvert}{\sum_j w_j'}
\;+\; \frac{\lvert \sum_j w_j' - \sum_j w_j\rvert}{\sum_j w_j'}
\;\le\; \frac{\eta \sum_i w_i}{\sum_j w_j'} \;+\; \frac{\eta \sum_i w_i}{\sum_j w_j'}
\;\le\; \frac{2\eta}{1-\eta}\,,
\end{align}
where we used $\sum_i \lvert w_i' - w_i\rvert \le \eta \sum_i w_i$ and
$\sum_j w_j' \ge (1-\eta)\sum_j w_j$. This proves the claim.
\end{proof}

From Lemma~\ref{lem:polyerrortotrdisterror}, if we target a trace distance error of $\frac{1}{2}\left\lVert \rho_{\calP}(H) - \rho_{\beta}(H)\right\rVert_1 \leq \eps<1$, choosing $\delta =\eps/6$ suffices. With this scaffolding in place, we are ready to bound the error of the degree-$m$ truncated Chebyshev expansion as a polynomial approximation of $e^{-au}$.

\begin{theorem}[Polynomial approximations of the exponential function on a bounded interval] \label{thm:polyexp}
Let $K, \beta \in \mathbb{R}_{+}$. There exists an efficiently constructable polynomial $p_m \in \mathbb{R}[x]$ such that
\begin{equation}
\sup_{u\in[-1,1]} \frac{|p_m(u)-e^{-au}|}{e^{-au}} \leq \frac{\eps}{6},
\end{equation}
as long as
\begin{equation}
    m\geq 1.12 \beta K + 0.648\ln \left(\frac{12}{\eps}\right).
\end{equation}
\end{theorem}

\begin{proof}
An expression for the relative error is
\begin{align}
    \delta_m(a) := \sup_{u\in[-1,1]} \frac{|p_m(u)-e^{-au}|}{e^{-au}} \leq \frac{|p_m(u)-e^{-au}|}{\inf_{u\in[-1,1]} e^{-au}} = e^a \lVert p_m(u) -e^{-au}\rVert_{\infty} \leq 2e^a \sum_{k> m} I_k(a).
\end{align}
We now relate this to the tails of a Skellam distribution \cite{skellam}: namely, the Skellam random variable $D\sim \textsf{Skellam}(\mu,\mu)$ is the difference of two independent $\textsf{Poisson}(\mu)$ variables. Its probability mass function is
\begin{align}
    \Pr\{D=k\}=e^{-2 \mu} I_{|k|}(2 \mu),
\end{align}
and therefore to bound $\delta_m(a)$ (recalling that $a=\beta K/2$) it suffices to bound $2e^{\beta K}\Pr(|D|>m)$ where $D\sim \textsf{Skellam}(a/2,a/2)$. In the following, set $\mu=a/2$. To bound the tails of the Skellam distribution, we will use the standard method of showing that the moment-generating function concentrates. The MGF of $D$ is:
\begin{equation}
    M_D(\theta):=\mathbf{E} [e^{\theta D}]=\exp \left(\mu\left(e^\theta-1\right)\right) \exp \left(\mu\left(e^{-\theta}-1\right)\right)=\exp (2 \mu(\cosh \theta-1)).
\end{equation}
By Markov's inequality applied to the moment-generating function,
\begin{equation}
\Pr(D >m)=\Pr\left(e^{\theta D} \geq e^{\theta m}\right) \leq e^{-\theta m} \mathbf{E} [e^{\theta D}] = \exp \Big[2 \mu(\cosh \theta-1) - \theta m\Big].
\end{equation}
The above holds for any $\theta$, so to get the tightest bound we optimize $\theta$ by setting the derivative of the RHS to 0:
\begin{equation}
    \frac{d}{d\theta} \Bigg[ \exp \Big[2 \mu(\cosh \theta-1) - \theta m\Big]\Bigg] = \exp \Big[2 \mu(\cosh \theta-1) - \theta m\Big] (-m+2\mu \sinh(\theta))
\end{equation}
This is $0$ when $\theta = \sinh^{-1}\left(\frac{m}{2\mu}\right).$ Plugging this back in and using the identity $\cosh(\theta)=\sqrt{1+\sinh^2(\theta)}$, we obtain 
\begin{equation}
   \Pr(D >m)\leq \exp \Big[2 \mu\left(\sqrt{1+\frac{m^2}{a^2}}-1\right) - m \sinh^{-1}\left(\frac{m}{a}\right)\Big]. 
\end{equation}
We actually need to bound both tails, so 
\begin{equation}
    \Pr(|D| >m) = 2 \Pr(D >m) =2e^{-\Lambda(m,a)}
\end{equation}
where 
\begin{equation}
    \Lambda(m,a) = m \sinh^{-1}\left(\frac{m}{a}\right) - a \left(\sqrt{1+\frac{m^2}{a^2}}-1\right).
\end{equation}
and to hit the target error of $\delta_m(a)=\eps/6$, we need to have $\Lambda(m,a)-2a \geq \log\left(\frac{12}{\eps}\right).$ Solving this equation numerically obtains 
\begin{equation}
    m\geq 1.12 \beta K + 0.64792\ln \left(\frac{12}{\eps}\right).
\end{equation}
\end{proof}
This completes the proof of Theorem~\ref{thm:gibbs_explicit}.

\section{Efficient expansion under a constant number of product relations}\label{app:block_expansion}
In this section, we show that if $m$ variables $z_1,\dots,z_m$ (with $z_i^2=1$) satisfy $k=\mathrm{const}$ independent multiplicative relations, then after eliminating $k$ variables the remaining $n=m-k$ variables can be partitioned into $r\le 2^k$ blocks so that, for any univariate polynomial $\calP$, the element $\calP(S)$ with $S=\sum_{i=1}^m z_i$ admits a unique expansion as a linear combination of products of elementary symmetric polynomials within these blocks. Moreover, the coefficients of this blockwise-symmetric expansion can be computed in time polynomial in $n$.

\subsection{Setting and statement}

Fix integers $m\ge 1$ and $k\ge 0$. Work over a field $\F$ of characteristic $0$
(e.g.\ $\R$), and in the quotient ring
\begin{align}
    \mathcal{R}\ :=\ \F[z_1,\dots,z_m]\big/\langle z_1^2-1,\dots,z_m^2-1\rangle .
\end{align}
That is, $z_i^2 = 1$ for all $i$.
Assume we are given $k$ (mod-$2$ independent) \emph{product relations}
\begin{equation}\label{eq:relations}
\prod_{i\in B_j} z_i \;=\; 1\qquad (j=1,\dots,k),
\end{equation}
where $B_j\subseteq[m]$ and ``independence'' means that no relation is the
(mod-$2$) product of the others. 
Let $P(x)=\sum_{r=0}^{\ell}p_r x^r\in\F[x]$
be a univariate polynomial of degree $\ell$, and let
\begin{align}
    S\ :=\ z_1+\cdots+z_m\in\mathcal{R}.
\end{align}
By Gaussian elimination over $\F_2$ we may eliminate $k$ variables and rewrite
every $z_i$ as a monomial in $n:=m-k$ \emph{independent} variables. After this
reduction we keep denoting the independent variables by $z_1,\dots,z_n$.

\blockexp*

The remainder of this Appendix provides a complete, self-contained proof.

\subsection{Block structure induced by the relations}

\begin{lemma}[Membership patterns and blocks]\label{lem:blocks}
After eliminating $k$ variables so that Eqn.~(\ref{eq:relations}) holds identically,
each remaining variable $z_i$ ($1\le i\le n$) acquires a $k$-bit \emph{membership pattern}
$\mathbf b(i)=(b_1(i),\dots,b_k(i))\in\{0,1\}^k$ recording in which relations
it appears. Group the indices $i$ by equal patterns; this yields
$r\le 2^k$ (possibly empty) groups, and discarding the empty ones produces blocks
$V_1,\dots,V_r$ that partition $[n]$. For each $j$, the monomial $p_{U_j}$ in
Eqn.~(\ref{eq:S-decomposition}) equals the product of all $z_i$ whose pattern has
$b_j(i)=1$, i.e.\ $U_j$ is a union of full blocks.
\end{lemma}

\begin{proof}
Row-reduce the system Eqn.~(\ref{eq:relations}) over $\F_2$, selecting $k$ pivot
variables and expressing them as monomials (products) in the remaining $n$ free
variables. Each original relation corresponds to a row $c_j\in\{0,1\}^{\,n}$
that indicates which free variables appear in the product on the right-hand side
after elimination. Define the membership pattern of $z_i$ to be the column
vector of its incidences across the $k$ rows; two variables share a pattern if
their columns coincide. The set of free variables whose $j$-th bit equals $1$
is therefore a union of whole pattern classes; multiplying those variables
exactly yields the monomial $p_{U_j}$. Finally $r\le 2^k$ because there are at
most $2^k$ distinct patterns.
\end{proof}

\begin{corollary}\label{cor:block-symmetry}
The element $S$ in Eqn.~(\ref{eq:S-decomposition}) is invariant under the natural
action of the direct product $G:=\mathfrak S_{V_1}\times\cdots\times \mathfrak
S_{V_r}$ (permutations within each block), hence so is $\calP(S)$.
\end{corollary}

\begin{proof}
Each $T_t$ is manifestly invariant under permutations of $V_t$, and each
$p_{U_j}$ is the product of all variables in a union of full blocks, hence
invariant under $G$. Polynomials in a $G$-invariant element are again
$G$-invariant.
\end{proof}

\subsection{A basis for blockwise symmetric polynomials}

Fix a block $V$ of size $n_V$. In the vector space $\mathcal{R}$, the set of
monomials $\{\prod_{i\in S}z_i:\ S\subseteq V\}$ is a basis. The subspace
$\mathcal{R}^{\mathfrak S_V}$ of polynomials invariant under permutations of $V$
is therefore spanned by the \emph{elementary symmetric sums}
\begin{align}
    e_a(V)\ :=\ \sum_{S\subseteq V,\ |S|=a}\ \prod_{i\in S}z_i
\qquad (0\le a\le n_V),
\end{align}
because these sums precisely aggregate monomials of equal cardinality.

\begin{lemma}[Single-block basis]\label{lem:single-block-basis}
For a fixed block $V$, the family $\{e_a(V)\}_{a=0}^{n_V}$ is a basis of
$\mathcal{R}^{\mathfrak S_V}$.
\end{lemma}

\begin{proof}
Spanning was noted above. For linear independence, observe that each monomial
$\prod_{i\in S}z_i$ appears in exactly one $e_a(V)$, namely with $a=|S|$, and
with coefficient $1$. Therefore a linear dependence among the $e_a(V)$ would
produce a nontrivial linear dependence among disjoint subsets of the monomial
basis, which is impossible.
\end{proof}

\begin{lemma}[Product basis]\label{lem:product-basis}
Let $V_1,\dots,V_r$ be the blocks from Lemma~\ref{lem:blocks}. Then the family
\begin{align}
    \big\{E_{\boldsymbol a}=\prod_{t=1}^r e_{a_t}(V_t)\ :\ 0\le a_t\le |V_t|\big\}
\end{align}
is a basis of the $G$-invariant subspace $\mathcal{R}^G$.
\end{lemma}

\begin{proof}
Since $\mathcal{R}\cong \bigotimes_{t=1}^r \mathcal{R}_t$ with
$\mathcal{R}_t=\F[V_t]/\langle z_i^2-1:i\in V_t\rangle$, and the group
$G$ acts independently on the tensor factors, we have
$\mathcal{R}^G\cong \bigotimes_{t=1}^r \mathcal{R}_t^{\mathfrak S_{V_t}}$.
By Lemma~\ref{lem:single-block-basis}, $\{e_{a_t}(V_t)\}_{a_t}$ is a basis of
$\mathcal{R}_t^{\mathfrak S_{V_t}}$. The tensor-product basis is therefore
exactly the family $\{E_{\boldsymbol a}\}$.
\end{proof}

\begin{proof}[Proof of Theorem~\ref{thm:blockwise_decomposition_theorem}(ii)]
By Corollary~\ref{cor:block-symmetry}, $\calP(S)\in\mathcal{R}^G$. 
By
Lemma~\ref{lem:product-basis}, $\{E_{\boldsymbol a}\}$ is a basis of
$\mathcal{R}^G$. 
Uniqueness and existence of the expansion
in Eqn.~(\ref{eq:block-expansion}) follow.
\end{proof}

\subsection{Structure identities and an invariant operator}

For a block $V$ of size $n_V$ write $e_a=e_a(V)$ for brevity. Two identities
hold in $\mathcal{R}$ (here and below we interpret $e_{-1}=e_{n_V+1}=0$):

\begin{lemma}[Sum-of-variables identity]\label{lem:sum-identity}
Let $T:=\sum_{i\in V} z_i$. Then
\begin{equation}\label{eq:Te}
T\,e_a\ =\ (a+1)\,e_{a+1} + (n_V-a+1)\,e_{a-1}\qquad(0\le a\le n_V).
\end{equation}
\end{lemma}

\begin{proof}
Write $e_a=\sum_{|J|=a} z_J$ with $z_J:=\prod_{j\in J}z_j$. Then
\begin{align}
    T e_a=\sum_{|J|=a}\ \sum_{i\in V} z_i z_J.
\end{align}
If $i\notin J$, then $z_i z_J=z_{J\cup\{i\}}$, a monomial of degree $a+1$.
Every monomial $z_K$ with $|K|=a+1$ arises exactly $(a+1)$ times in this manner
(from the pairs $(J,i)$ with $J=K\setminus\{i\}$). If $i\in J$, then $z_i^2=1$
implies $z_i z_J=z_{J\setminus\{i\}}$, a monomial of degree $a-1$.
Every monomial of degree $a-1$ arises exactly $n_V-a+1$ times (choose any
$i\in V\setminus(J\setminus\{i\})$ to form $J$). Summing over all $J$ yields
Eqn.~(\ref{eq:Te}).
\end{proof}

\begin{lemma}[Full-product reflection]\label{lem:reflection}
Let $p_V:=\prod_{i\in V} z_i$. Then
\begin{equation}\label{eq:reflection}
p_V\,e_a\ =\ e_{n_V-a}\qquad(0\le a\le n_V).
\end{equation}
\end{lemma}

\begin{proof}
We have
$p_V e_a=\sum_{|J|=a} \big(\prod_{i\in V}z_i\big)\big(\prod_{j\in J}z_j\big)
=\sum_{|J|=a} \prod_{i\in V\setminus J} z_i$, where we used $z_j^2=1$.
As $J$ ranges over all $a$-subsets of $V$, $V\setminus J$ ranges over all
$(n_V-a)$-subsets; hence the sum equals $e_{n_V-a}$.
\end{proof}

Return now to the decomposition in Eqn.~(\ref{eq:S-decomposition}).
For a multi-index $\boldsymbol a=(a_1,\dots,a_r)$ set $E_{\boldsymbol a}=
\prod_{t=1}^r e_{a_t}(V_t)$, and define a linear operator
$A:\mathrm{span}\{E_{\boldsymbol a}\}\to \mathrm{span}\{E_{\boldsymbol a}\}$
by left multiplication by $S$.

\begin{proposition}[Invariant tridiagonal-plus-reflection structure]\label{prop:A-action}
For all $\boldsymbol a$,
\begin{equation}\label{eq:A-action}
A\,E_{\boldsymbol a}
=\sum_{t=1}^r\Big((a_t+1)\,E_{\boldsymbol a+\mathbf e_t}
+(n_t-a_t+1)\,E_{\boldsymbol a-\mathbf e_t}\Big)
\ +\ \sum_{j=1}^k E_{R_j(\boldsymbol a)},
\end{equation}
where $n_t:=|V_t|$, $\mathbf e_t$ is the $t$-th unit vector, and the
\emph{reflection} $R_j$ acts by
\begin{align}
    (R_j(\boldsymbol a))_t\ :=\
\begin{cases}
n_t-a_t, & \text{if }V_t\subseteq U_j,\\
a_t, & \text{otherwise.}
\end{cases}
\end{align}
\end{proposition}

\begin{proof}
Eqn.~\eqref{eq:S-decomposition} writes $S$ as a sum of blocks $T_t$
plus blockwise full products $p_{U_j}$.
Multiplication by $T_t$ changes only the $t$-th factor; by
Lemma~\ref{lem:sum-identity} it contributes the first sum in
Eqn.~\eqref{eq:A-action}. 
Multiplication by $p_{U_j}$ multiplies by $p_{V_t}$
on each block $V_t\subseteq U_j$, which by Lemma~\ref{lem:reflection}
replaces $e_{a_t}(V_t)$ by $e_{n_t-a_t}(V_t)$; this yields the second sum.
\end{proof}

\subsection{Efficient coefficient computation via Horner's rule}

Let $\mathcal{V}:=\mathrm{span}\{E_{\boldsymbol a}\}$.
Represent elements of $\mathcal{V}$ by their coefficient arrays
$v=\{v_{\boldsymbol a}\}_{0\le a_t\le n_t}$.
Proposition~\ref{prop:A-action} immediately gives the component-wise action
of $A$ on arrays:
\begin{equation}\label{eq:A-on-arrays}
(Av)_{\boldsymbol a}
=\sum_{t=1}^r \Big(a_t\,v_{\boldsymbol a-\mathbf e_t}
+(n_t-a_t)\,v_{\boldsymbol a+\mathbf e_t}\Big)
\ +\ \sum_{j=1}^k v_{R_j(\boldsymbol a)},
\end{equation}
where out-of-range indices are interpreted as zero (this is just
Eqn.~\eqref{eq:A-action} re-indexed to collect the coefficient of $E_{\boldsymbol a}$).

Define $e^{(\mathbf 0)}\in\mathcal{V}$ by
$e^{(\mathbf 0)}_{\boldsymbol a}=1$ if $\boldsymbol a=\mathbf 0$, else $0$.
Since $P(S)=\sum_{r=0}^{\ell}p_r S^r=\sum_{r=0}^{\ell}p_r A^r e^{(\mathbf 0)}$,
Horner's rule yields the following recurrence for the coefficient array
$\gamma = \{\gamma_{\boldsymbol a}\}$ of $P(S)$ in the basis $\{E_{\boldsymbol a}\}$:
\begin{equation}\label{eq:horner}
v^{(\ell)}:=p_\ell e^{(\mathbf 0)},\qquad
v^{(q)}:=A\,v^{(q+1)}+p_q e^{(\mathbf 0)}\quad (q=\ell-1,\dots,0),\qquad
\gamma := v^{(0)}.
\end{equation}

\begin{lemma}[Complexity per step]\label{lem:complexity-A}
Let $D=\prod_{t=1}^r (n_t+1)$ be the size of the coefficient array.
Given $v$, the array $u:=Av$ defined by Eqn.~\eqref{eq:A-on-arrays} can be computed
using $O((r+k)D)$ arithmetic operations and $O(D)$ working memory.
\end{lemma}

\begin{proof}
Eqn.~ \eqref{eq:A-on-arrays} shows that to obtain $(Av)_{\boldsymbol a}$ for
each $\boldsymbol a$ one must:
(i) read the $2r$ neighbor entries
$v_{\boldsymbol a\pm \mathbf e_t}$ and form the $r$ weighted sums
$a_t v_{\boldsymbol a-\mathbf e_t}+(n_t-a_t) v_{\boldsymbol a+\mathbf e_t}$, and
(ii) read the $k$ reflected entries $v_{R_j(\boldsymbol a)}$ and add them.
This is $O(r+k)$ work per array entry, hence $O((r+k)D)$ in total.
Storing both $v$ and $u$ requires $O(D)$ memory.
\end{proof}

\begin{proof}[Proof of Theorem~\ref{thm:blockwise_decomposition_theorem}(iii)]
Initialize $v^{(\ell)}$ and iterate Eqn.~\eqref{eq:horner} down to $q=0$.
Each application of $A$ costs $O((r+k)D)$ by Lemma~\ref{lem:complexity-A},
and adding $p_q e^{(\mathbf 0)}$ costs $O(1)$. The total time is
$O((r+k)D\,\ell)$ and the peak memory is $O(D)$.
Finally, since $r\le 2^k$ and $k$ is constant, $r$ is a constant and
\begin{align}
    D=\prod_{t=1}^r (n_t+1)\ \le\ \Big(\tfrac{1}{r}\sum_{t=1}^r (n_t+1)\Big)^{\!r}
=\Big(\tfrac{n}{r}+1\Big)^{\!r}\ =\ n^{O(1)}
\end{align}
by AM--GM. Thus the algorithm runs in time polynomial in $n=m-k$ and $\ell$.
\end{proof}

\subsection{Completion of the proof}

\begin{proof}[Proof of Theorem~\ref{thm:blockwise_decomposition_theorem}(i)]
This is Lemma~\ref{lem:blocks}.
\end{proof}

\begin{proof}[Proof of Theorem~\ref{thm:blockwise_decomposition_theorem}(ii)]
This is Corollary~\ref{cor:block-symmetry} and Lemma~\ref{lem:product-basis}.
\end{proof}

\begin{proof}[Proof of Theorem~\ref{thm:blockwise_decomposition_theorem}(iii)]
This is the Horner scheme from Eqn.~\eqref{eq:horner} together with
Proposition~\ref{prop:A-action} and Lemma~\ref{lem:complexity-A}.
\end{proof}

We end with a remark on possible extensions.

\begin{remark}[Signs \texorpdfstring{$\prod z_i=-1$}{product = -1}]
If some relations equal $-1$ rather than $1$, the decomposition
Eqn.~\eqref{eq:S-decomposition} holds with $p_{U_j}$ multiplied by $-1$.
The operator formula Eqn.~\eqref{eq:A-action} then has $\pm E_{R_j(\boldsymbol a)}$
with the appropriate signs. All proofs and complexities are unchanged.
\end{remark}

\section{Efficient computation of combinatorial coefficients}
\label{app:combinatorics}

Consider the polynomial
\begin{equation}
    \mathcal{P}_{\ell,m}(z_1,\ldots,z_m) = (z_1 + z_2 + \ldots + z_m)^\ell,
\end{equation}
where, for each $i \in \{1,\ldots,m\}$, $z_i$ is a formal variable satisfying $z_i^2 = 1$. We can re-express this in the form
\begin{equation}
    \mathcal{P}_{\ell,m}(z_1,\ldots,z_m) = \sum_{r=0}^\ell a(m,\ell,r) P^{(r)}(z_1,\ldots,z_r),
\end{equation}
where $P^{(r)}$ is the elementary symmetric polynomial
\begin{equation}
    P^{(r)}(z_1,\ldots,z_m) = \sum_{\substack{\mathbf{p} \in \mathbb{F}_2^m \\ |\mathbf{p}|=r}} \prod_{j=1}^m z_j^{p_j}. 
\end{equation}
The coefficient $a(m,\ell,r)$ with $r=\ell$ arises from the pure cross terms in which no variable is repeated. The coefficients with lower $r$ arise from repeated factors that square to the identity, and thus lead to a ``trickle down'' into terms of lower degree. 
In this Appendix, we show that the coefficients $a(m,\ell,r)$ can be computed classically in time $\mathrm{poly}(m,\ell,r)$.
Therefore, the symmetric polynomial expansion in Theorem~\ref{thm:commuting_symmetric_polynomial_expansion} can be efficiently implemented.

First, consider $a(m,\ell,0)$ in the case that $\ell$ is even. 
This is the number of terms in which all monomials appear to even powers, hence yielding the trivial monomial, which is represented by the all-zeros bitstring. 
Direct counting is infeasible as there are exponentially many terms. 
Instead, consider the following expression.
\begin{equation}
    \frac{1}{2^m} \sum_{z_1 \in \{+1,-1\}} \ldots \sum_{z_m \in \{+1,-1\}} \left(z_1 + \ldots + z_m \right)^\ell
\end{equation}
For each $j \in \{1,\ldots,m\}$, the summation $\frac{1}{2} \sum_{z_j \in \{+1,-1\}}$ eliminates all terms that are odd with respect to $z_j$ and leaves all other terms unchanged. Thus, the above expression is equal to $a(m,\ell,0)$. We can re-express this as
\begin{equation}
    \label{eq:asum}
    a(m,\ell,0) = \frac{1}{2^m} \sum_{n=0}^m \binom{m}{p} \left( m - 2n \right)^\ell,
\end{equation}
where $n$ can be interpreted as the number of variables amongst $\{z_1,\ldots,z_m\}$ that are equal to -1. This is recognizable as a shifted moment of the binomial distribution. 
Such moments to not have simple closed form expressions. However, since the summation in Eqn.~(\ref{eq:asum}) has only $m$ terms it can be classically evaluated in polynomial time, which is enough for our purposes.

The other combinatorial coefficients can be computed by similar means. Specifically, $a(m,\ell,r)$ is equal to the sum of the coefficients on all the monomials in $(z_1+ \ldots + z_m)^\ell$ in which $z_1,\ldots,z_r$ all appear with odd powers and $z_{r+1},\ldots,z_m$ all appear with even powers. This sum of coefficients can be obtained by antisymmetrizing over variables $z_1,\ldots,z_r$ and symmetrizing over variables $z_{r+1},\ldots,z_m$. That is,
\begin{equation}
    a(m,\ell,r) = \frac{1}{2^m} \sum_{z_1,\ldots,z_m \in \{-1,+1\}} z_1 \ldots z_r \left( z_1 + \ldots + z_m \right)^\ell.
\end{equation}
Let $n_1$ be the number of variables amongst $\{z_1,\ldots,z_r\}$ that are equal to $-1$ and let $n_2$ be the number of variables amongst $\{z_{r+1}, \ldots, z_m \}$ that are equal to $-1$. Then, the above sum can be rewritten as
\begin{equation}
    a(m,\ell,r) = \frac{1}{2^m} \sum_{n_1 = 0}^r \binom{r}{n_1} (-1)^{n_1} \sum_{n_2 = 0}^{m-r} \binom{m-r}{n_2} \left( m - 2n_1 - 2n_2 \right)^\ell.
\end{equation}
This sum can be evaluated in polynomial time, as it has only $O(m^2)$ terms.

\end{document}